\documentclass[aip,jmp,reprint,numerical]{revtex4-1}
\usepackage{graphicx}
\usepackage{latexsym}

\usepackage{amsmath}
\usepackage{amssymb}

\usepackage{amsthm}

\newtheorem{definition}{Definition}[section]
\newtheorem{theorem}{Theorem}[section]
\newtheorem{corollary}{Corollary}[theorem]
\newtheorem{lemma}{Lemma}[section]

\begin{document}

\title{Graph Approach to Quantum Systems}

\author{Mladen Pavi\v ci\'c}
\homepage{http://m3k.grad.hr/pavicic}
\affiliation{Institute for Theoretical Atomic, Molecular,
and Optical Physics at Physics Department at
Harvard University and Harvard-Smithsonian Center for
Astrophysics, Cambridge, MA 02138 USA \\
and Chair of Physics, Faculty of Civil Engineering,
University of Zagreb, 10000 Zagreb, Croatia.}

\author{Brendan D.~McKay}
\homepage{http://cs.anu.edu.au/$\sim$bdm}
\affiliation{School of Computer Science,
Australian National University,
Canberra, ACT, 0200, Australia.}

\author{Norman D.~Megill}
\homepage{http://www.metamath.org}
\affiliation{Boston Information Group, 19 Locke Ln., Lexington,
MA 02420, U.~S.~A.}

\author{Kre\v simir Fresl}
\homepage{http://www.grad.hr/nastava/gs}
\affiliation{Department of Technical Mechanics,
Faculty of Civil Engineering,
University of Zagreb, 10000 Zagreb, Croatia.}

\begin{abstract}

Using a graph approach to quantum systems, we show that
descriptions of 3-dim Kochen-Specker (KS) setups as well
as descriptions of 3-dim spin systems by means of Greechie
diagrams (a kind of lattice) that we find in the literature
are wrong. Correct lattices generated by McKay-Megill-Pavicic
(MMP) hypergraphs and Hilbert subspace equations are given. 
To enable future exhaustive
generation of 3-dim KS setups by means of our recently found
{\em stripping technique}, bipartite graph generation is used to
provide us with lattices with equal numbers of elements and blocks
(orthogonal triples of elements)---up to 41 of them. We obtain
several new results on such lattices and hypergraphs, in
particular on properties such as superposition and orthoraguesian
equations.
\end{abstract}



\maketitle

\section{\label{sec:intro}Introduction}

We make use of hypergraphs (defined in Sec.~\ref{sec:represent}) and
bipartite graphs (defined in Sec.~\ref{sec:one-state}) to describe
large 3-dim quantum setups.  One way to describe a quantum system in
Hilbert space is through the use of lattices, specifically Hilbert
lattices (Def.~\ref{def:hl}), and our approach is based on a
correspondence between graphs and lattices.

Many authors have tried to justify empirically a mathematically
well-proved ortho-isomorphism between a Hilbert lattice
and the lattice of subspaces of an
infinite-dimensional Hilbert space, which has been worked out by
over the last 60 years.\cite{beltr-cass-book,holl95}
The finite-dimensional case was elaborated even earlier by
G.~Birkhoff and J.~von Neumann.\cite{birk-v-neum} The results 
were crowned by the result of Maria Pia Sol\`er\cite{soler}
that the field (e.g., complex numbers) over which the Hilbert space
can be defined follows from the Hilbert lattice conditions.

Yet, a satisfactory empirical justification has not been achieved.
First steps have been attempted with a description of spin-1, i.e.,
3-dim systems. Several
authors\cite{shimony,hultgren,svozil-tkadlec,svozil-book-ql,tka98,tkadlec,tkadlec-2,smith03,foulis99,pmmm04b,pavicic-rmp09}
have obtained a number of results in applications of the so-called
Greechie diagrams (see Subsec.~\ref{subsec:greechie}) to spin
systems. For instance, a correspondence found between
orthomodular lattices and MMP hypergraphs (see Sec.~\ref{sec:represent}, 
Def.~\ref{def:mmp})
enabled an exhaustive generation of all 3-dim Kochen-Specker (KS)
sets with up to 24 vectors.\cite{pmm-2-09}

On the other hand, many results on equations holding in Hilbert
lattices (see  Sec.~\ref{subsec:eqns})  have recently been
obtained.\cite{mpoa99,mayet06,pm-ql-l-hql2,mayet06-hql2,mp-alg-hilb-eq-09}
An immediate idea was to verify these equations on the sets for
which an experimental setups was designed---KS sets. To our
surprise it turned out that the standard KS setups described by
Greechie diagrams do not allow a verification of these equations.
Moreover known KS systems described by Greechie diagrams do not
pass even the property of modularity which any spin lattice should
pass. Hence, something was missing in the known description of those sets.

A missing link between empirical quantum measurements
and its lattice structure was a proper description of a correspondence
between the standard quantum measurements, which use Hilbert space
vectors and states, and Hilbert lattices, which make use of Hilbert space
subspaces that contain these vectors and/or are spanned by them.
What hampered a search for such a correspondence was a too
narrow focus on orthogonality via lattices represented by Greechie
diagrams (Def.~\ref{def:greechie}).

As we show in Subsec.~\ref{subsec:greechie} Greechie diagrams
cannot serve the purpose because they in general turn out not to be
subalgebras of a Hilbert lattice (Theorem \ref{th:notsubalgebra}).

We give two examples which were most elaborated in the literature:
empirical reconstruction of quantum mechanics via lattice theory
and a description of Kochen-Specker's setups via lattice theory.
The examples show how the application of the Greechie diagrams
lead these elaborations to a dead end.

As for the empirical reconstruction,  B.~O.~Hultgren, III
and A.~Shimony used Greechie diagrams in their
detailed attempt to build up a Hilbert lattice of a realistic
quantum system for a 3-dim spin-1 system passing through
Stern-Gerlach filters.\cite{shimony,hultgren}
They did not succeed in building a Hilbert lattice because
the Greechie diagrams, as we show below (Th.~\ref{th:notsubalgebra}), are not
subalgebras of a Hilbert lattice. They failed to obtain some features
they  thought they should have obtained and they obtained some
features they thought they should not have obtained. As for
the former features, e.g., superposition, we show
that they do not cause necessarily problems (see the remark after
Lemma~\ref{lem:subalgebra}). As for the latter
features, it has been shown that their appearance was due
to the fact that they did not take into account
both electric and magnetic fields.\cite{anti-shimony}
However, even if Hultgren and Shimony had used them
they could have only repaired some faulty Greechie diagrams.
In particular, they could have patched the missing links in their
Fig.~3 (dashed lines) and with them their lattice would read: 
{\tt 123,456,789,ABC,58B} (using MMP hypergraph encoding,
described below).

As for the KS setup, S.~Kochen and E.~P.~Specker\cite{koch-speck}
in their proof used a {\em partial Boolean
algebra} (PBA), which is a very general class of algebras.
The closed subspaces of a Hilbert space form a particular,
specialised PBA.  However, conditions that make  PBA
isomorphic to a lattice of Hilbert space subspaces have not been
discovered, although steps in that direction have been taken by
D.~Smith.\cite{smith99,smith03,smith04}
The equivalence of PBA and atomic ortholattices was proved by
I.~Pitowsky in 1982.\cite{pitowski} Apparently misled by this
equivalence, some authors have represented KS setups by means
of Greechie diagrams in a series of
publications.\cite{svozil-tkadlec,svozil-book-ql,tka98,tkadlec,tkadlec-2}
In Sec.~\ref{sec:represent}, we show that KS setups
cannot be described by means of Greechie diagrams because
Greechie diagrams are not subalgebras of a Hilbert lattice.

Now, in Sec.~\ref{sec:represent} we show that both a lattice
reconstruction of quantum mechanics and a lattice
description of KS setups must take nonorthogonal subsets into
account. They are required by the conditions and equations that
must hold in every Hilbert space.\footnote{Pitowski
(Ref.~\onlinecite{pitowski}, p.~392) says:  ``Kochen and Specker (1967)
constructed a finitely generated sublattice L' of L for which
no truth function exists,'' but neither he nor Kochen and
Specker gave a blueprint for such a lattice, i.e., we do
not have their constructive definition.}
This is the reason why KS setups cannot be described by means of
Greechie diagrams, as we prove for all known spin-1 KS
setups, notably Kochen-Specker's\cite{koch-speck}, Peres'\cite{peres},
Kernaghan's\cite{kern}, Bub's\cite{bub}, and
Conway-Kochen's\cite{bub}.

We also find a way to obtain lattices that we {\em can}
use to describe a quantum setup to any desired degree of
accuracy. They make use of subspaces that contain
non-orthogonal vectors and/or are spanned by them.
The subspaces that appear in them are filtered by the
aforementioned conditions and equations that must hold in
every Hilbert space. We call such lattices
MMPLs (see Def.~\ref{def:mmpl} and Fig.~\ref{fig:bub-proof}).

However, our programs written for a generation of arbitrary
MMP hypergraphs that can be used for a construction
of MMPLs with more than 30 vectors take too much time.
Therefore, we consider lattices that have some of the
properties MMPLs require and lack some others, with the
idea---which turns out to be rewarding---of getting
lattices with more than 40 vectors that can be obtained
faster and that can in turn give all interesting MMPLs by
means of different very fast algorithms and programs.
For instance, to obtain all 4-dim KS sets with 18 through
24 vectors requires several months on a cluster with 500
3GHz CPUs, while in Ref.~\onlinecite{pmm-2-09} we found an
algorithm and a program to obtain them all from a
single KS set with 24 vectors in less than 10 min on
a single PC. This 24 vector KS set also belongs to the
aforementioned class of lattices that have ``some of the
properties MMPLs require and lack some others.''
Vectors correspond to atoms in lattices and to vertices in
MMP hypergraphs, and tetrads correspond to blocks in
lattices and edges in MMP hypergraphs. MMP hypergraphs are
defined in Ref.~\onlinecite{pmmm04b} and in Sec.~\ref{sec:represent}, 
Def.~\ref{def:mmp}.

The aforementioned ``10 min'' method we call a
{\em stripping technique}.\cite{pmm-2-09}
It consists in stripping blocks off of a single
initial KS set with 24 vertices (vectors) and 24
edges (tetrads) until we reach the  smallest such
sets---called {\em critical KS sets}---in the sense
that any of them would cease to be a KS set if we
stripped any further blocks away. The technique
provided us with all 1232 KS subsets with vector component
values from \{-1,0,1\} contained in the 24-24 class
of KS \{-1,0,1\} set.

We also applied the same technique to 60-60 KS sets
that we obtained from a 60-75 set and generated a huge
number of critical sets (for 60-65 through 60-75 we
rigorously verified that no critical set exists and
for 60-61 through 60-64 we confirmed that statistically
with a high confidence). All the KS sets and critical
sets we generated in this way form a new KS class (we
call it ``60-75 KS class'') which is
disjoint from the 24-24 class.\cite{mp-nm-pk-10}
The smallest critical set from this class is a
13-26 KS set shown in Fig.~1 of Ref.~\onlinecite{mp-nm-pk-10}.

In the above generation of 4-dim KS sets by the stripping technique
we were fortunate to find covering KS sets with the same number
of vertices and edges (24-24, 60-60). For 3-dim KS sets no such
covering set with the the same number of vertices and edges is
known so in this paper we pave the road of its generation by
means of bipartite graphs (see Secs.~\ref{sec:one-state} and
\ref{sec:one-state-p}).

As we report in Sec.~\ref{sec:one-state}, such a generation of
bipartite graphs is still computationally too demanding. We
previously considered 3-dim systems with equal number of atoms
(vertices) and blocks (edges) with up to 38 atoms and
blocks.\cite{pavicic-rmp09} Now we use much faster algorithms
and programs and are able to reach 41 atoms and blocks. This is still
not enough for a realistic system, but we obtain
several important properties of such classes of lattices that might
help us to obtain even better algorithms and reach the 50 atoms
required for generation of realistic KS setups with the help of
the {\em stripping technique}.

The results we invoke and make use of are well-known in
lattice theory. They have not been reformulated
in Hilbert space theory itself, so, we present all our results in
the lattice theory, and only when it would really help the reader
to see what a Hilbert-space version of particular properties and axioms
would look like, do we formulate some result directly in the
Hilbert space parlance as, e.g., in Theorems \ref{th:hs-ssnoa}
and \ref{th:hs-noa}. Hence for the reader who is not too
familiar with the lattice theory, we first introduce and characterise
its basic notions in Sec.~\ref{sec:ortho}, and here we give a
general framework in which we shall make use of the
lattice theory.

A spin state of a system is assumed to be repeatedly prepared,
manipulated, and/or filtered by a device. The directions of vectors
of the spin projections coincide with the orientations of the device.
Hilbert space subspaces that contain these vectors form lattices.
To distinguish between device orientations and spin orientations
we use the term {\em experimental setup} to mean a description of
the devices and their fields. We use the term {\em formalised setup}
to mean a theoretical description of the quantum systems.

We start with a very general class of lattices---orthomodular
lattices (OMLs) (see Def.~\ref{def:OML-MOL-BA}). Elements of spin-1
OMLs correspond to subspaces (1-dim rays and 2-dim planes) spanned by
Hilbert space vectors which must satisfy two classes of conditions:

(1)  {\em Equations}, e.g., the orthoarguesian
and Godowski equations (see Table~\ref{tab:eqns} for a summary
of these and other equations mentioned);

(2) {\em Quantified expressions}, e.g.,
the superposition principle [Def.~\ref{def:hl}(3); Eq.~(\ref{eq:superp})].

They  are essential for understanding the
ramification of all quantum setups:

(1) {\em Equations} that fail in a subalgebra
of a lattice will also fail in the lattice
(see Lemma~\ref{lem:subalgebra} below).  So no experimental setup
for which quantum mechanical equations cannot have a solution can
be used for measuring properties of a quantum system. Such setups
are non-quantum setups;

(2) {\em Quantified expressions} that fail in a
subalgebra of a lattice may, however, pass in the lattice
(see the remark after Lemma~\ref{lem:subalgebra} below).
Smaller setups, in which e.g.~superposition cannot be measured, are
``sub-setups'' of setups in which superposition is possible.

{\em Quantum setups} and {\em quantum lattices} refer to systems
whose OMLs are subalgebras of a Hilbert lattice. {\em Semi-quantum
lattices} refer to systems whose OMLs are not subalgebras of a Hilbert
lattice. Examples of the former are proper KS lattices
in the sense of being subalgebras of a Hilbert lattice.

Semi-quantum lattices with equal number of atoms and blocks we
consider are atomic lattices. They admit real-valued and
vector states, satisfy superposition, and yet violate, e.g.,
orthoarguesian equations. To deal with them we can use
Greechie diagrams because we consider lattices that consist
of concatenated orthogonal triples and are not subalgebras
of a Hilbert lattice.

To generate semi-quantum lattices we proceed as follows.
We first use algorithms that exhaustively generate cubic
bipartite graphs. We then show that they are equivalent
to MMP hypergraphs which in turn
correspond to OMLs with equal numbers of atoms and blocks.
We generate OMLs with up to 41 of atoms and blocks, and prove
that they all have the above features. The obtained OMLs narrow
down the non-quantum classes of OMLs and might enable us
to generate quantum classes of OMLs of high
complexity and KS setups.   They also enable us to obtain
several new results in Hilbert lattice theory that rely on the
features that the generated OMLs possess. In Sec.~\ref{sec:one-state},
we analyse the properties of the OMLs obtained in Sec.\
\ref{sec:one-state-p}, and provide a new type of graphical
representation for them in Sec.~\ref{sec:figs}. We discuss the
obtained results in Sec.~\ref{sec:conclusions}.

The ``negative results'' that we consider in this paper (classes
of lattices that do {\em not} pass particular equations) we have
recently used as a tool for generating other equations
\cite{mp-alg-hilb-eq-09} and, in the case of the aforementioned
4-dim KS sets, for generating new KS sets.

Our results also provide us with novel algorithms and results
in the theory of bipartite graphs and hypergraphs.
Lattices that do not admit strong sets of states serve
as inputs to algorithms for finding new Hilbert lattice equations,
and lattices that admit just one state serve for establishing
new lattice features and theorems.\cite{shultz74,navara08}

Bipartite graphs have recently been studied extensively in the
field of quantum information.
A bipartite entanglement of the states constructed
from the algebra of a finite group with a bilocal
representation ($G$) acting on a separable reference state
has been studied in Ref.~\onlinecite{hamma}.
If $G$ is a group of spin flips acting on a set of qubits, these
states are locally equivalent to bipartite (two-colorable) graph
states and they include GHZ, CSS, cluster states,  etc.
Equivalence of CSS states (of which GHZ states are a special case)
and bipartite graph states has been shown in Ref.~\onlinecite{chen}.

Graph states form class of multipartite entangled states
associated with combinatorial graphs (see, e.g.,
Refs.~\onlinecite{heinbriegel04} and \onlinecite{cosentino})
and have applications in diverse areas of quantum
information processing, such as quantum error correction
and the one-way model.

On the other hand, bipartite graphs have been shown to have
an important application for quantum search and related
quantum walks, span-programs, and search algorithms such
as Grover's.\cite{carneiro-knight,reichard-arXiv}

\section{\label{sec:ortho}Preliminary Definitions and Theorems
and the Semi-Quantum Lattices}

This section covers most definitions and background material.
It is organized as follows.
\begin{itemize}
\item Hilbert lattices (Subsection~\ref{subsec:hl})
\item Overview of equations holding in Hilbert lattices (Subsection~\ref{subsec:eqns})
\item States (Subsection~\ref{subsec:states})
\item Vector-valued states (Subsection~\ref{subsec:vectorstates})
\item Superposition (Subsection~\ref{subsec:superpos})
\item Orthoarguesian equations (Subsection~\ref{subsec:noa})
\item Greechie diagrams (Subsection~\ref{subsec:greechie})
\item Semi-quantum lattices (Subsection~\ref{subsec:semi})
\end{itemize}

\subsection{Hilbert lattices}\label{subsec:hl}

The closed subspaces of a Hilbert space $\mathcal{H}$ form an algebra
called $\mathcal{C}(\mathcal{H})$, which is a member of the class of
orthomodular lattices (OML).  An OML, in turn, is a member of a more
general class called OL (ortholattices).  We will first define OLs,
OMLs, and related structures, then we will describe how the closed
subspaces of Hilbert space form a member of (some) of these
classes of structures.

We define OL as follows, along with auxiliary constants $0$ and $1$, an
ordering relation, and an implication operation.  The binary operations
$\cup$ and $\cap$ are called {\em join} and {\em meet} respectively, and
the unary operation $'$ is called {\em orthocomplementation}.  Recall
that an {\em algebra} is an $n$-tuple consisting of a base set and $n-1$
operations on that base set.

\begin{definition}\label{def:ourOL}
An {\em ortholattice}, {\rm OL\/}, is an algebra
$\langle{\mathcal{OL}}_0,',\cup,\cap\rangle$
such that the following conditions are satisfied for any
$a,b,c\in \,{\mathcal{OL}}_0$ {\rm \cite{mpqo02}}:
$a\cup b\>=\>b\cup a$, $(a\cup b)\cup c\>=\>a\cup (b\cup c)$,
$a''\>=\>a$, $a\cup (b\cup b\,')\>=\>b\cup b\,'$,
$a\cup (a\cap b)\>=\>a$, and $a\cap b\>=\>(a'\cup b\,')'$.
In addition, since $a\cup a'=b\cup b\,'$ for any $a,b\in
\,{\mathcal{OL}}_0$, we define the {\em greatest element of
the lattice} {\rm (1)} and the {\em least element of the
lattice} {\rm (0)}, $\textstyle{1}\,{\buildrel\rm def\over=}a\cup a'$
and \ $\rm\textstyle{0}\,{\buildrel\rm def\over =}a\cap a'$, respectively
and the {\rm ordering relation ($\le$) on the lattice}:
$a\le b\ \quad{\buildrel\rm def\over\Longleftrightarrow}\quad\ a\cap b=a
\quad\Longleftrightarrow\quad a\cup b=b$. Quantum (Sasaki) implication
 is defined as $a\to b=a'\cup(a\cap b)$.
\end{definition}

When we say a lattice is an OL (or an OML, etc.) we mean that the lattice
is a member of the class OL (OML, etc).

By adding an additional condition, we can restrict the class OL to
become the successively smaller (less general) classes OML, MOL, and BA
as follows.

\begin{definition}\label{def:OML-MOL-BA}
An {\em ortholattice} {\rm (OL)} in which
\begin{align}
b\le a\ \&\ c\le a'\ \Rightarrow\ &a\cap(b\cup c)=(a\cap
b)\cup(a\cap c),\label{eq:oml} \\
b\le a\ \Rightarrow\ &a\cap(b\cup
c)=(a\cap b)\cup(a\cap c),\label{eq:mod}\\
\text{\rm or }\qquad&a\cap(b\cup c)=(a\cap b)\cup(a\cap c)
\end{align}
holds, is an orthomodular lattice {\rm (OML)}, modular
ortholattice {\rm (MOL)}, or Boolean algebra {\rm (BA)},
respectively.
\end{definition}

Our primary interest is in the subclass of OML called HL (Hilbert
lattices).

\begin{definition}\label{def:hl}\footnote{For additional
definitions of the terms used in this section see
Refs.~\onlinecite{beltr-cass-book,holl95,kalmb83}}
An orthomodular lattice that satisfies the following
con\-di\-tions is a {\em Hilbert lattice} ({\rm HL}).
\begin{enumerate}
\item {\em Completeness:\/}
The meet and join of any subset of
an {\rm HL} exist.
\item {\em Atomicity:\/}\label{atomicity}
Every non-zero element in an {\rm HL} is greater
than or equal to an atom. (An {\em atom} $a$ is a non-zero lattice element
with $0< b\le a$ only if $b=a$.)
\item {\em Superposition principle:\/}
(The atom $c$
is a {\em superposition} of the atoms $a$ and $b$ if
$c\ne a$, $c\ne b$, and $c\le a\cup b$.)
\begin{description}
\item[{\rm (a)}] Given two different atoms $a$ and $b$, there is at least
one other atom $c$, $c\ne a$ and $c\ne b$, that is a superposition
of $a$ and $b$.
\item[{\rm (b)}] If the atom $c$ is a superposition of  distinct atoms
$a$ and $b$, then atom $a$ is a superposition of atoms $b$ and $c$.
\end{description}
\item {\em Minimum height:\/} The lattice contains at least
two elements $a,b$ satisfying: $0<a<b<1$.
\end{enumerate}
\end{definition}

These conditions imply an infinite number of atoms in HL, as
shown by Ivert and Sj{\"o}din.\cite{ivertsj}

With suitably defined operations, the closed set of subspaces of
a Hilbert space, $\mathcal{C}(\mathcal{H})$, can be shown to be
a Hilbert lattice (a member of HL).   The meet operation $a\cap b$
corresponds to the
set intersection ${\mathcal H}_a\bigcap{\mathcal H}_b$
of subspaces ${\mathcal
H}_a,{\mathcal H}_b$ of Hilbert space ${\mathcal H}$; the ordering relation
$a\le b$ corresponds to ${\mathcal H}_a\subseteq{\mathcal H}_b$; the
join operation $a\cup b$ corresponds to the smallest closed subspace of
$\mathcal H$ containing the set union
${\mathcal H}_a\bigcup{\mathcal H}_b$; and
the orthocomplementation operation $a'$ corresponds
to ${\mathcal H}_a^\perp$, the set of vectors orthogonal to all vectors in
${\mathcal H}_a$. Within Hilbert space there is also an operation which
has no parallel in the Hilbert lattice: the sum of two subspaces
${\mathcal H}_a+{\mathcal H}_b$, which is defined as the set of sums of vectors
from ${\mathcal H}_a$ and ${\mathcal H}_b$. We also have
${\mathcal H}_a+{\mathcal H}_a^\perp={\mathcal H}$, i.e.~the
subspace that equals the whole of Hilbert space itself. One can define
all the lattice operations on a Hilbert space itself following the above
definitions (${\mathcal H}_a\cap{\mathcal H}_b={\mathcal H}_a\bigcap{\mathcal H}_b$,
etc.). Thus we have
${\mathcal H}_a\cup{\mathcal H}_b=\overline{{\mathcal H}_a+{\mathcal H}_b}=
({\mathcal H}_a+{\mathcal H}_b)^{\perp\perp}=
({\mathcal H}_a^\perp\bigcap{\mathcal
H}_b^\perp)^\perp$,\cite[p.~175]{isham} where
$\overline{{\mathcal H}_c}$ is the closure of ${\mathcal H}_c$, and therefore
${\mathcal H}_a+{\mathcal H}_b\subseteq{\mathcal H}_a\cup{\mathcal H}_b$.
When ${\mathcal H}$ is finite-dimensional or when
the closed subspaces ${\mathcal H}_a$ and  ${\mathcal H}_b$ are orthogonal
to each other then ${\mathcal H}_a+{\mathcal H}_b=
{\mathcal H}_a\cup{\mathcal H}_b$.~(Refs.~\onlinecite[pp.~21-29]{halmos}, 
\onlinecite[pp.~66,67]{kalmb83}, \onlinecite[pp.~8-16]{mittelstaedt-book})

Using these operations, it is straightforward to verify that closed
subspaces $\mathcal{C}(\mathcal{H})$ of a finite- or infinite-dimensional
 Hilbert space form an OML (Ref.~\onlinecite[pp.~66,67]{kalmb83}) and
more specifically an HL (Ref.~\onlinecite[pp.~105--108,166,167]{beltr-cass-book}).
[In the case of a finite Hilbert space, $\mathcal{C}(\mathcal{H})$ is
also an MOL.~(Ref.~\onlinecite[p.~107]{beltr-cass-book})]
Specifically, we have the following theorem.

\begin{theorem}
Let $\mathcal H$ be a finite- or infinite-dimensional Hilbert space
over a field $\mathcal K$ and let
\begin{align}
{\mathcal C}({\mathcal H})\ {\buildrel\rm def\over =}\ \{ {\mathcal X}\
\subseteq {\mathcal H}\ | \>{\mathcal X}^{\perp\perp}={\mathcal X}\}
\end{align}
be the set of all closed subspaces of $\mathcal H$.
Then ${\mathcal C}({\mathcal H})$ is a Hilbert lattice
relative to:
\begin{align}
a\cap b={\mathcal X}_a\cap {\mathcal X}_b
\ \ \text{\rm and}\ \ a\cup b
 =({\mathcal X}_a+{\mathcal X}_b)^{\perp\perp}.\qquad
\end{align}
\end{theorem}

A more difficult problem is to determine, given an HL, how much of
Hilbert space can be reconstructed from it.  An {\em isomorphism} is a
bijection between two lattices that preserves the lattice ordering (or
equivalently the meet and join operations).  An {\em ortho-isomorphism}
is an isomorphism that also preserves the orthocomplement operation.
One can prove the following representation
theorem.\cite{mackey,maclaren,varad}

\begin{theorem}\label{th:repr}
For every Hilbert lattice
({\rm HL}), there exists a field $\mathcal K$ and a Hilbert space
$\mathcal H$ over $\mathcal K$ such that the set of closed
subspaces of the Hilbert space, ${\mathcal C}({\mathcal H})$, is
ortho-isomorphic
\index{ortho-isomorphism}%
to {\rm HL}.
(Note that multiplication is not necessarily commutative in this field,
which some authors call a ``division ring'' or ``skew field.'')
\end{theorem}

In order to determine the field over which the Hilbert space
in Theorem \ref{th:repr} is defined, we make use of a
theorem proved by Maria Pia Sol{\`e}r.~\cite{soler,holl95} 
First, we need a definition.

\begin{definition}\label{def:har-conj}
Let $p$ and $q$ be orthogonal atoms in a Hilbert lattice and $c$
be an atom different from $p$ and $q$ such that $c\le p\cup q$.
Let $x$ be any atom such that
$x \nleq p\cup q$.  Let $y$ an atom different from $x$ and $p$
such that $y\le x\cup p$.
Define $d_1=(c\cup y)\cap(q\cup x)$ and $d_2=(p\cup d_1)\cap (q\cup y)$.
Then $(x\cup d_2)\cap(p\cup q)$ is the (unique) {\em harmonic conjugate} of $c$
with respect to $p$ and $q$.
\end{definition}

Now we can state the following application of Sol{\`e}r's theorem
to an HL lattice.\cite[Th.~4.1]{holl95}
\begin{theorem}\label{th:sol}
The Hilbert space $\mathcal H$ from Theorem \ref{th:repr} is
an infinite-dimen\-sional Hilbert space defined over
a real,
\index{field!of real numbers}%
\index{real numbers!field of}%
complex, or
\index{field!of complex numbers}%
\index{complex numbers!field of}%
quaternion (skew) field
\index{field!of quaternions}%
\index{quaternions!(skew) field of}%
if the following conditions are met:
\begin{itemize}
\item {\em Infinite orthogonality:} The {\rm HL} contains a countably
infinite sequence of orthogonal atoms $p_i, i=1,2,\ldots$
\item {\em Harmonic conjugate condition:} The {\rm HL}
\index{harmonic conjugate}%
contains a corresponding sequence of atoms
$c_i\le p_i\cup p_{i+1}$ such that the harmonic conjugate of
$c_i$ with respect to $p_i,p_{i+1}$ equals
$c'_i\cap(p_i \cup p_{i+1})$.
\end{itemize}
\end{theorem}

Thus we do arrive at a full Hilbert space, but as we can see the
axioms for the Hilbert lattices that we used for this purpose are
rather involved. This is because in the past, the axioms were
simply read off from the Hilbert space structure and were
formulated as first-order quantified statements that cannot
be implemented into a quantum computer. As opposed to this, the
equations describing properties of Hilbert lattices and elaborated
on in Defs.~\ref{def:strong}, \ref{def:noa}, Eqs.~(\ref{eq:E-3}),
(\ref {eq:E-4}), (\ref{eq:oa3}), and Th.~\ref{th:noa} are directly
applicable to experimental setups and that is the reason why the
results we obtain in this paper and in the recent previous
paper of ours have not been conjectured previously.

\subsection{Overview of equations holding in Hilbert lattices}
  \label{subsec:eqns}

The families of lattices OL, OML, MOL, and BA are completely
characterized by identities, i.e., equational conditions.  Such
families are called {\em equational varieties}.  Equations, as
opposed to quantified conditions, offer many advantages, such as
fast algorithms for testing finite lattice examples and the
use of tools and techniques from propositional calculus.  At the
very least, the manipulation of identities is much simpler both
conceptually and practically than the use of predicate calculus
to work with quantified conditions.

Until 1975, it was thought that the equations defining OML were the only
ones holding in HL.  Then Alan Day discovered the orthoarguesian
equation that holds in any Hilbert lattice but does not in all
OMLs.~\cite{gr-non-s}  Since then, much progress has been made
in finding many new equations that hold in HL and are independent
from the others.

By Birkhoff's HSP theorem \cite[p.~2]{jipsen}, the family {\rm HL} is
not an equational variety, since a finite sublattice is not an {\rm HL}.
A goal of studying equations that hold in {\rm HL} is to find the
smallest variety that includes {\rm HL}, so that the fewest number of of
non-equational (quantified) conditions such as those in
Def.~\ref{def:hl} will be needed to complete the specification of {\rm
HL}.

First we will summarize the equations known so far that hold in
HLs but not in all OMLs (see
Table~\ref{tab:eqns}).  They fall into three major categories:
geometry-related, state-related, and vector-state-related.
The last hold in all ``quantum'' HLs, i.e., those ortho-isomorphic
to Hilbert spaces with real, complex, or quaternion fields but
not necessarily with other fields.

\begin{table}[htp] \caption{Summary of known equations holding in
  (quantum) Hilbert lattices} \label{tab:eqns}
\begin{center}
\begin{tabular}{|l|l|l|l|}
\hline
Equation & Variety & Based on & Definition\\
\hline
\hline
Orthoarguesian & 4OA & geometry & Eq.~(\ref{eq:noa}) \\
\hline
Generalized OA & $n$OA, $n\ge 3$ & geometry & Eq.~(\ref{eq:noa}) \\
\hline
Mayet's $\mathcal{E}_A$ & $\mathcal{E}_A$ & geometry
               & Ref.~\onlinecite{mp-alg-hilb-eq-09} \\
\hline
Godowski & $n$GO, $n\ge 3$ & states & Th.~\ref{th:god-eq} \\
\hline
Mayet-Godowski & MGO & states & Def.~\ref{def:gge} \\
\hline
Mayet's E-equations & $E_n$,\ $n\ge 3$ & vector
                                 & Eqs.~(\ref{eq:E-3}), \\  [-0.5ex]
          &                & states &  (\ref{eq:E-4})\\
\hline
\end{tabular}
\end{center}
\end{table}

The geometry-related equations are derived using the properties
of vectors and subspace sums that hold in a Hilbert space.
They include Day's original
orthoarguesian equation, the generalized orthoarguesian equations,
and Mayet's $\mathcal{E}_A$ equations.

The state-related equations are derived by imposing states (probability
measures) onto Hilbert lattices, and include Godowski's equations and
Mayet-Godowski equations.  (The justification for doing so is that
such states can be defined in Hilbert space, and we map them back to
HL via the ortho-isomorphism of Th.~\ref{th:repr}.) These
equations are derived by finding finite OMLs that
do not admit the ``strong set of states'' condition
(Def.~\ref{def:strong}) that Hilbert lattices do admit, then analyzing
the strong set of states failure in a prescribed way in order to derive
an equation holding in HL but failing in the finite OML.

Vector-state-related equations are derived by imposing
``states'' onto HLs that map to Hilbert-space vectors instead of real
numbers (again, justified by the fact that such ``states'' can
be defined in Hilbert space).  They do not always hold when the
Hilbert-space field implied
by the representation theorem (Th.~\ref{th:repr}) does not have
characteristic 0. ({\em Characteristic 0} means, roughly, that the
number 1 added to itself repeatedly grows without limit.)
This remarkable property narrows down, from the
equation alone, the possible fields for the Hilbert space.
The real, complex, and quaternion fields of quantum mechanics have
characteristic 0, so vector-state-related equations do hold in all
``quantum'' HLs that have the additional properties demanded
by Sol{\`e}r's theorem in Th.~\ref{th:sol}.  The vector-state-related
equations known to date are Mayet's E-equations.

\subsection{States}\label{subsec:states}

\begin{definition}\label{def:state} A state on a lattice {\rm L}
 is a function $m:{\rm L}\longrightarrow [0,1]$
(for real interval $[0,1]$) such
that $m(1)=1$ and $a\perp b\ \Rightarrow\ m(a\cup b)=m(a)+m(b)$,
where $a\perp b$ means $a\le b'$.
\end{definition}

This implies $m(a)+m(a')=1$ and $a\le b\ \Rightarrow\ m(a)\le
m(b)$.

Now, let us recall that the KS theorem and the Bell
inequalities and equalities are all about states and their
experimental recordings that cannot be predetermined i.e.\
fixed in advance. The latter states might be called ``purely''
quantum,\footnote{Because there are quantum states
that can be predetermined for particular setups---{\em repeatable
measurements}.}  as opposed to those that can be
{\em only} predetermined and are called {\em classical}.
We can formalize these two kinds of states as follows.

\begin{definition}\label{def:strong} A nonempty set $S$ of
states on {\rm L} is called a
strong set of {\em classical\/} states if
\begin{align}
(\exists m \in S)(\forall a,b\in{\rm L})((m(a)=1\Rightarrow
 m(b)=1)\Rightarrow a\le b)\,\label{eq:st-cl}
\end{align}
and a strong set of {\em quantum\/} states if
\begin{align}
(\forall a,b\in{\rm L})(\exists m \in S)((m(a)=1 \Rightarrow
 m(b)=1) \Rightarrow a\le b)\,.\label{eq:st-qm}
\end{align}
We assume that {\rm L} contains more than one element and that
an empty set of states is not strong.
\end{definition}

Two important classes of equations that hold in all OMLs with
strong sets of states (and in particular all HLs), but not in all
OMLs, are the Godowski equations and the more general Mayet-Godowski
equations.  Here we only define them for reference; for theorems
and proofs, see Refs.~\onlinecite{pm-ql-l-hql2,mp-alg-hilb-eq-09}.

\begin{definition}\label{def:god-equiv}Let us call the
following expression the {\em Godowski identity}:
\begin{align}
a_1{\buildrel\gamma\over\equiv}a_n{\buildrel{\rm def}
\over =} & (a_1\to a_2)\cap(a_2\to a_3)\cap\cdots
\cap(a_{n-1}\to a_n)  \notag \\ & \qquad
  \cap(a_n\to a_1),
\ n=3,4,\dots\label{eq:god-equiv}
\end{align}
We define $a_n{\buildrel\gamma\over\equiv}a_1$ in the same way with
variables $a_i$ and $a_{n-i+1}$ swapped.
\end{definition}

\begin{theorem}\label{th:god-eq} Godowski's equations {\em\cite{godow}}
\begin{align}
a_1{\buildrel\gamma\over\equiv}a_3
=&a_3{\buildrel\gamma\over\equiv}a_1
\label{eq:godow3o}\\
a_1{\buildrel\gamma\over\equiv}a_4
=&a_4{\buildrel\gamma\over\equiv}a_1
\label{eq:godow4o}\\
a_1{\buildrel\gamma\over\equiv}a_5
=&a_5{\buildrel\gamma\over\equiv}a_1
\label{eq:godow5o}\\
&\dots \notag
\end{align}
hold in all {\rm OML}s with strong sets of states.
\end{theorem}

We call these equations {\rm $n$-Go} {\rm (}{\rm 3-Go},
{\rm 4-Go}, etc.\/{\rm )}.  We also denote by
{\rm $n$GO} {\rm (}{\rm 3GO}, {\rm 4GO}, etc.\/{\rm )} the
{\rm OL} variety determined by {\rm $n$-Go}, and we call
equation {\rm $n$-Go}  the {\rm $n$GO law}.

Next, we define a generalization of this family, first described by
Mayet. \cite{mayet85}  These equations also hold in all lattices
admitting a strong set of states, and in particular in all HLs.

\begin{definition}\label{def:gge}
A {\em Mayet-Godowski equation} ({\rm MGE}) is an equality
with $n\ge 2$ conjuncts on each side:
\begin{align}
t_1 \cap \cdots\cap t_n =& u_1 \cap \cdots\cap u_n
\end{align}
where each conjunct $t_i$ (or $u_i$) is a term consisting of
either a variable or a disjunction of two or more distinct
variables:
\begin{align}
t_i =& a_{i,1}\cup \cdots\cup a_{i,p_i}\qquad \mbox{i.e., $p_i$ disjuncts}\\
u_i =& b_{i,1}\cup \cdots\cup b_{i,q_i}\qquad \mbox{i.e., $q_i$ disjuncts}
\end{align}
and where the following conditions are imposed on the set of variables
in the equation:
\begin{enumerate}
\item{All variables in a given term $t_i$ or $u_i$ are
      mutually orthogonal.}
\item{ Each variable occurs the same number of times on each side of
       the equality.}
\end{enumerate}
\end{definition}
We call a lattice in which all MGEs hold an MGO; i.e., MGO is the
largest class of lattices (equational variety) in which all MGEs hold.
The simplest known example of an equation implied by an MGE that
is independent from all Godowski equations is\cite[p.~775]{pm-ql-l-hql2}
\begin{align}
((a\to  b)\to (c\to  b))\cap(a\to  c)\cap(b\to  a)\le & c\to  a.
\label{eq:newst1d}
\end{align}

Note that a strong set of classical states can be a special case of a
strong set of quantum states for which there exists only a single state
$m$ in Eq.~(\ref{eq:st-qm}). According to the following
theorems, that means that both quantum and classical
states must be orthomodular.

\begin{theorem}\label{th:strong-oml} Any ortholattice that admits a
strong set of quantum states is orthomodular.
\end{theorem}

\begin{proof}
The proof follows from Theorem 3.10 of Ref.~\onlinecite{mpoa99}. Note
that an ortholattice that admits a strong set of quantum states
is much stronger than a bare OML because an infinite sequence of
the Godowski equations holds in every such lattice.
\end{proof}

\begin{theorem}\label{th:strong-distr} Any ortholattice that admits a
strong set of classical states is distributive and therefore also
orthomodular.
\end{theorem}

\begin{proof} Eq.~(\ref{eq:st-qm}) follows from Eq.~(\ref{eq:st-cl})
and by Theorem \ref{th:strong-oml} an ortholattice that admits a strong
set of classical states is orthomodular. Let now $a$ and $b$ be
any two lattice elements.  Assume, for state $m$,
that $m(b)=1$.  Since the lattice admits a strong set of classical
states, this implies $b=1$, so $m(a\cap b)=m(a\cap 1)=m(a)$.
But $m(a')+m(a)=1$ for any state, so $m(a\to b)=m(a')+m(a\cap b)=1$.
Hence we have $m(b)=1\Rightarrow m(a\to b)=1$, which means (since
the ortholattice admits a strong set of classical states)
that $b\le a\to b$.  This is another way of saying $aCb$. \cite{zeman}
By F-H (the Foulis-Holland theorem), an OML in which any
two elements commute is distributive.
\end{proof}

This receives the following explanation within experiments.
Systems submitted to a series of preparations and measurements
are described in a Hilbert space, which is often a product of
Hilbert spaces, but in the Bell and KS experiments,
the experiments are counterfactual. If they give different
outcomes for the same observable under the same preparation
and detection depending on the preparations of other observables,
then they might turn out to be genuinely ``quantum.'' If,
however, they always give one and the same outcome for each
observable, then they are genuinely classical.

\subsection{Vector-valued states}\label{subsec:vectorstates}

What underlies all quantum measurements
is the orthomodular structure of subspaces, i.e., vectors
and---as recently shown by Mayet \cite{mayet06}---states that
related to to the fields over which both quantum and classical
spaces are built: real, complex, or quaternion (skew) field.
These {\em Mayet vector states}
are admitted by quantum, classical, and KS setups but also
those that are wider than quantum.

We stress here that the term {\em setup} basically means a physical
experimental arrangement of devices that manipulate and/or
measure quantum systems. But when we describe the behavior
of a system subjected to these manipulations and measurements,
we include the way the devices affect the systems in the
equations we describe the systems with. Such a description, which
includes the operators and equations that refer to experimental
manipulation and measurements, we also call a {\em setup}.
In our approach, the latter term refers to the particular set of
OML equations that apply to corresponding experimental
manipulations---{\em setup} in the former meaning.
When an ambiguity in the meaning appears, we call the former term
an {\em experimental setup} or {\em e-setup}  for short and the
latter term a {\em formalized setup} or {\em f-setup} for short.
In this paper, the distinction is always clear from the context.
For instance, KS setups are {\em f-setups} throughout
because no realistic experiment is discussed. We formalize the
definition of a {\em setup} as follows.

\begin{definition}\label{def:setup} An {\em experimental setup
(e-setup)} is an experimental arrangement of devices that
manipulate and/or measure quantum systems.
A {\em formalized setup (f-setup)} is a theoretical
description of an {\em experimental setup} within a Hilbert
lattice or a Hilbert space formalism. When it is clear from
context which setup is meant we use the term {\em setup} for
both of them.
\end{definition}

Not all OMLs admit Mayet vector states. There is a class of lattice
OML equations that characterize OMLs that admit these states. Two
smallest equations from the class, $E_3$ and $E_4$, respectively, read:
\begin{align}\hskip-30pt
a\perp &b\ \&\ a\perp c\ \&\ b\perp c\ \&\ a\perp d\ \&\ b\perp e\
    \&\ c\perp f
         \notag \\
& \Rightarrow\
((a\cup b)\cup c)\cap(((a\cup d)\cap(b\cup e))\cap(c\cup f))   \notag \\
&\le
(d\cup e)\cup f,\label{eq:E-3}
  \\
a\perp b\ & \&\ a\perp c\ \&\ a\perp d\ \&\ b\perp c\ \&\ b\perp d\
               \notag \\
& \qquad\&\ c\perp d\
\&\ a\perp e\ \&\ b\perp f\ \&\ c\perp g\ \&\ d\perp h\quad
         \notag \\
\Rightarrow &(((a\cup b)\cup c)\cup d)\cap
        ((((a\cup e)\cap (b\cup f))
                  \notag \\ &\cap (c\cup g))\cap
         (d\cup h))\le ((e\cup f)\cup g)\cup h.
\label{eq:E-4}
\end{align}
These equations pass in most OMLs that characterize properties of both
quantum (Hilbert) and classical spaces including all our
lattices with equal number of vertices (atoms) and
edges (blocks) that we primarily consider in this paper.
However, Eq.~(\ref{eq:E-3}) fails in (a) and (b) OMLs from
Fig.~\ref{fig:l42-e3-e4} and Eq.~(\ref{eq:E-4}) fails in Fig.\
\ref{fig:l42-e3-e4}$\>$(c).

\subsection{Superposition}\label{subsec:superpos}

What also characterizes the quantum---as opposed to
classical---measurements as well as those wider than quantum
is the principle of superposition.
Its main feature is that any two pure states can be superposed
generate a new pure state. In a lattice a pure state $m$
corresponds to an atom $a(m)$. (Atoms are
defined in Def.~\ref{def:hl}(2).)

The following two theorems then cast the superposition within
an OML framework that we need.

\begin{theorem}\label{th:bc1}{\em [Th.~14.8.1 from
\cite{beltr-cass-book}]} Two pure states $m,n$
admit quantum superpositions iff the join of atoms
$a=s(m)$ and  $b=s(n)$, $a\cup b$, contains
at least one different atom $c$, which
then satisfies: $c\ne a$,  $c\ne b$,  $c\le a\cup b$.
\end{theorem}

\begin{theorem}\label{th:bc2}{\em [Th.~14.8.2 from
\cite{beltr-cass-book}]} An OML is classical (distributive)
iff no pair of pure states admits quantum superpositions.
\end{theorem}

The superposition from Theorem \ref{th:bc1} can be formulated
in prenex normal form (to make it easier to use in conjunction
with certain first-order logic algorithms, including our
{\tt latticeg.c} program) as follows
\begin{align}\label{eq:superp}
\hskip-60pt
&(\exists c)(\exists z)(\forall w)
  \\
&((((\neg (a=0)\ \&\ ((\neg (z=0)\ \&\
(z\le a))\ \Rightarrow\ (z=a))) \notag\\
&\&\ (\neg (b=0)
\&\ ((\neg (z=0)\ \&\ (z\le b))\
\Rightarrow\ (z=b)))
  ) \notag\\
&\&\ \neg (a=b))
\Rightarrow\ ((\neg (c=0)\ \&\ ((\neg (w=0)\ \&\ (w\le c))\notag\\
& \Rightarrow\ (w=c)))\&((\neg (c=a)\&\neg (c=b))\&(c\le (a\cup b)))))\notag
\end{align}
where $\neg$, $\&$, and $\Rightarrow$ are classical metaoperations:
negation, conjunction, and implication, respectively.

\subsection{Orthoarguesian equations}\label{subsec:noa}

In the end, there is a series of algebraic equations---we call them
{\em generalized orthoarguesian equations} ($n$OA, $n=3,4,\dots$)---at
least properly overlapping with those characterizing
states and superpositions, that must  hold in all lattices of
closed subspaces of both finite- and infinite-dim Hilbert space (and
therefore in a Hilbert lattice).
They follow from the following set of equations that hold in any
Hilbert space.

\begin{theorem} \label{th:hs-ssnoa}
Let ${\mathcal M}_0,\ldots, {\mathcal M}_n$ and ${\mathcal N}_0,\ldots,
{\mathcal N}_n$, $n\ge 1$,
 be any subspaces (not necessarily closed) of a Hilbert
space, and let $\bigcap$ denote set-theoretical
intersection and $+$ subspace sum.
We define the subspace term ${\mathcal T}_n(i_0,\ldots,i_n)$ recursively as
follows, where $0\le i_0,\ldots,i_n\le n$:
\begin{align}
&{\mathcal T}_1(i_0,i_1)=({\mathcal M}_{i_0}
         +{\mathcal M}_{i_1})\text{$\bigcap$}({\mathcal N}_{i_0}
         +{\mathcal N}_{i_1}) \label{eq:hs-rec1} \\
{\mathcal T}_m&(i_0,\ldots,i_m)={\mathcal T}_{m-1}(i_0,i_1,i_3,\ldots,i_m)
          \notag \\
&\text{$\bigcap$}({\mathcal T}_{m-1}
(i_0,i_2,i_3,\ldots,i_m)+{\mathcal  T}_{n-1}(i_1,i_2,i_3,\ldots,i_m)),\notag \\
       &\qquad\qquad 2\le m\le n    \label{eq:hs-recn}
\end{align}

For $m=2$, this means ${\mathcal T}_2(i_0,i_1,i_2)={\mathcal T}_1(i_0,i_1)$
$\bigcap\ ({\mathcal T}_1(i_0,i_2)+{\mathcal T}_1(i_1,i_2))$.  Then the following condition holds in any
finite- or infinite-dimensional Hilbert space for $n\ge 1$:
\begin{align}
&({\mathcal M}_0+{\mathcal N}_0)\text{$\bigcap$}
\cdots\text{$\bigcap$}({\mathcal M}_n+{\mathcal N}_n)\notag\\
&\subseteq {\mathcal N}_0
+({\mathcal M}_0\text{$\bigcap$}
         ({\mathcal M}_1+{\mathcal T}_n(0,\ldots,n) )).  \label{eq:hs-ssnoa}
\end{align}
\end{theorem}

\begin{proof} (Originally given---in effect---in the proof of Theorem~5.2
of \cite{mpoa99}; a similar proof was also given by R.~Mayet
\cite{mayet06-hql2})
We will use $+$ to denote subspace sum when connecting two subspaces
and vector sum when connecting two vectors; no confusion should arise.
Let  $x$ be a vector belonging to the left-hand side of
Eq.~(\ref{eq:hs-ssnoa}).  Then
$x\in {\mathcal M}_i+{\mathcal N}_i$ for $i=0,\ldots,n$.
From the definition of subspace sum,
$x\in {\mathcal M}_i+{\mathcal N}_i$ implies there
exist vectors $x_i$ and $y_i$ such that $x_i\in {\mathcal M}_i$,
$y_i\in {\mathcal N}_i$, and $x=x_i+y_i$.  From the last property,
we have $x_i+y_i=x=x_j+y_j$ or
\begin{align}
  x_i-x_j&=-y_i+y_j,   &0\le i,j\le n.     \label{eq:hs-diff}
\end{align}

For the case $n=1$ of Eq.~(\ref{eq:hs-ssnoa}), we need to prove
\begin{align}
&({\mathcal M}_0+{\mathcal N}_0) \text{$\bigcap$}({\mathcal M}_1
            +{\mathcal N}_1)
    \notag \\
      & \hskip-10pt\subseteq
          {\mathcal N}_0+({\mathcal M}_0\text{$\bigcap$} ({\mathcal M}_1
               +(({\mathcal M}_0+{\mathcal M}_1)\text{$\bigcap$}({\mathcal N}_0
                   +{\mathcal N}_1)) ))  \label{eq:hs-3oa}
\end{align}


Any linear combination of vectors from two subspaces belongs to their
subspace sum.  Since $y_0\in {\mathcal N}_0$ and $y_1\in {\mathcal N}_1$, we have $-y_0+y_1\in
{\mathcal N}_0+{\mathcal N}_1$.  Therefore by Eq.~(\ref{eq:hs-diff}), $x_0-x_1\in {\mathcal N}_0+{\mathcal N}_1$.
Also, $x_0-x_1\in {\mathcal M}_0+{\mathcal M}_1$.  Therefore
\begin{align}
x_0-x_1\in ({\mathcal M}_0+{\mathcal M}_1)\mbox{$\bigcap$}({\mathcal N}_0
     +{\mathcal N}_1) \label{eq:hs-01}.
\end{align}
Since $x_1\in {\mathcal M}_1$, we have $x_0=x_1+(x_0-x_1)\in
{\mathcal M}_1+(({\mathcal M}_0+{\mathcal M}_1)\bigcap({\mathcal N}_0+{\mathcal N}_1))$.  Also, $x_0\in {\mathcal M}_0$, so $x_0\in
{\mathcal M}_0\bigcap({\mathcal M}_1+(({\mathcal M}_0+{\mathcal M}_1)\bigcap({\mathcal N}_0+{\mathcal N}_1)))$.  Finally, since $y_0\in {\mathcal N}_0$, we
have $x=y_0+x_0\in {\mathcal N}_0+({\mathcal M}_0\bigcap({\mathcal M}_1+(({\mathcal M}_0+{\mathcal M}_1)\bigcap({\mathcal N}_0+{\mathcal N}_1))))$, proving
that $x$ belongs to the right-hand side of Eq.~(\ref{eq:hs-3oa}) and
thus establishing the subset relation.  This argument is illustrated by
the following diagram:
\begin{widetext}
\begin{align}
  \cdots\subseteq
\underbrace{
  \underbrace{{\mathcal N}_0}_{\textstyle y_0}
  +(
  \underbrace{{\mathcal M}_0}_{\textstyle x_0}
  \mbox{$\bigcap$}
  \underbrace{(
    \underbrace{{\mathcal M}_1}_{\textstyle x_1}
    +(
    \underbrace{
      \underbrace{({\mathcal M}_0+{\mathcal M}_1)}_{\textstyle x_0-x_1}
      \mbox{$\bigcap$}
      \underbrace{({\mathcal N}_0+{\mathcal N}_1)}_{\textstyle -y_0+y_1=x_0-x_1}
      )))
    }_{\textstyle x_0-x_1}
  }_{\textstyle x_1+(x_0-x_1)=x_0}
}_{\textstyle y_0+x_0=x}.    \notag 
\end{align}
\end{widetext}

For $n>1$, notice that on the right-hand side, the term
${\mathcal T}_1(0,1)=({\mathcal M}_0+{\mathcal M}_1)\bigcap({\mathcal N}_0+{\mathcal N}_1)$ in Eq.~(\ref{eq:hs-3oa}) is replaced
by the larger term ${\mathcal T}_n(0,\ldots,n)$, with the rest of the right-hand
side the same.  From the diagram above, it is apparent that if we can
prove
\begin{align}
x_0-x_1\in {\mathcal T}_n(0,\ldots,n), \label{eq:hs-0-1}
\end{align}
then Eq.~(\ref{eq:hs-ssnoa}) is
established.
We will actually prove a more general result,
\begin{align}
x_{i_0}-x_{i_1} &\in {\mathcal T}_m(i_0,\ldots,i_m), &
                  0\le i_0,\ldots,i_m\le n, &  1\le m\le n
                \label{eq:hs-i0-i1}
\end{align}
from which Eq.~(\ref{eq:hs-0-1}) follows as a special case by setting
$m=n$ and $i_0=0,\ldots,i_m=n$.

We will prove Eq.~(\ref{eq:hs-i0-i1}) by induction on $m$.
For the basis step $m=1$,
the same argument that led to Eq.~(\ref{eq:hs-01}) above shows that
\begin{align}
x_{i_0}-x_{i_1}\in {\mathcal T}_1(i_0,i_1) &=({\mathcal M}_{i_0}
     +{\mathcal M}_{i_1})\mbox{$\bigcap$}({\mathcal N}_{i_0}
       +{\mathcal N}_{i_1}).
             \notag
\end{align}
for $0\le i_0,i_1\le n$.  For $m>1$, assume we have proved
$x_{i_0}-x_{i_1}\in {\mathcal T}_{m-1}(i_0,i_1,\ldots,i_{m-1})$
for all $0\le i_0,\ldots,i_{m-1}\le n$.
Then, in particular, we have the substitution instances
\begin{align}
x_{i_0}-x_{i_1}&\in {\mathcal T}_{m-1}(i_0,i_1,i_3,\ldots,i_m) \label{eq:hs-n01} \\
x_{i_0}-x_{i_2}&\in {\mathcal T}_{m-1}(i_0,i_2,i_3,\ldots,i_m) \label{eq:hs-n02} \\
x_{i_1}-x_{i_2}&\in {\mathcal T}_{m-1}(i_1,i_2,i_3,\ldots,i_m). \label{eq:hs-n12}
\end{align}
Combining Eqs.~(\ref{eq:hs-n02}) and (\ref{eq:hs-n12}),
\begin{align}
&x_{i_0}-x_{i_1}=(x_{i_0}-x_{i_2})-(x_{i_1}-x_{i_2})\notag\\
&\in{\mathcal T}_{m-1}(i_0,i_2,i_3,\ldots,i_m)  +{\mathcal T}_{m-1}(i_1,i_2,i_3,\ldots,i_m). \notag
\end{align}
Combining this with Eq.~(\ref{eq:hs-n01}) and using
Eq.~(\ref{eq:hs-recn}),
\begin{align}
&x_{i_0}-x_{i_1} \in
{\mathcal T}_{m-1}(i_0,i_1,i_3,\ldots,i_m)    \notag \\
 &\mbox{$\bigcap$}({\mathcal T}_{m-1}(i_0,i_2,i_3,\ldots,i_n)
  +{\mathcal T}_{m-1}(i_1,i_2,i_3,\ldots,i_m))  \notag \\
    & ={\mathcal T}_m(i_0,\ldots,i_m)                  \notag
\end{align}
as required.
\end{proof}

We will use the above theorem to derive a condition that holds in the
lattice of closed subspaces of a Hilbert space. In doing so we will
make use of the definitions introduced at the beginning of Sec.\
\ref{sec:ortho} and  the following well-known \cite[p.~28]{halmos}
lemma.
\begin{lemma}
\label{lem:hs-sum} Let ${\mathcal M}$ and ${\mathcal N}$ be two closed
subspaces of a Hilbert space. Then
\begin{align}
{\mathcal M}+{\mathcal N} & \subseteq {\mathcal M}\text{$\bigcup$} {\mathcal N} \label{eq:hs-sumss} \\
{\mathcal M}\perp {\mathcal N} \quad & \Rightarrow\quad {\mathcal M}+{\mathcal N}= {\mathcal M}\text{$\bigcup$} {\mathcal N}
\label{eq:hs-sumeq}
\end{align}
\end{lemma}
\begin{theorem} \label{th:hs-noa} {\rm (Generalized Orthoarguesian Laws)}
Let ${\mathcal M}_0,\ldots, {\mathcal M}_n$ and ${\mathcal N}_0,\ldots,
{\mathcal N}_n$, $n\ge 1$, be closed subspaces of a Hilbert space.
We define the term ${\mathcal T}^{\tiny\mbox{$\bigcup$}}_n(i_0,\ldots,i_n)$ by
substituting $\bigcup$ for $+$ in the term
${\mathcal T}_n(i_0,\ldots,i_n)$ from Theorem~\ref{th:hs-ssnoa}.
Then following condition holds in any finite- or
infinite-dimensional Hilbert space for
$n\ge 1$:
\begin{align}
{\mathcal M}_0 & \perp {\mathcal N}_0 \ \& \ \cdots \ \& \
      {\mathcal M}_n\perp {\mathcal N}_n \quad\Rightarrow\quad
          \notag \\
  &  ({\mathcal M}_0\text{$\bigcup$} {\mathcal N}_0)\text{$\bigcap$}
\cdots\text{$\bigcap$}({\mathcal M}_n\text{$\bigcup$} {\mathcal N}_n)
          \notag \\
&\le
          {\mathcal N}_0\text{$\bigcup$} ({\mathcal M}_0\text{$\bigcap$}
         ({\mathcal M}_1\text{$\bigcup$}
              {\mathcal T}^\text{$\bigcup$}_n(0,\ldots,n) )). \label{eq:hs-noa}
\end{align}
\end{theorem}
\begin{proof}
By the orthogonality hypotheses and Eq.~(\ref{eq:hs-sumeq}), the
left-hand side of Eq.~(\ref{eq:hs-noa}) equals the left-hand side
of Eq.~(\ref{eq:hs-ssnoa}).  By Eq.~(\ref{eq:hs-sumss}),
the right-hand side of Eq.~(\ref{eq:hs-ssnoa}) is a subset
of the right-hand side of Eq.~(\ref{eq:hs-noa}).  Eq.~(\ref{eq:hs-noa})
follows by Theorem~\ref{th:hs-ssnoa} and the transitivity of
the subset relation.
\end{proof}

Ref.~\onlinecite{mpoa99} shows that in any OML (which includes the lattice of
closed subspaces of a Hilbert space, i.e., the Hilbert lattice),
Eq.~(\ref{eq:hs-noa}) is equivalent to the $m$OA law Eq.~(\ref{eq:noa})
for $m=n+2$, thus establishing the proof of  Theorem~\ref{th:noa}.

\begin{definition}
\label{def:noa}
We define an operation
${\buildrel (n)\over\equiv}$ on $n$ variables
$a_1,\ldots,a_n$ ($n\ge 3$) as follows:
\begin{align}
a_1&{\buildrel (3)\over\equiv}a_2\notag\\
&{\buildrel\rm def\over =}\
((a_1\to  a_3)\cap(a_2\to  a_3))
\cup((a_1'\to  a_3)\cap(a_2'\to  a_3)) \notag\\
a_1&{\buildrel (n)\over\equiv}a_2\notag\\
&{\buildrel\rm def\over =}\ (a_1{\buildrel (n-1)\over\equiv}a_2)\cup
((a_1{\buildrel (n-1)\over\equiv}a_n)\cap
(a_2{\buildrel (n-1)\over\equiv}a_n)), \notag\\
&\qquad\qquad n\ge 4\,.\label{eq:noaoper}
\end{align}
\end{definition}

\begin{theorem}\label{th:noa}
The $n${\rm OA} {\em laws}
\begin{align}
(a_1\to a_3) \cap (a_1{\buildrel (n)\over\equiv}a_2)
\le a_2\to  a_3\,.\label{eq:noa}
\end{align}
hold in any Hilbert lattice.
\end{theorem}

The class of  equations (\ref{eq:noa}) are the {\em generalized
orthoarguesian equations} $n${\rm OA} discovered by Megill and
Pavi\v ci\'c. \cite{mpoa99,pm-ql-l-hql2} They also play a role in
proving the semi-quantum lattice theorem (Subsection~\ref{subsec:semi}).

The smallest of the generalized  orthoarguesian equations is the following
3OA:
\begin{align}
(x\to z)\cap &(((x\to z)\cap (y\to z))\cup ((x'\to z)\cap (y'\to z)))\notag\\
    &\le y\to z
\label{eq:oa3}
\end{align}
All $n$OA imply 3OA, so, if an OML does not satisfy 3OA
it will not admit any $n$OA.

\subsection{\label{subsec:greechie}Greechie diagrams}

A Greechie diagram of an OML is a shorthand graphical representation
of a Hasse diagram of an OML.

\begin{definition}\label{def:hasse}
A {\em Hasse diagram of an\/} {\rm OML} is a graphical representation of
an {\rm OML} displayed via its ordering relation with an implied
upward orientation. A point is drawn for each element of the
{\rm  OML} and line segments are drawn between these points
according to the following two rules:
\begin{enumerate}
\item[(1)] If $a<b$ in the lattice, then the point corresponding to $a$
appears lower in the drawing than the point corresponding to $a$;
\item[{2}] A line segment is drawn between the points corresponding to any
two elements $a$ and $b$ of the lattice iff
either $a$ covers $b$ or $b$ covers $a$.  ($a$ {\em covers} $b$ iff
$b<a$ and there is no $c$ such that $b<c<a$.)
\end{enumerate}
\end{definition}

The most general definition of a Hasse diagram is given for
a partially ordered set (poset), but all we deal with in this paper
is a very special poset---OML---and therefore we defined a Hasse
diagram directly for an OML above.

\begin{definition}\label{def:greechie}
A {\em Greechie diagram of an\/} {\rm OML} is a graphical
representation of a Hasse diagram of an {\rm OML\/} in which points
represent atoms [Def.~\ref{def:hl}(2)]  and smooth lines---called
{\em blocks\/}---that connect points/atoms---represent the orthogonalities
between atoms.
\end{definition}

The most general definition of a Greechie diagram is also given
for a poset but this is again too general for our purpose.  A precise
definition can be found, for example, in Ref.~\onlinecite[p.~38]{kalmb83},
which includes conditions---e.g., that there be no loops of order less
than five---necessary for the diagram to be an OML.  To avoid certain
complications, we consider only those Greechie diagrams with three or
more atoms per block.

In Fig.~\ref{fig:mmp} we show two Greechie diagrams and their
Hasse diagrams. The points in a Hasse diagrams that represent mutually
orthogonal atoms, which themselves represent orthogonal vectors, span a
hyperplane or the whole space.
Thus the orthogonalities imply that
the top elements under 1 in the diagrams are complements of the atoms
in the lowest level above 0.

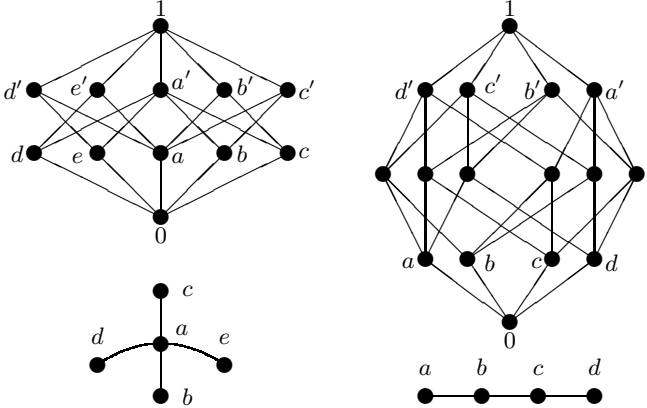
\begin{figure}[htp]\centering
  \setlength{\unitlength}{0.8pt}
  \begin{picture}(330,200)(0,0)
    \put(35,0){
      \begin{picture}(60,80)(0,0)
        \qbezier(0,40)(30,60)(60,40)
        \put(30,25){\line(0,1){50}}
        \put(30,75){\circle*{7}}
        \put(30,25){\circle*{7}}
        \put(30,50){\circle*{7}}
        \put(0,40){\circle*{7}}
        \put(60,40){\circle*{7}}
        \put(0,50){\makebox(0,0)[b]{$d$}}
        \put(40,54){\makebox(0,0)[b]{$a$}}
        \put(60,50){\makebox(0,0)[b]{$e$}}
        \put(40,25){\makebox(0,0)[l]{$b$}}
        \put(40,75){\makebox(0,0)[l]{$c$}}
      \end{picture}
    } 
    \put(5,110){
      \begin{picture}(60,90)(0,0)
        \put(60,90){\line(-2,-1){60}}
        \put(60,90){\line(-1,-1){30}}
        \put(60,90){\line(0,-1){30}}
        \put(60,90){\line(1,-1){30}}
        \put(60,90){\line(2,-1){60}}
        \put(60,0){\line(-2,1){60}}
        \put(60,0){\line(-1,1){30}}
        \put(60,0){\line(0,1){30}}
        \put(60,0){\line(1,1){30}}
        \put(60,0){\line(2,1){60}}
        \put(0,60){\line(1,-1){30}}
        \put(0,60){\line(2,-1){60}}
        \put(30,60){\line(-1,-1){30}}
        \put(30,60){\line(1,-1){30}}
        \put(60,60){\line(-2,-1){60}}
        \put(60,60){\line(-1,-1){30}}
        \put(60,60){\line(1,-1){30}}
        \put(60,60){\line(2,-1){60}}
        \put(90,60){\line(-1,-1){30}}
        \put(90,60){\line(1,-1){30}}
        \put(120,60){\line(-2,-1){60}}
        \put(120,60){\line(-1,-1){30}}

        \put(60,-5){\makebox(0,0)[t]{$0$}}
        \put(-5,30){\makebox(0,0)[r]{$d$}}
        \put(18,28){\makebox(0,0)[l]{$e$}}
        \put(65,28){\makebox(0,0)[l]{$a$}}
        \put(96,30){\makebox(0,0)[l]{$b$}}
        \put(125,30){\makebox(0,0)[l]{$c$}}
        \put(-5,60){\makebox(0,0)[r]{$d'$}}
        \put(17,62){\makebox(0,0)[l]{$e'$}}
        \put(65,65){\makebox(0,0)[l]{$a'$}}
        \put(96,62){\makebox(0,0)[l]{$b'$}}
        \put(125,60){\makebox(0,0)[l]{$c'$}}
        \put(60,95){\makebox(0,0)[b]{$1$}}
        \put(60,0){\circle*{7}}
        \put(0,30){\circle*{7}}
        \put(30,30){\circle*{7}}
        \put(60,30){\circle*{7}}
        \put(90,30){\circle*{7}}
        \put(120,30){\circle*{7}}
        \put(0,60){\circle*{7}}
        \put(30,60){\circle*{7}}
        \put(60,60){\circle*{7}}
        \put(90,60){\circle*{7}}
        \put(120,60){\circle*{7}}
        \put(60,90){\circle*{7}}
      \end{picture}
}
   \put(170,25){
      \begin{picture}(100,0)(0,0)
        \put(20,0){\line(1,0){80}}
        \put(20,0){\circle*{7}}
        \put(46.67,0){\circle*{7}}
        \put(73.33,0){\circle*{7}}
        \put(100,0){\circle*{7}}
        \put(20,10){\makebox(0,0)[b]{$a$}}
        \put(46.67,10){\makebox(0,0)[b]{$b$}}
        \put(73.33,10){\makebox(0,0)[b]{$c$}}
        \put(100,10){\makebox(0,0)[b]{$d$}}
      \end{picture}
    }

    \put(170,60){
      \begin{picture}(100,140)(0,0)
        \put(60,0){\line(-2,3){20}}
        \put(60,0){\line(2,3){20}}
        \put(60,0){\line(-4,3){40}}
        \put(60,0){\line(4,3){40}}
        \put(60,140){\line(-2,-3){20}}
        \put(60,140){\line(2,-3){20}}
        \put(60,140){\line(-4,-3){40}}
        \put(60,140){\line(4,-3){40}}

        \put(20,30){\line(-1,2){20}}
        \put(20,30){\line(0,1){40}}
        \put(20,30){\line(1,2){20}}
        \put(40,30){\line(-1,1){40}}
        \put(40,30){\line(1,1){40}}
        \put(40,30){\line(3,2){60}}
        \put(80,30){\line(-3,2){60}}
        \put(80,30){\line(0,1){40}}
        \put(80,30){\line(1,1){40}}
        \put(100,30){\line(-3,2){60}}
        \put(100,30){\line(0,1){40}}
        \put(100,30){\line(1,2){20}}

        \put(100,110){\line(1,-2){20}}
        \put(100,110){\line(0,-1){40}}
        \put(100,110){\line(-1,-2){20}}
        \put(80,110){\line(1,-1){40}}
        \put(80,110){\line(-1,-1){40}}
        \put(80,110){\line(-3,-2){60}}
        \put(40,110){\line(3,-2){60}}
        \put(40,110){\line(0,-1){40}}
        \put(40,110){\line(-1,-1){40}}
        \put(20,110){\line(3,-2){60}}
        \put(20,110){\line(0,-1){40}}
        \put(20,110){\line(-1,-2){20}}

        \put(60,140){\circle*{7}}
        \put(20,110){\circle*{7}}
        \put(40,110){\circle*{7}}
        \put(80,110){\circle*{7}}
        \put(100,110){\circle*{7}}
        \put(0,70){\circle*{7}}
        \put(20,70){\circle*{7}}
        \put(40,70){\circle*{7}}
        \put(80,70){\circle*{7}}
        \put(100,70){\circle*{7}}
        \put(120,70){\circle*{7}}
        \put(20,30){\circle*{7}}
        \put(40,30){\circle*{7}}
        \put(80,30){\circle*{7}}
        \put(100,30){\circle*{7}}
        \put(60,0){\circle*{7}}

        \put(60,-5){\makebox(0,0)[t]{$0$}}
        \put(105,28){\makebox(0,0)[l]{$d$}}
        \put(75,28){\makebox(0,0)[r]{$c$}}
        \put(48,28){\makebox(0,0)[l]{$b$}}
        \put(15,28){\makebox(0,0)[r]{$a$}}
        \put(105,110){\makebox(0,0)[l]{$a'$}}
        \put(75,110){\makebox(0,0)[r]{$b'$}}
        \put(48,113){\makebox(0,0)[l]{$c'$}}
        \put(15,110){\makebox(0,0)[r]{$d'$}}
        \put(60,145){\makebox(0,0)[b]{$1$}}


      \end{picture}
    }
  \end{picture}
\caption{3- and 4-dim Greechie diagrams and their
corresponding Hasse diagrams shown
above them.~\cite[Fig.~18, p.~84]{beran}
\label{fig:mmp}}
\end{figure}

The Hasse diagrams shown in Fig.~\ref{fig:mmp} is a subalgebra of
a Hilbert lattice but, as we show below (Th.~\ref{th:notsubalgebra}),
already a 3-dim one with
a third orthogonal triple attached to it is not. Therefore, if we tried
to arrive at complete lattices in a realistic application by reading off
all properties from a  corresponding Hilbert space description, we
would end up with complicated and unmanageable properties.
If we used just orthogonalities between, say, spin projections of a
considered system, we would arrive at an incorrect  description
by means of Greechie diagrams. In other words Greechie diagrams cannot
represent all possible OMLs---to do so, we also need more complicated
interconnections of blocks called {\em pastings}\cite[p.~48]{kalmb83}
that we do not describe here.

As mentioned below Def.~\ref{def:hl}, the number of atoms in an HL is
infinite, which means that finite Greechie diagrams cannot represent an
HL.  However, because of their practical advantages, it is natural to
ask whether Greechie diagrams can serve in the role of partial
representations or approximate representations of HLs, as has
been sometimes assumed in the literature as mentioned in the Introduction.
First, we make precise the notion of a partial representation with
the following definition.

\begin{definition}\label{def:subalgebra}

A {\em subalgebra} of an {\rm OL} (and thus an {\rm OML}, {\rm HL}, etc.)
$L=\langle L_0,',\cup,\cap\rangle$ is a set
$M=\langle M_0,',\cup, \cap\rangle$ where $M_0$ is a subset of $L_0$,
the operations $',\cup,\cap$ of $M$ are the same as the operations of $L$
(optionally restricted to $M_0$), and $M_0$ is closed under the
operations of $L$ (and therefore of $M$).

\end{definition}

Because the notion of subalgebra is crucial to our argument, we will
elaborate on it slightly.  Some literature definitions can be misleading
if not read carefully.  For example, Kalmbach \cite[p.~22]{kalmb83}
omits the algebra component breakdown as well as the word ``same.''
The reader could interpret an OML $M$ as being a subalgebra of
$L$ as long as $M_0$ is a subset of $L_0$ and $M_0$ is closed under the
operations of $M$ (even if different from the operations of $L$, which
might be the case if the operation symbols are interpreted as being
local to their associated algebras as is that author's convention elsewhere).
A careful definition can be found in e.g., Beran \cite[p.~18]{beran}.

\begin{lemma} \label{lem:subalgebra}
If $M$ is a subalgebra of $L$, then any equation (identity) that holds
in $L$ will continue to hold in $M$.  Equivalently, if an equation
fails in $M$ but holds in $L$, then $M$ cannot be a subalgebra of
$L$.
\end{lemma}
\begin{proof}
This is obvious from the fact that the operations on $M$ are equal to
the operations on $L$ (when restricted to the base set $M_0$ of $M$).
Any evaluation of an equation in $M$, i.e.~using elements from $M_0$,
will have the same final value as the same evaluation in $L$.  Since the
equation always holds in $L$, it will also always hold in $M$.
\end{proof}

{\em Remark.} Note that the above lemma does not necessarily apply to
quantified conditions.  A quantified condition, such as superposition
[Def.~\ref{def:hl}(3); Eq.~(\ref{eq:superp})], that holds in a lattice
may not hold in a sublattice.  As a trivial example, the quantified
condition ``has more than two elements'' does not hold in the
two-element subalgebra consisting of 0 and 1. Although superposition
holds vacuously in the two-element subalgebra (because it has only one
atom), it fails in the 3-dim Greechie diagram of Fig.~\ref{fig:mmp},
which is a subalgebra of any HL (in which superposition holds).

In the case of an OML represented by a Greechie diagram, a subgraph is
not necessarily a subalgebra.  A counterexample is provided by Fig.~8a
and Fig.~8b of Ref.~\onlinecite{bdm-ndm-mp-1}, where the first figure is
a Greechie diagram that is a subgraph of the second, but the
corresponding OMLs do not have a subalgebra relationship.  In
particular, an equation holding in a Greechie diagram may not hold in a
subgraph of it, as that example shows.

The question as to whether Greechie diagrams can be subalgebras of
Hilbert lattices is answered by the following theorem.

\begin{theorem}\label{th:notsubalgebra}
Any Greechie diagram containing blocks that do not share atoms is not a
subalgebra of the lattice $\mathcal{C}(\mathcal{H})$ for a Hilbert
space with dimension 3 or greater.
\end{theorem}
\begin{proof}
Choose one atom from each such block that are not shared with a common
third block.  (There will always be such atoms due to the requirement
that there be no loops of order less than five in a Greechie diagram.)
The join of these two atoms is the lattice unit.  However,
in any $\mathcal{C}(\mathcal{H})$, the join of any two distinct
atoms (one-dimensional subspaces spanned by vectors) whatsoever
spans a 2-dimensional subspace, which for a Hilbert space of
dimension $>2$ is not the whole space (lattice unit).  This violates
the requirement of Def.~\ref{def:subalgebra} that the operations be
the same.
\end{proof}

Thus the only Greechie diagrams that can be subalgebras of a lattice of
Hilbert space subspaces with dim $> 2$, and thus of the ortho-isomorphic
HL, are either single blocks, such as in Fig.~\ref{fig:mmp}, or
those in a ``star'' configuration where all blocks share a
common atom (Fig.~\ref{fig:mod-star}).

\subsection{Semi-quantum lattices}\label{subsec:semi}

Now we can state our main theorem.
\begin{theorem}\label{th:rksa}{\em [Semi-quantum lattice algorithms]}
There exist OMLs represented
by Greechie diagrams that admit superposition,
real-valued states, and a vector state given by Eq.~(\ref{eq:E-4})
but do not admit other conditions that have to be satisfied
by every Hilbert lattice, in particular equations like the orthoarguesian
and Godowski ones. As a consequence of violating Godowski
equations, these OMLs do not admit strong sets of states.
\end{theorem}

We point out here that we developed special algorithms and
programs (e.g., {\tt states}) that follow the definition Def.~\ref{def:strong}
of the strong set of states and are much faster than
those that check whether an equation passes in a lattice.
Besides, a lattice that satisfies Godowski equations need not
admit a strong set of states.

The generation algorithms mentioned in Theorem \ref{th:rksa}
are presented in Sec.~\ref{sec:one-state}.
The outcomes of our massive computations, given in Sec.\
\ref{sec:one-state-p} and based on these algorithms, provide
Theorem \ref{th:rksa} with the following corollary:

\begin{corollary}\label{th:rks}{\em [Semi-quantum lattices]}
There exists a class of OMLs that admit superposition,
real-valued states, and a vector state but do not admit other
conditions that have to be satisfied by every Hilbert lattice.
\end{corollary}

This corollary corresponds to the original KS theorem and
Theorem \ref{th:rksa} corresponds to the algorithms that
generate KS vectors as given in Ref.~\onlinecite{pmmm03a}.
Moreover, hopefully we shall be able use the same algorithms
to generate genuine and complete KS setups and prove
a non-vacuous KS theorem, because an OML
that admits Mayet vector states and superposition and
all other Hilbert lattice conditions corresponds to
a realistic quantum system whose measurement does not allow
a classical interpretation. For the time being, however,
this project apparently exceeds today's computing power.

As shown in the next sections, we can give the proof of the
theorems in several different ways. However, our main proof
is provided by algorithms for exhaustive generation of
Greechie diagrams with
equal number of atoms and blocks generated from cubic bipartite
graphs presented in Sec.~\ref{sec:one-state}. We generated
all such lattices from the smallest ones with 35 atoms and 35
blocks through all those that have 41 atoms and 41 blocks
in which particular known Hilbert lattice equations fail.  Thus,
although they satisfy a number of Hilbert lattice conditions
they represent impossible setups.

\section{\label{sec:represent}Why 3D Kochen-Specker
Setups Cannot Be Described with Greechie
Diagrams, and How They Can Be}

In the Introduction we mentioned that the Hultgren and Shimony
tried to build up a lattice that would correspond to a
spin-1 Stern-Gerlach experiment.  Orthogonal vectors of spin-1
projections determine directions in which we prepare spin projections
of a particle or orient our detection devices.
We can choose one-dimensional subspaces ${\mathcal H}_a,\dots
,{\mathcal H}_e$ as shown in Fig.~\ref{fig:mmp}, where we
denote them as $a,\dots,e$. The first Hasse diagram shown in
Fig.~\ref{fig:mmp} graphically represents the orthogonality
between the vectors in a 3-dim space---in our case the ones between
each chosen vector and a plane determined by the other two.
In particular, the orthogonalities are $a\perp b,c,d,e$ since
$a\le b',c',d',e'$, $b\perp c$ since $a\le c'$, and
$d\perp e$ since $d\le e'$. Also, e.g., $b'$ is a
complement of $b$ and that means a plane to which $b$
is orthogonal: $b'=a\cup c$. Eventually $b\cup b'=1$ where
$1$ stands for $\mathcal H$.

That shows that if we wanted to use a Greechie diagram for some
application or if wanted to just generate it or check on some of its
properties we have to use all the elements of its Hasse diagrams.
So, our idea is to use a graphical pattern of Greechie diagrams
directly and to go around all the elements contained in the Hasse
diagrams. For that we needed another definition of a Greechie
diagram which exploited only graphical elements of its
shorthand representation of a Hasse diagram---atoms and blocks.
The following lemma provide us with such a definition.

\begin{lemma}\label{lemma:greechie}
A definition equivalent to Def.~\ref{def:greechie} is the following
one~\cite{greechie71}
\begin{enumerate}
\item[(1)] Every atom belongs to at least one block;
\item[(2)] If there are at least two atoms, then every block is at least
2-element;
\item[(3)] Every block which intersects with another block is at least
3-element;
\item[(4)] Every pair of different blocks intersects in at most one atom;
\item[(5)] There is no loop of order less than 5,
\end{enumerate}
where {\em loop} of order $n\ge 2$---$(b_1,\dots, b_n)$ is a sequence of
different blocks such that there are mutually distinct atoms
$a_1,\dots,a_n$ with $a_i\in b_i\cap b_{i+1}\ (i=1,\ldots,n;\ b_{n+1}=b_1)$.
\end{lemma}

Using this definition we recognize a Greechie diagram as a special
case of an MMP hypergraph.

\begin{definition}\label{def:mmp}
A {\em hypergraph} is a set of vertices (drawn as points) together
with a set of
edges (drawn as line segments connecting points).
An {\em MMP hypergraph} is a hypergraph in which
\begin{itemize}
\item[(i)] Every vertex belongs to at least one edge;
\item[(ii)] Every edge contains at least 3 vertices;
\item[(iii)] Edges that intersect each other in $n-2$
         vertices contain at least $n$ vertices.
\end{itemize}
\end{definition}

This definition enables us to formulate algorithms for
exhaustive generation of MMP hypergraphs, which are
exponentially faster than possible generation of Greechie
diagrams by means of Def.~\ref{def:greechie}, because MMP
hypergraphs are just sets of vertices and edges with no other
meaning or conditions imposed on them. Any condition we want
lattices to satisfy we build into generation algorithms, which
can speed up the generation further. As opposed to this,
a lattice approach requires the generation of all
possible lattices first and then filtering out lattices
that meet the condition. For the time being we just
assume that each vertex (atom; see below) is orthogonal
to other two on the edge they share. But as opposed
to Greechie diagrams we shall also have relations
between nonorthogonal vertices.

\font\1=cmss8
\font\2=cmssdc8
\font\3=cmr8

We encode MMP hypergraphs by means of alphanumeric
and other printable ASCII characters. Each vertex (atom)
is represented by one of the following
characters: {\11\hfil $\:$2\hfil $\:$3\hfil $\:$4\hfil
$\:$5\hfil $\:$6\hfil $\:$7\hfil $\:$8\hfil $\:$9\hfil
$\:$A\hfil $\:$B\hfil $\:$C\hfil $\:$D\hfil $\:$E\hfil
$\:$F\hfil $\:$G\hfil $\:$H\hfil $\:$I\hfil $\:$J\hfil
$\:$K\hfil $\:$L\hfil $\:$M\hfil $\:$N\hfil $\:$O\hfil
$\:$P\hfil $\:$Q\hfil $\:$R\hfil $\:$S\hfil $\:$T\hfil
$\:$U\hfil $\:$V\hfil $\:$W\hfil $\:$X\hfil $\:$Y\hfil
$\:$Z\hfil $\:$a\hfil $\:$b\hfil $\:$c\hfil $\:$d\hfil
$\:$e\hfil $\:$f\hfil $\:$g\hfil $\:$h\hfil $\:$i\hfil
$\:$j\hfil $\:$k\hfil $\:$l\hfil $\:$m\hfil $\:$n\hfil
$\:$o\hfil $\:$p\hfil $\:$q\hfil $\:$r\hfil $\:$s\hfil
$\:$t\hfil $\:$u\hfil $\:$v\hfil $\:$w\hfil $\:$x\hfil
$\:$y\hfil $\:$z\hfil $\:$!\hfil $\:$"\hfil $\:$\#\hfil
$\:${\scriptsize\$}\hfil $\:$\%\hfil $\:$\&\hfil $\:$'\hfil $\:$(\hfil
$\:$)\hfil $\:$*\hfil $\:$-\hfil $\:$/\hfil $\:$:\hfil
$\:$;\hfil $\:$$<$\hfil $\:$=\hfil $\:$$>$\hfil $\:$?\hfil
$\:$@\hfil $\:$[\hfil $\:${\scriptsize$\backslash$}\hfil $\:$]\hfil
$\:$\^{}\hfil $\:$\_\hfil $\:${\scriptsize$\grave{}$}\hfil
$\:${\scriptsize\{}\hfil
$\:${\scriptsize$|$}\hfil $\:${\scriptsize\}}\hfil
$\:${\scriptsize\~{}}}\ , and then again all these characters prefixed
by `+', then prefixed by `++', etc. There is no upper limit on the
number of characters.

Each block is represented by a string of characters that
represent atoms (without spaces). Blocks are separated by
commas (without spaces). All blocks in a line form a
representation of a hypergraph.  The order of the blocks is
irrelevant---however, we shall often present them
starting with blocks forming the biggest loop to facilitate
their possible drawing.  The line must end with a full stop.
Skipping of characters is allowed.

For 3-atoms-in-a-block lattices, the biggest possible loop
is either $n/2$ (for an even $n$) or $(n-1)/2$ (for an odd $n$),
where $n$ is the number of atoms. To see this this let use
all (or all except one for odd number of atoms) atoms to
form such a loop. If we did not count the first atom in the
first block, each concatenated new block would contribute
with two new atoms to the chain and when we finally close the
chain so as to form the loop, one of the two new atoms
in the last block will coincide with the atom from their first
block, which we did not take into account at the beginning of
our enumeration. That means that all additional blocks will
only connect atoms already making the biggest loop (apart
from the free remaining one in lattices with odd number of
atoms). The more restrictions we impose on a lattice the
smaller the biggest loop will be.

As a functional example, below we present lattices in which
Eq.~(\ref{eq:E-3}) and (\ref{eq:E-4}) fail.

\begin{figure}[htp]
\begin{center}
\includegraphics[width=0.47\textwidth,height=0.12\textwidth]{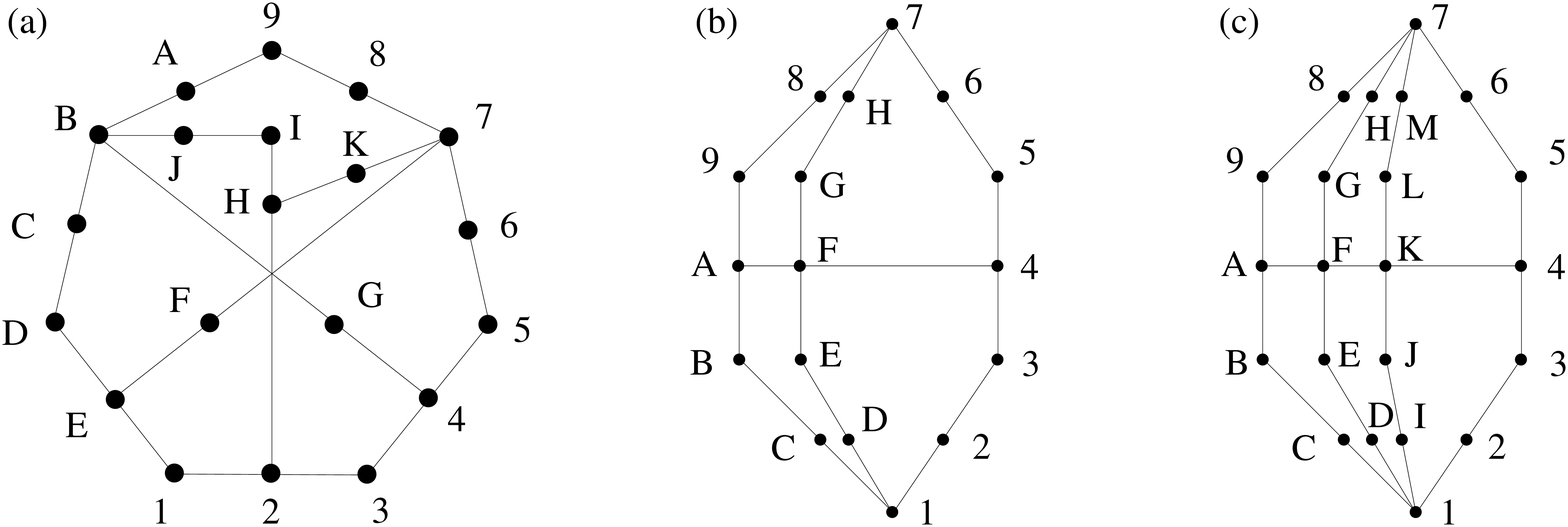}
\end{center}
\caption{(a) OML L42 \cite{mpoa99} in which  Eq.~(\ref{eq:E-3})
fails; (b) $E_3$ \cite{pm-ql-l-hql2} in which Eq.~(\ref{eq:E-3})
fails; (b) $E_4$ \cite{pm-ql-l-hql2} in which Eq.~(\ref{eq:E-4}) fails.}
\label{fig:l42-e3-e4}
\end{figure}

We can now come back to the problem of finding lattices
that would correspond to realistic experiments.
To understand the problem better we shall discuss most known
3D KS lattices that are usually considered to be experimentally
feasible. This will make clear why none of these
KS setups can be experimentally verified and why they are not
``quantum,'' following the idea presented in.~\cite{pavicic-primosten09}

We start with the original KS to show how it can be
represented as an MMP hypergraph in our notation:
{\1123,  345,  567,  789,  9AB,  BC1, \dots , D7z,
\dots, 1z+U.}, as shown in  Fig.~\ref{fig:ks-117}.
The other atoms and blocks can easily be read off from
the figure of the hypergraph.

We give MMP hypergraphs of 4  well-known 3D KS setups below to
enable computer verification of our present and other future
statements on them. Notice that number of atoms and therefore the number of
vectors is 192 and not 117 as commonly assumed. For an explanation of
this discrepancy see the comment on Fig.~\ref{fig:conw-k} in the text below.

\begin{figure*}[htp]
\begin{center}
\includegraphics[width=0.97\textwidth,height=0.17\textwidth]{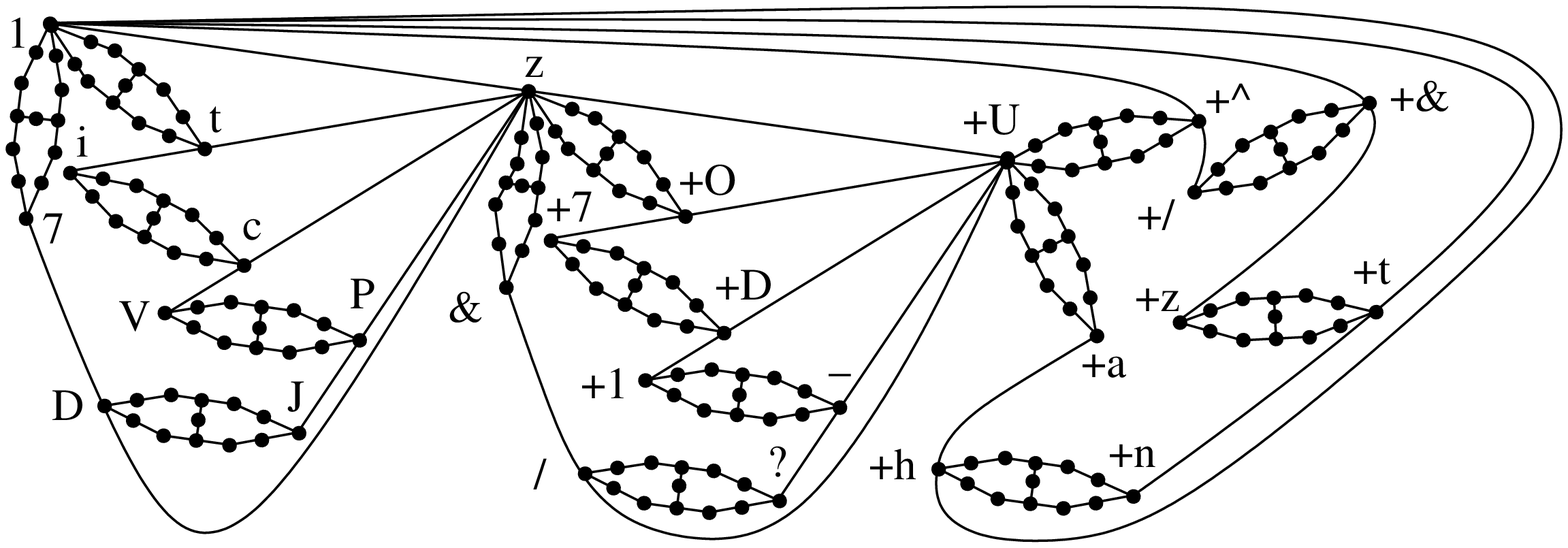}
\end{center}\vskip-10pt
\caption{KS original setup.\cite{koch-speck}. {\em How to} turn their figure
into an MMP is in Ref.~\onlinecite{pmmm03a}.}
\label{fig:ks-117}
\end{figure*}

We establish an OML representation of KS setups as follows. Three mutually
orthogonal directions of spin projections correspond  to three atoms
within a block, say $a,b,c$ in Fig.~\ref{fig:mmp}, because in an OML
$a\le b'$ means $a\perp b$. These three directions also correspond to the
orientation of a device we use to detect spin along them. Keeping one of
the directions fixed, say $a$ in Fig.~\ref{fig:mmp}, means a rotation
of the other two in the plane spanned by $d$ and $e$, what corresponds
to $a\le d'$ and $a\le e'$. As we show below, the aforementioned
Hilbert lattice equations require that the OMLs also have relations
between non-orthogonal atoms and therefore we cannot represent the
considered KS setups by means of Greechie diagrams. Therefore
until we come to that point we shall speak only of MMP hypergraphs.

Asher Peres found another highly symmetrical (in 3D) but much smaller KS
setup.\cite{peres-book} Its MMP hypergraph exhibits symmetry
similar to the MMP hypergraph of the original KS setup as shown in
Fig.~\ref{fig:peres}.
\begin{figure}[htp]
\begin{center}
\includegraphics[width=0.47\textwidth,height=0.12\textwidth]{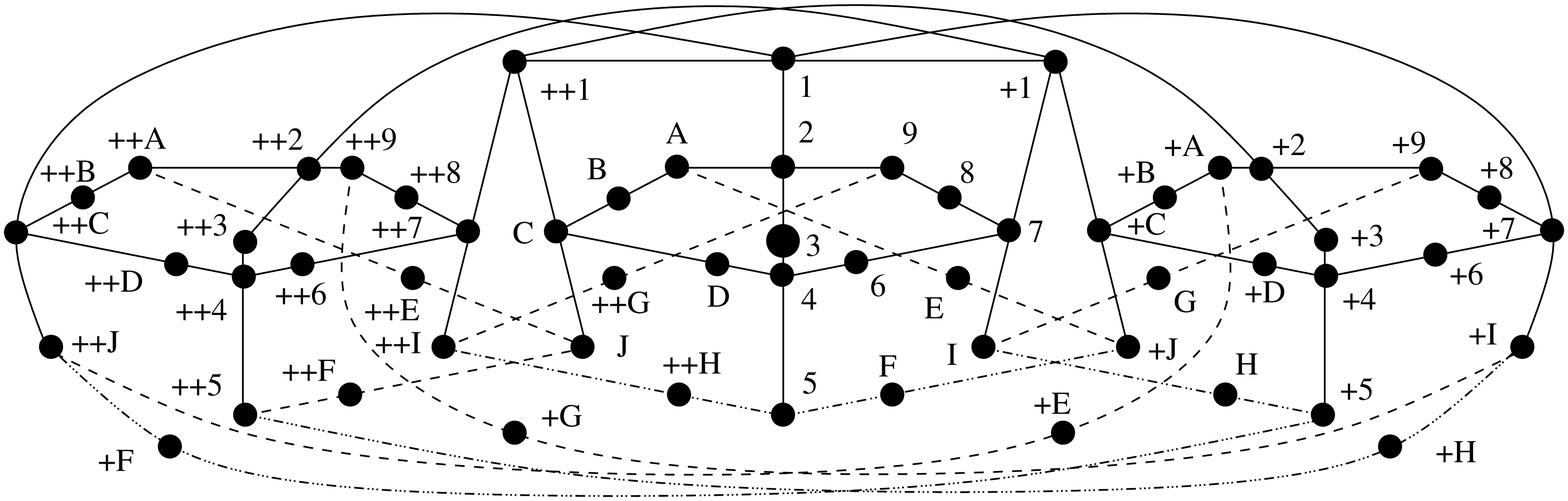}
\end{center}\vskip-15pt
\caption{Peres' KS MMP hypergraph.}
\label{fig:peres}
\end{figure}

The smallest known KS setup was found by Jeffrey Bub.~\cite{bub}
It is shown in Fig.~\ref{fig:bub} l\phantom{.}

\begin{figure}[htp]
\begin{center}
\includegraphics[width=0.47\textwidth,height=0.105\textwidth]{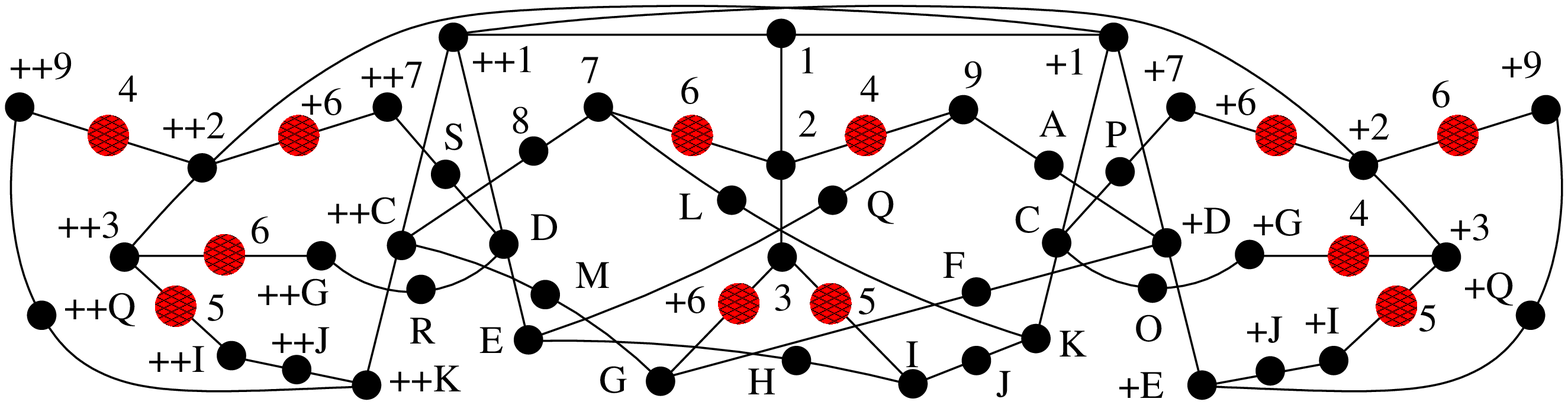}
\end{center}
\caption{Bub's MMP with 49 atoms and 36 blocks. Notice that 12 bigger
dots with a pattern (red online) represent just 4 atoms: 4, 5, 6, and +6.}
\label{fig:bub}
\end{figure}

In Fig.~\ref{fig:conw-k} we show MMP hypergraph of the Conway-Kochen
KS setup.\cite{bub} It reads: {\1123,  249,  267,
  9A+D,  +1CK,  ++1DE,  9QE,  35I,  3+6G,
  EHI,  IJK,  CP+7,  +1+D+E,  CO+G,  DN++7,
  DW++G,  ++GRS,  +7+V+T,  S1+T,  ++7TU,
  1U+S,  +26+9,  +2+6+7,  ++1+2+3,  +S+W+9,
  +S+R+G,  +34+G,  +35+I,  +T+U+I,  +I+J+E,
  +9+Q+E,  ++3++2+1,  ++2+6++7,  ++36++G,
++94++2, ++35++I, 1+1++1.}
It was considered to be the smallest known KS setup, but it
turned out that we cannot remove atoms 7, G, Q, and others that do
not share two or more blocks because they represent one of the three
orientations of the spin projections.\cite{pmmm03a,larsson} Hence,
it has 51 and not 31 vectors as originally assumed. This holds for all
considered KS setups. Thus, Peres' and Bub's setups contain 57 and
49 vectors and not both 33 as commonly assumed.

\begin{figure}[htp]
\begin{center}
\includegraphics[width=0.47\textwidth,height=0.105\textwidth]{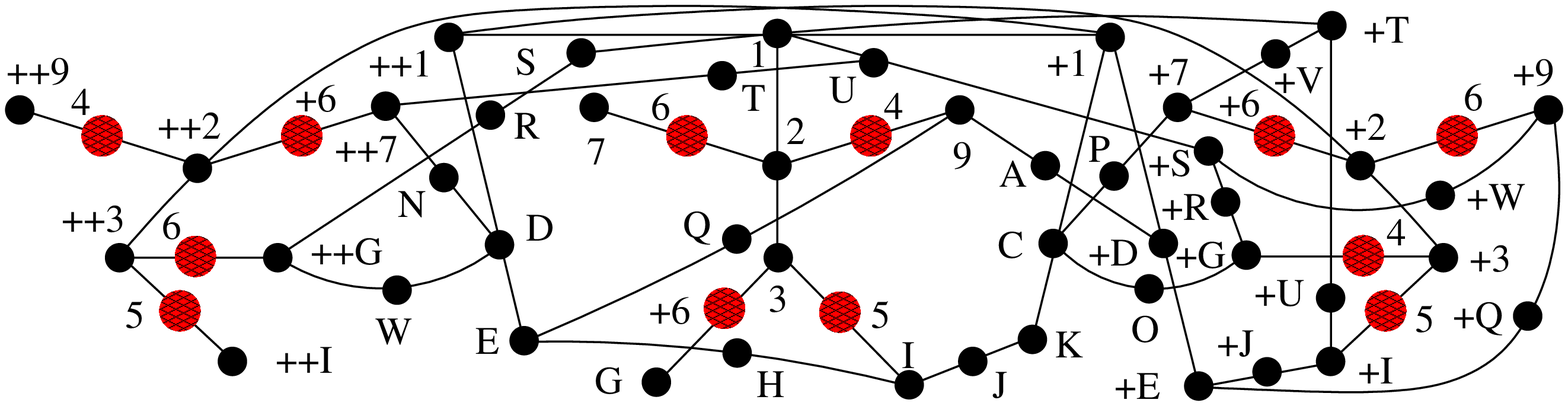}
\end{center}
\caption{Conway-Kochen's MMP. Notice that we cannot drop
blocks containing ++9, ++I, 7, and G because atoms
4, 5, 6, and +6, from them also share two other blocks each.
Why we cannot drop atoms 7, G, H, J, etc., is explained in the
text.}
\label{fig:conw-k}
\end{figure}

Our program {\tt vectorfind} gives possible values of the vectors
corresponding to atoms belonging to orthogonal triples of any of the
above MMPs as explained in Ref.~\onlinecite{pmmm03a}.
Using our program {\tt states},\cite{mpoa99} we can easily verify
that all the above MMPs interpreted as lattices, even Hasse and
Greechie diagrams, admit a strong set of states, and
using our program {\tt latticeg},\cite{mpoa99} we can prove
that they all really are OMLs (by confirming that Eq.~(\ref{eq:oml})
is satisfied by all of them) and that they all admit Mayet vector states
characterized by Eqs.~(\ref{eq:E-3}) and (\ref{eq:E-4}) 
(by verifying that they pass in them).

On the other hand, using {\tt latticeg} we can also show that if we
interpret MMP hypergraph as Greechie diagrams, none
of the considered lattices is modular since the modular law given
by Eq.~(\ref{eq:mod}) fails in each of them. This might come as a
surprise since Birkhoff and von Neumann \cite{birk-v-neum}
proved that a finite-dimensional lattice has to be modular.
However, it turns out that this is because Greechie diagrams
cannot describe relations between nonorthogonal vectors and
planes they span.

To understand this better we exhaustively generated Greechie
lattices with up to 16 blocks and then filtered them all for
modularity given by Eq.~(\ref{eq:mod}). For each
number of blocks we find only one modular
lattice---the biggest one has 33 atoms and 16 blocks.
They all have star-like shape as shown in
Fig.~\ref{fig:mod-star}(a). In the figure we show the first four:
{\1123}, \ {\1123,145}, \ {\1123,145,167}, and
{\1123,145,167,189}---over each other---with
vectors {\1\{\{0,0,1\}\{0,1,0\}\{1,0,0\}\}, \
\{\{0,0,1\}\{1,-2,0\}\{2,1,0\}\}, \
\{\{0,0,1\}\{1,-1,0\}\{1,1,0\}\}, \
\{\{0,0,1\}\{1,2,0\}\{2,-1,0\}\}} (over each other).
And for all those lattices up to 16 blocks we
generated there is always only one such star-like modular
lattice among them. They all admit strong sets of states, but
because of their planar distribution, they cannot describe spin
vectors in a realistic spin space.

\begin{figure}[htp]
\begin{center}
\includegraphics[width=0.47\textwidth,height=0.2\textwidth]{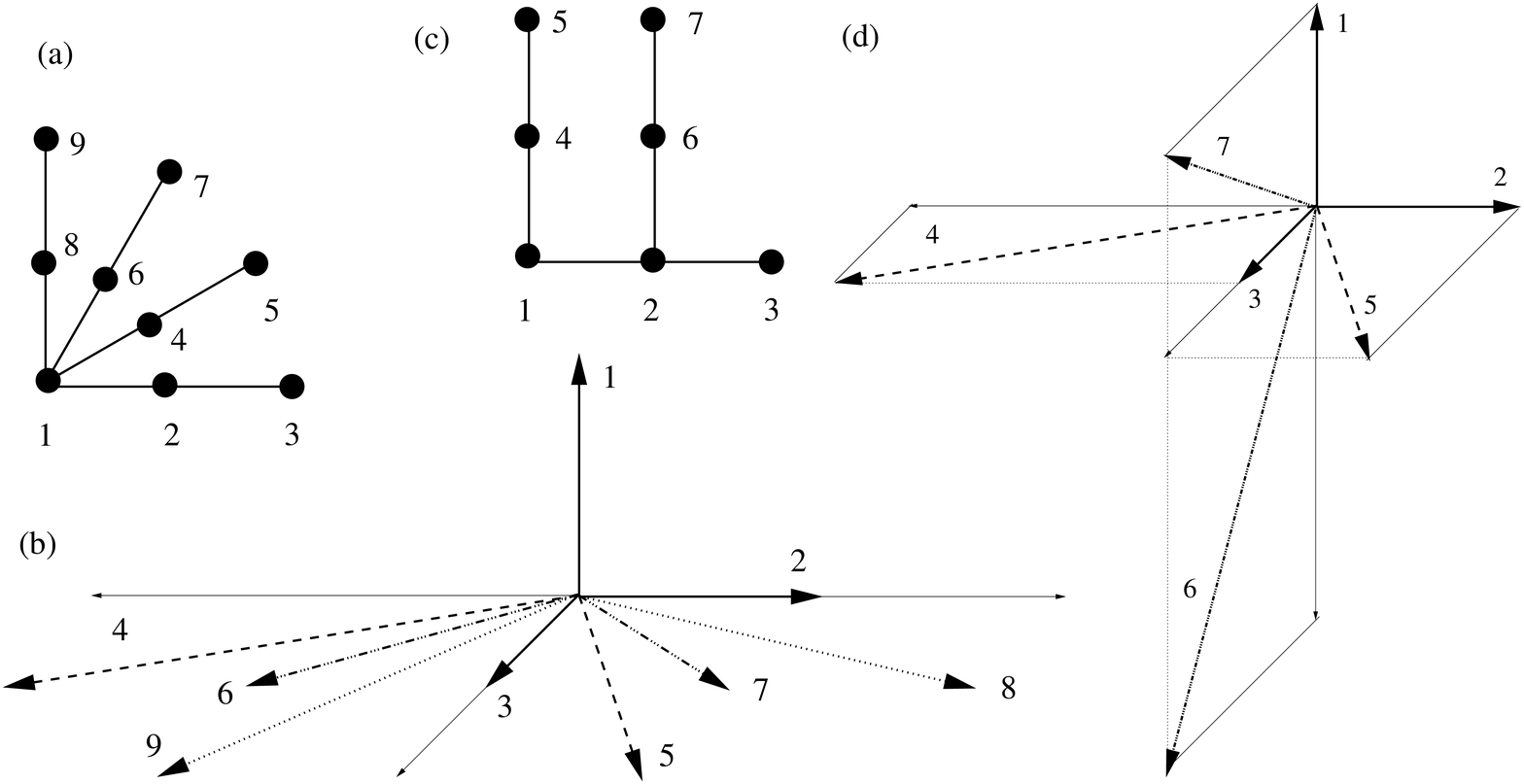}
\end{center}\vskip-10pt
\caption{(a) Greechie diagrams that correspond to
modular lattices with 1, 2, 3 and 4 blocks (over each other);
(b) their vectors triples; (c) the smallest non-modular
lattice; (d) its vector triples.}
\label{fig:mod-star}
\end{figure}

For a comparison, in Fig.~\ref{fig:mod-star}~(c), we show the
smallest OML {\1123,145,267}, with vectors
{\1\{\{\{0,0,1\}\{0,1,0\}\{1,0,0\}\},
  \{\{0,0,1\}\{1,-2,0\}\{2,1,0\}\}, \{\{0,1,0\}\{1,0,-2\}\{2,0,1\}\}\}}
shown in Fig.~\ref{fig:mod-star}~(d), which
allows a ``3D'' rotation that can correspond
to a more complex experimental setup than the ``2D''
rotations given in Figs.~\ref{fig:mod-star}~(a) and (b).
This means that Greechie/Hasse diagrams cannot represent even
the simplest experiment where we let a particle pass
successive magnetic fields, i.e., successive Stern-Gerlach
devices, mutually rotated along different axes by means of Euler
angles.

The same is true of the generalized orthoarguesian equations $n$OA given
by Theorem \ref{th:hs-noa} and Eq.~(\ref{eq:hs-noa}) in a Hilbert
space and by Theorem \ref{th:noa} and Eq.~(\ref{eq:noa}) in a
Hilbert lattice. If these equations failed in a sub-lattice,
they would fail in the lattice as well. And the point here is that
smallest orthoarguesian equation 3OA---and therefore all $n$OA with
$n>3$---fail in almost all known KS Greechie diagrams.  Peres' fails
$n$OA for $n=7$. Again, this means that we cannot represent KS
setups with the help of Greechie diagrams.

The details are as follows. We consider Bub's KS setup. To be able
to apply our program {\tt vectorfind} for finding the vector
components of Bub's setup shown in Fig.~\ref{fig:bub}, we have to
write down its MMP representation without gaps in letters. So, we
have {\1123,\dots,DFH,\dots}, where we present only those
Greechie/Hasse diagrams atoms in which 3OA failed. Their
Hilbert space vectors are: {\11=\{0,0,1\}}, {\12=\{1,0,0\}},
{\1F=\{1,-2,-1\}}, and {\1D=\{1,1,-1\}}.

In a Hilbert space representation, Bub's KS setup does pass 3OA.
Let us consider 3OA in the following form \
\begin{align}
&a\perp b  \quad\&\quad q\perp n\quad\notag\\
       & \Rightarrow
(a\cup b)\cap (q\cup n)\le b\cup(a\cap
(q\cup((a\cup q)
\cap (b\cup n)))).\notag
\end{align}
In 3-dim Euclidean space, all subspaces are closed (they are lines,
planes, or the whole space), so $a\cup b=a+b$, i.e., subspace join
and subspace sum are the same. Thus, converting joins in the previous
equation to subspace sums and using the orthogonality we get:
\begin{align}
a\perp b & \quad\&\quad q\perp n
\Rightarrow (a+b)\cap(q+n)\notag\\
&\le b+(a\cap(q+((a+q)\cap (b+n)))).\label{eq:bubs-proof}
\end{align}

Now, using the subspaces determined by the afore mentioned vectors
and their spans in a Hilbert space we can easily check that
Bub's representation pass 3OA. For instance, vectors {\11},
{\12}, {\1F}, and {\1D}, determine subspaces {\1\{0,0,$\alpha$\}},
{\1\{$\beta$,0,0\}}, {\1\{$\gamma$,-2$\gamma$,-$\gamma$\}}, and
{\1\{$\delta$,$\delta$,-$\delta$\}}, with arbitrary coefficients
$\alpha,\dots$. They represent lines in both 3-dim Hilbert
space and 3-dim Euclidean space. {\1\{0,0,$\alpha$\}}+{\1\{$\beta$,0,0\}}=
{\1\{$\beta$,0,$\alpha$\}} is a plane spanned by {\11} and {\12}, etc.
We show a verification of Eq.~(\ref{eq:bubs-proof}) in
Fig.~\ref{fig:bub-proof}.

\begin{figure}[htp]
\begin{center}
\includegraphics[width=0.47\textwidth,height=0.4\textwidth]{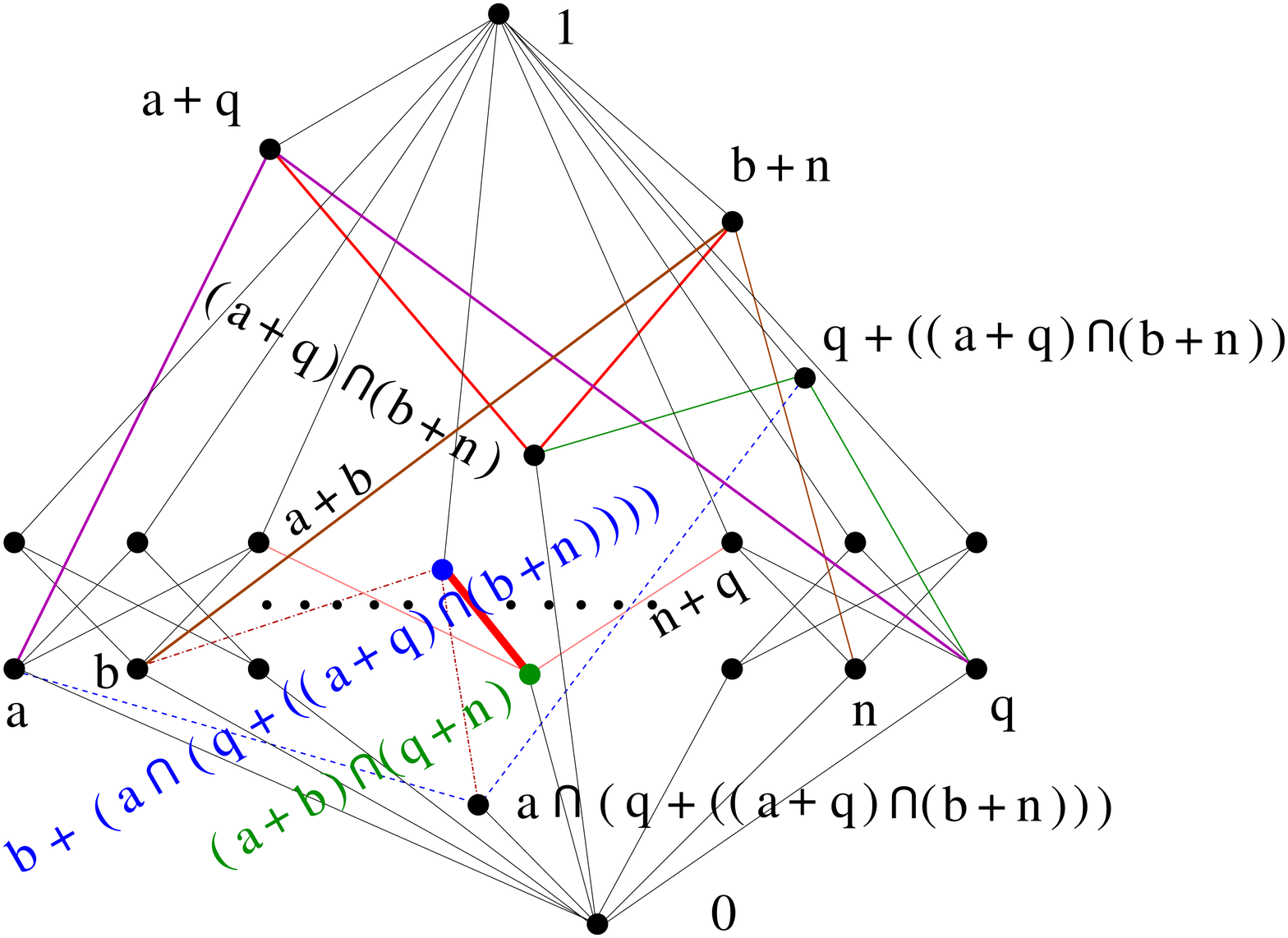}
\end{center}
\caption{A new kind of lattice (MMPL) in which Bub's setup passes 3OA.
The inequality relation in Eq.~(\ref{eq:bubs-proof}) is represented
by the thick red line.}
\label{fig:bub-proof}
\end{figure}

Such a lattice---we call it MMPL---can be used for a Hilbert
lattice representation of
a Hilbert space setup by the following procedure. Whenever we check
an equation and we need either a plane formed as a span of two
existing lines or a line formed as an intersection of two
existing planes, we just add them to the basic
Greechie/Hasse diagram that describes the triples of orthogonal
spin vectors. However, details of such a construction are not
within the scope of the present paper. We will elaborate on it
in a forthcoming publication and here we just give a constructive
definition.

\begin{definition}\label{def:mmpl} An {\rm MMPL} is a lattice of
setup in which we explicitly state:
\begin{enumerate}
\item all orthogonality relation required by
the setup (spins within it);
\item only those non-orthogonal relations that are required by
equations and conditions that lattice atoms of a particular
setup have to satisfy for at least one set of subspace
(vector) components.
\end{enumerate}
\end{definition}

So, the most general MMPL would be a lattice that would
contain all possible atoms corresponding to all possible
Hilbert space subspaces allowed by all possible Hilbert
space conditions and equations. But our primary goal of
considering MMPLs is to enable our algorithms to find
minimal lattices for a particular setup which would generate
just one or just a desired set of vector component values
for orientation of spins and devices that would handle
these spins.

Next, the superposition condition given by Eq.~(\ref{eq:superp})
fails in all considered KS OMLs. However, the superposition condition
is a quantified expression that involves an existential
quantifier, so it is possible that it passes in an enlarged lattice
even though it fails in the original one. For instance,
Eq.~(\ref{eq:superp}) fails in any five block loop but passes in
the 36-36 OML shown in Fig.~\ref{fig:36-36}, which contains five
block loops. That means that we may be able to enlarge the above KS
OMLs so as to admit superposition. Of course, a first-order statement
containing existential quantifiers (when expressed in prenex normal
form) that holds in a lattice need not hold in a subalgebra of the
lattice.  As a trivial example, the statement ``There exist 16
elements'' is true for a 16-element lattice but false for a smaller
subalgebra.

\section{\label{sec:no-god}Lattices that Admit Almost No
Hilbert Lattice Equations}

There are a number of OMLs that admit a full set of
states but do not admit a strong set of states and also
those that admit a strong set of states (and therefore also
a full set of states) but violate equations that must
hold in any Hilbert lattice. Using algorithms developed
in Ref.~\onlinecite{mpoa99,pm-ql-l-hql2} we can easily generate
such lattices. For instance, a lattice with
13 atoms (one dimensional Hilbert space subspaces) and
7 blocks (connected orthogonal triples of one dimensional
Hilbert space subspaces) shown in Fig.~\ref{fig:examples}\
(a) does admit a strong and therefore also a full set of
states but violates all orthoarguesian equations.
Any Hilbert lattice admits a strong and therefore a full
set of states, and the orthoarguesian equations hold in any
 Hilbert lattice.\cite{mpoa99,pm-ql-l-hql2}

\begin{figure}[htp]
\begin{center}
\includegraphics[width=0.47\textwidth,height=0.125\textwidth]{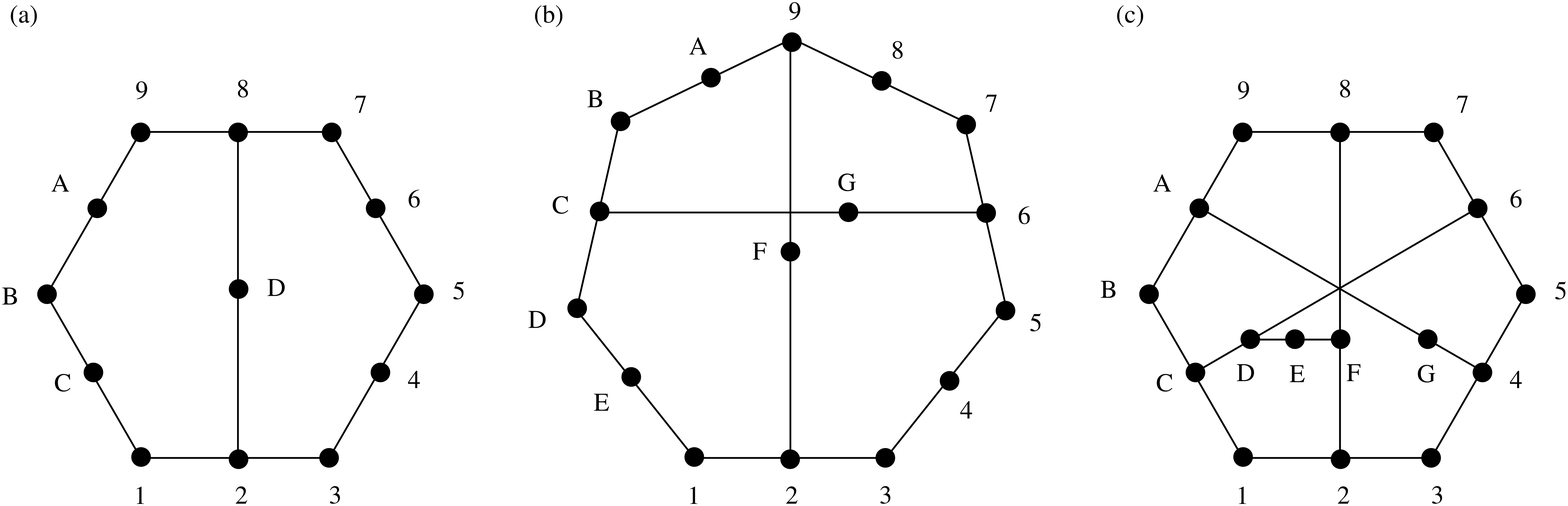}
\end{center}\vskip-20pt
\caption{(a) A 13-7 OML that admits a strong set of states
but violates Hilbert-space orthoarguesian equations;
(b) A 16-9 OML that does not admit a strong set of states
but satisfies orthoarguesian equations;
(c) A 16-10 OML that admits neither a strong nor a full set
of states and violates all orthoarguesian equations.}
\label{fig:examples}
\end{figure}

On the other hand, the 16-9 OML in Fig.~\ref{fig:examples}\ (b)
satisfies orthoarguesian equations and admits a full
set of states but does not admit a strong set of states,
L42 from Fig.~\ref{fig:l42-e3-e4}$\>$(a) satisfies orthoarguesian
equations and admits a strong set of states, but does not
admit Mayet vector state Eq.~(\ref{eq:E-3}),
while 16-10 OML in Fig.~\ref{fig:examples}\ (c)
neither admits a strong (and therefore also not a full) set
of states nor satisfies the orthoarguesian equations.
All these OMLs and many more provided in 
Refs.~\onlinecite{mpoa99,pm-ql-l-hql2} are examples semi-quantum 
lattices. Yet other examples are provided by lattices that satisfy 
the Godowski equations (corresponding to strong sets of states)
of lower order but violate those of higher orders.~\cite{mpoa99}
While all OMLs admitting strong sets of states satisfy all
Godowski equations, there are examples showing the converse isn't
true.~\cite[Fig.~10, p.~780]{pm-ql-l-hql2}

Such examples can be exhaustively generated, but no common structural
feature has been recognized so far. To be more precise,
features and general rules for generation of infinite classes of
lattices that admit a strong set of states---Godowski
equations,\cite{godow,mayet85,mayet86,mayet06,pm-ql-l-hql2}
satisfy the orthoarguesian properties---$n$OA
equations,\cite{mpoa99,pm-ql-l-hql2} and a
class of lattices that admit real Hilbert-space-valued
states---$E_n$ equations, \cite{ mayet06-hql2,pm-ql-l-hql2} have
all been   discovered, but the rule for generating all lattices that
lack all these properties has not been found.
Since we still do not have a single example of a complete
realistic lattice for $n\ge 3$, it would be important
to find a class of lattices that would narrow down the
search for a complete lattice description of
Hilbert space. Therefore, in Sec.~\ref{sec:one-state-p} we consider
a class of OMLs that admit a field over which a Hilbert
space is defined but neither a strong set of states nor any of the
Hilbert space algebraic properties.

We stress here that an OML admitting a strong set of states will satisfy
the Godowski equations.
\cite{godow,mayet85,mayet86,mayet06,pm-ql-l-hql2,mp-alg-hilb-eq-09} Thus
OMLs that violate Godowski equations do not admit strong sets of
states.  Moreover, most likely
they cannot be enlarged to admit such a set in
order to satisfy these equations---similarly to what we have with the
modular and orthoarguesian equations in Sec.~\ref{sec:represent}.

\section{\label{sec:one-state}MMP Hypergraphs  with
Equal Number of Vertices and Edges Generated from
Cubic Bipartite Graphs}

To avoid confusion with vertices and edges in (bipartite)
graphs in this section (and only in this section) we use term atom
for a vertex of a MMP hypergraph and block for an edge of an MMP
hypergraph. Since later we shall consider the corresponding lattices
anyway, this terminology is not inconsistent.
Here we describe the exhaustive computation of MMP hypergraphs
with equal numbers of atoms and blocks having 3 atoms in
each block and 3 blocks containing each atom.
This special case allows exploitation of a connection with
graph theory in order to considerably speed up the generation
compared to our earlier methods~\cite{bdm-ndm-mp-1,pmmm03a}.

We begin by representing MMP hypergraphs as graphs with two
types of vertex.  An atom $a$ is converted to a white vertex $A$,
and a block $b$ to a black vertex~$B$.  If atom~$a$ lies in
block~$b$, then vertex~$A$ is joined by an edge to vertex $B$.
In graph theory terminology, the resulting graph is \textit{cubic\/}
(each atom is in~3 blocks and each block has~3 atoms), and
\textit{bipartite\/} (edges have ends of different color).
In Ref.~\onlinecite[Sec.~5(x)]{pmmm03a} we have shown that
3-dim KS MMP hypergraphs have no loops of length less than 5.
This corresponds to the graph having \textit{girth at least~10} (i.e.,
having no cycles of length less than~10).
Apart from taking the dual MMP hypergraph, which corresponds to exchanging
the colors of the vertices, isomorphism of the MMPs corresponds
to isomorphism of the graphs.

For definiteness, we consider the case of 41 atoms and
41 blocks.  That is, we seek 82-vertex cubic bipartite graphs
of girth at least~10.
The method used is an extension of one used in the non-bipartite
case by McKay et al.~\cite{mmn98}.

We begin with 41 white vertices and 41 black vertices, plus the
61 edges at distance at most 4 from an arbitrary fixed edge.
These 61 edges form a tree, since otherwise there would be cycles
of length less than~10.
This starting configuration is shown in Fig.~\ref{fig:bipart},
with dashed lines indicating the places available for
extra edges.

\begin{figure}[htp]
\begin{center}
\includegraphics[width=0.47\textwidth]{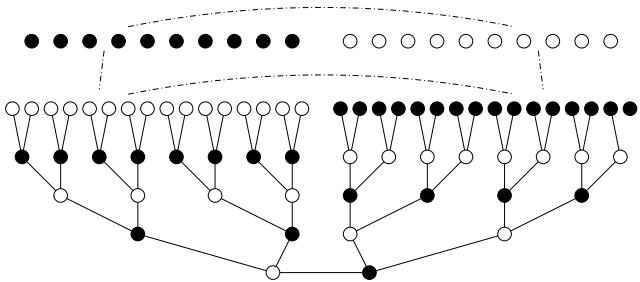}
\end{center}
\caption{Starting configuration for generation of 41-41 MMPs}
\label{fig:bipart}
\end{figure}

The task is now to add 62 extra edges so that each vertex has
3 edges and there are no short cycles.
This is a non-trivial task since there are 676 places where an
edge may potentially be placed, but fortunately many of the
possibilities are equivalent.
We proceed using a backtrack search
together with some mechanisms for isomorphism rejection.
The backtrack search looks for an incomplete vertex whose
set of potential neighbours is as small as possible, then
recursively tries each of them.

Isomorphism rejection is achieved by two methods which are
described in detail in Ref.~\onlinecite{exoo10}.
First, the starting configuration has a large group of
symmetries and we avoid trying more than one possibility that
is equivalent under those symmetries.
This can be done without explicit isomorphism testing since the
structure of the starting configuration is rather simple.

Second, when the space of supergraphs of any configuration $C$
has been completely explored, we reject any future configuration
$C'$ that contains $C$ as a subgraph.
This is valid since any cubic graph constructible by adding edges
to $C'$ was previously seen (up to isomorphism) when edges were
added to~$C$.
This technique is too expensive to apply throughout the search,
because subgraph finding is very difficult.
As a compromise, we applied the technique only limited circumstances
with at most 78 edges (the initial 61 edges plus 17 more).
We did this using the graph isomorphism package
{\tt nauty}~\cite{mckay90}.

These isomorph-rejection methods are not complete, so each
isomorphism type of graph was generated a few thousand times.

The complete search on order 41-41 involved about $10^{14}$
separate configurations and took approximately 60 GHz-years.
The computation can be efficiently divided into independent parts
(see~\cite{exoo10} for an explanation), so it was run over a
few weeks on a multi-processor cluster.

\section{\label{sec:one-state-p}Properties of Lattices with
Equal Numbers of Atoms and Blocks}

In this section we consider OMLs that correspond to MMP
hypergraphs we obtain by means of methods presented in
the previous section.

In Ref.~\onlinecite{bdm-ndm-mp-1} we mentioned five 35-35 OMLs
(OMLs with 35 atoms and 35 blocks), eight 38-38s
and gave a graphical representation of the single 36-36
(there is no 37-37). They were obtained by different algorithms
and at the time we were not aware of their properties and
did not yet have tools to analyze them. In ~\cite{pavicic-rmp09}
we wrote down all 35-35s and 38-38s, gave two graphical images of
them and obtained some features of them in a different context.
So, in this section we shall focus on 39-39s, 40-40s, and 41-41s.
In doing so, we will make use of a new way of presenting
MMP hypergraphs, because our previous one becomes
unreadable for so many edges.
We introduce the new way as opposed to the previous one in
Fig.~\ref{fig:36-36}.

\begin{figure}[htp]
\begin{center}
\includegraphics[width=0.23\textwidth]{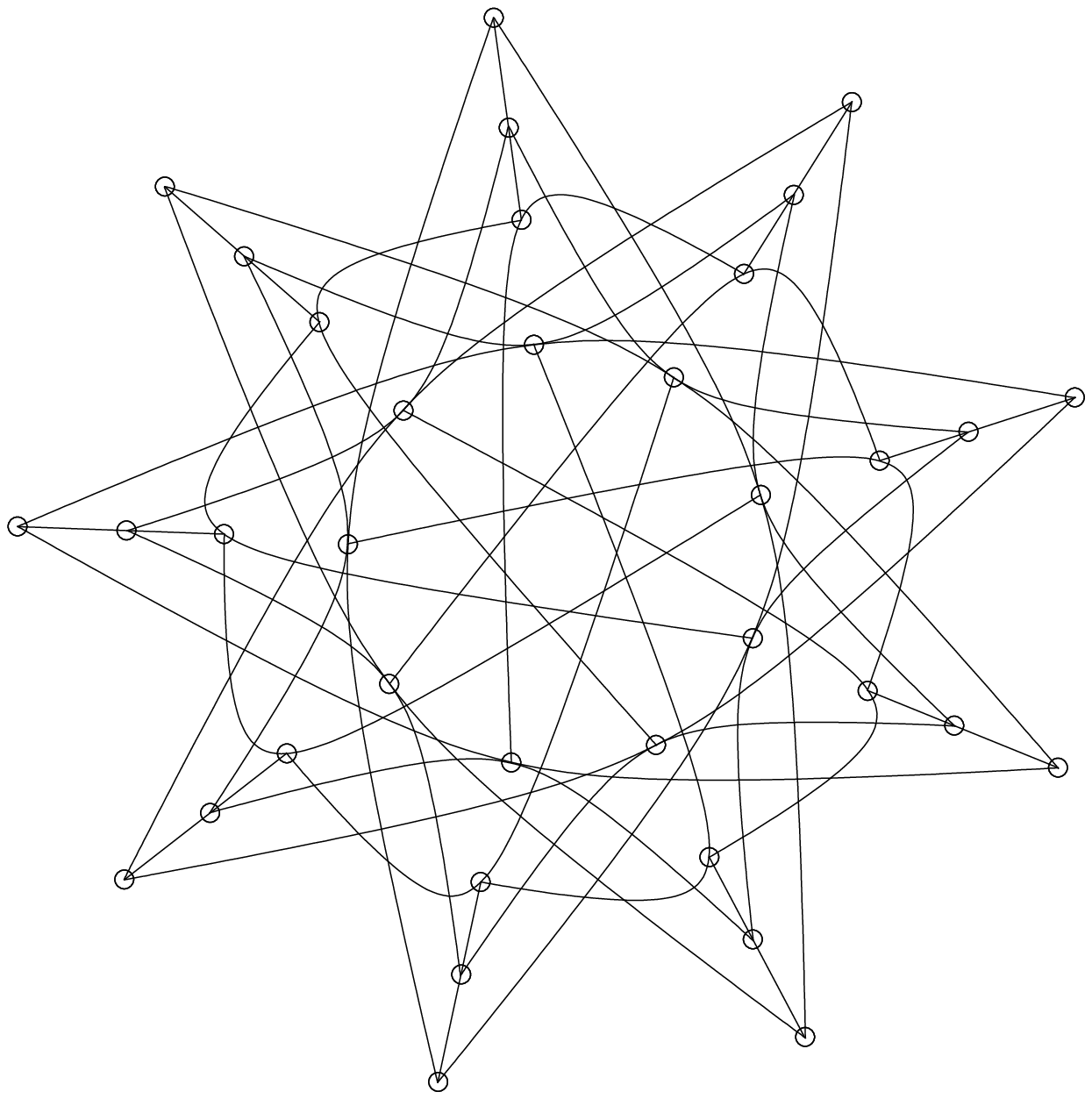}
\includegraphics[width=0.23\textwidth]{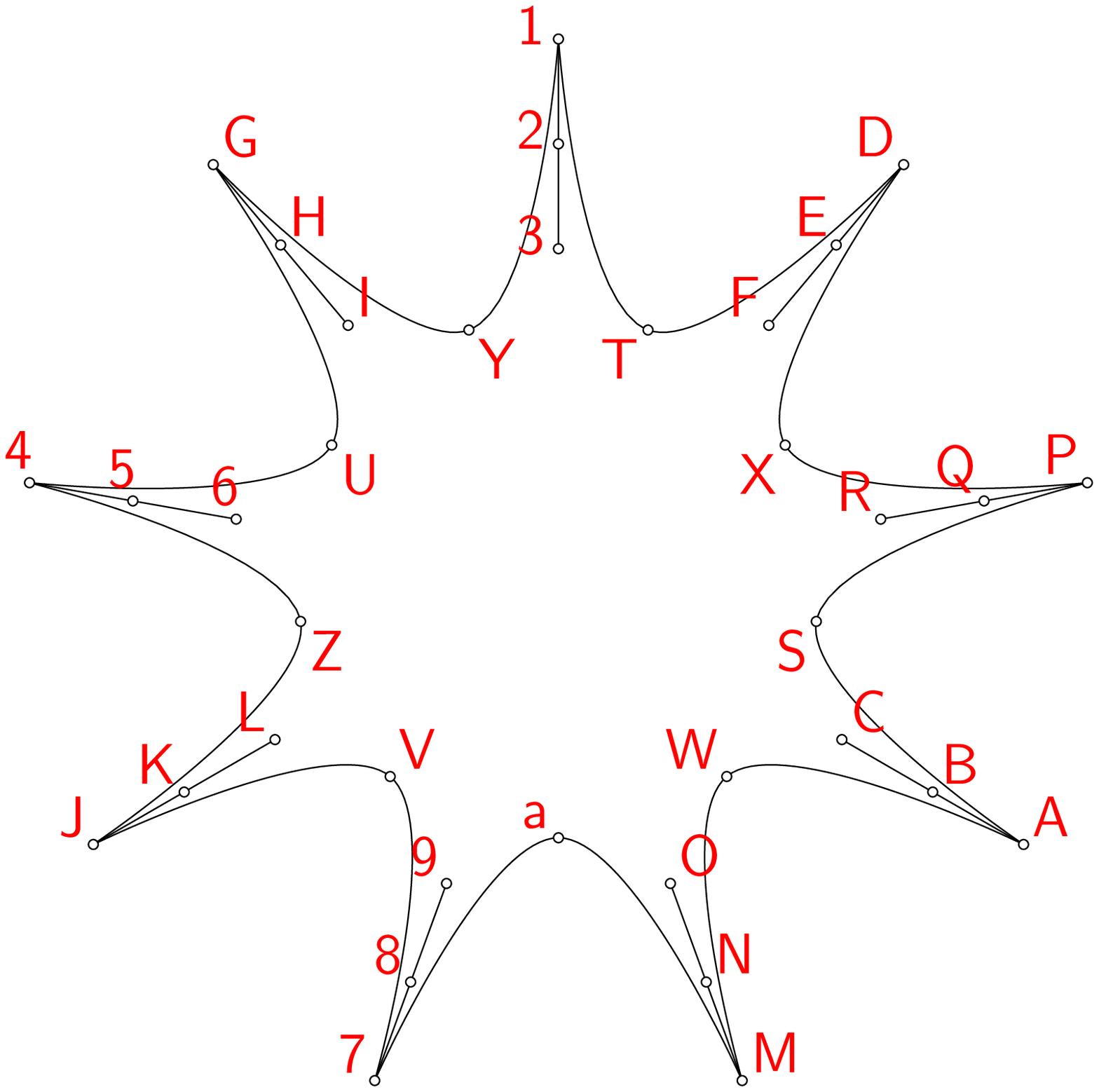}
\includegraphics[width=0.23\textwidth]{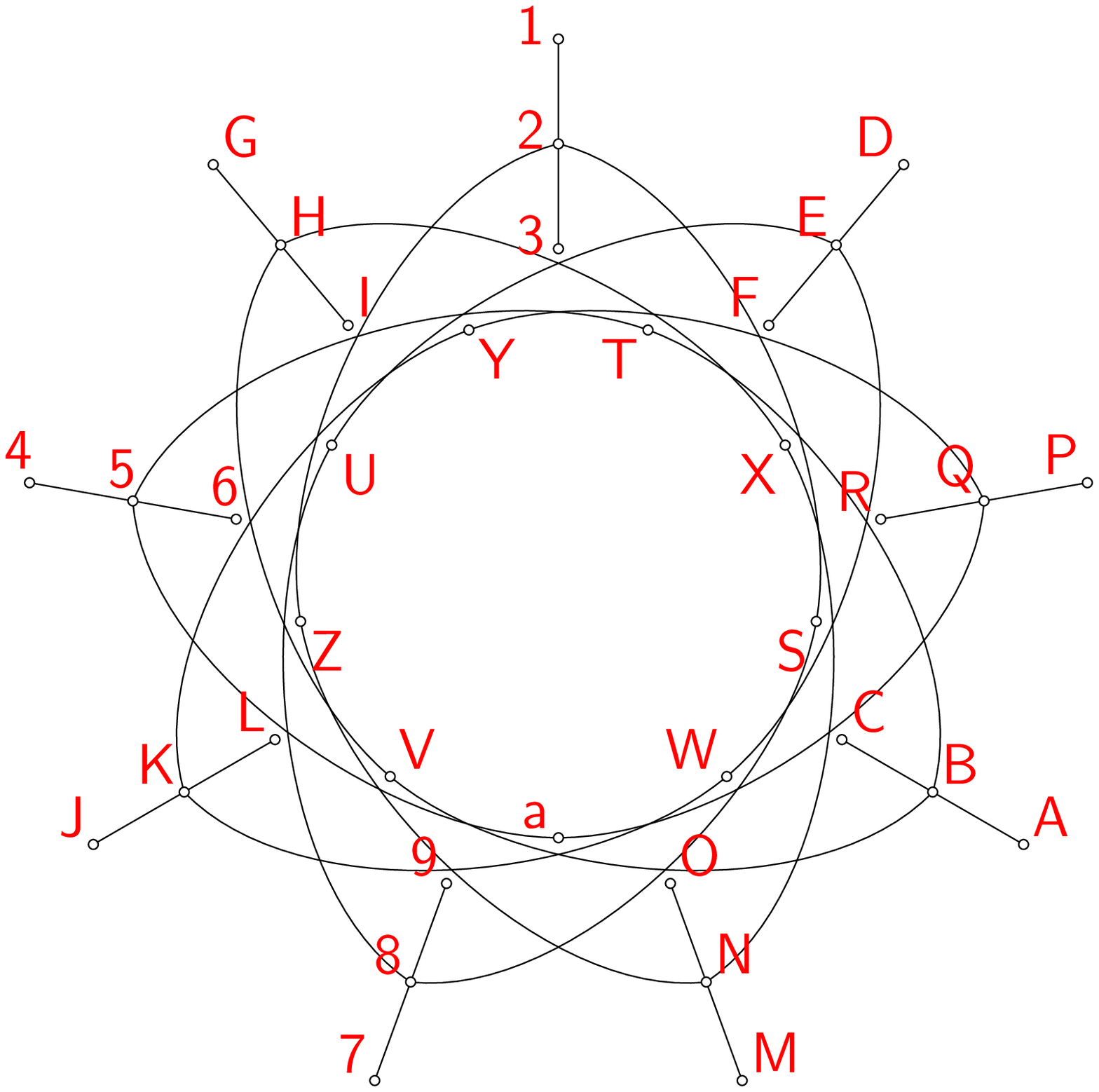}
\includegraphics[width=0.23\textwidth]{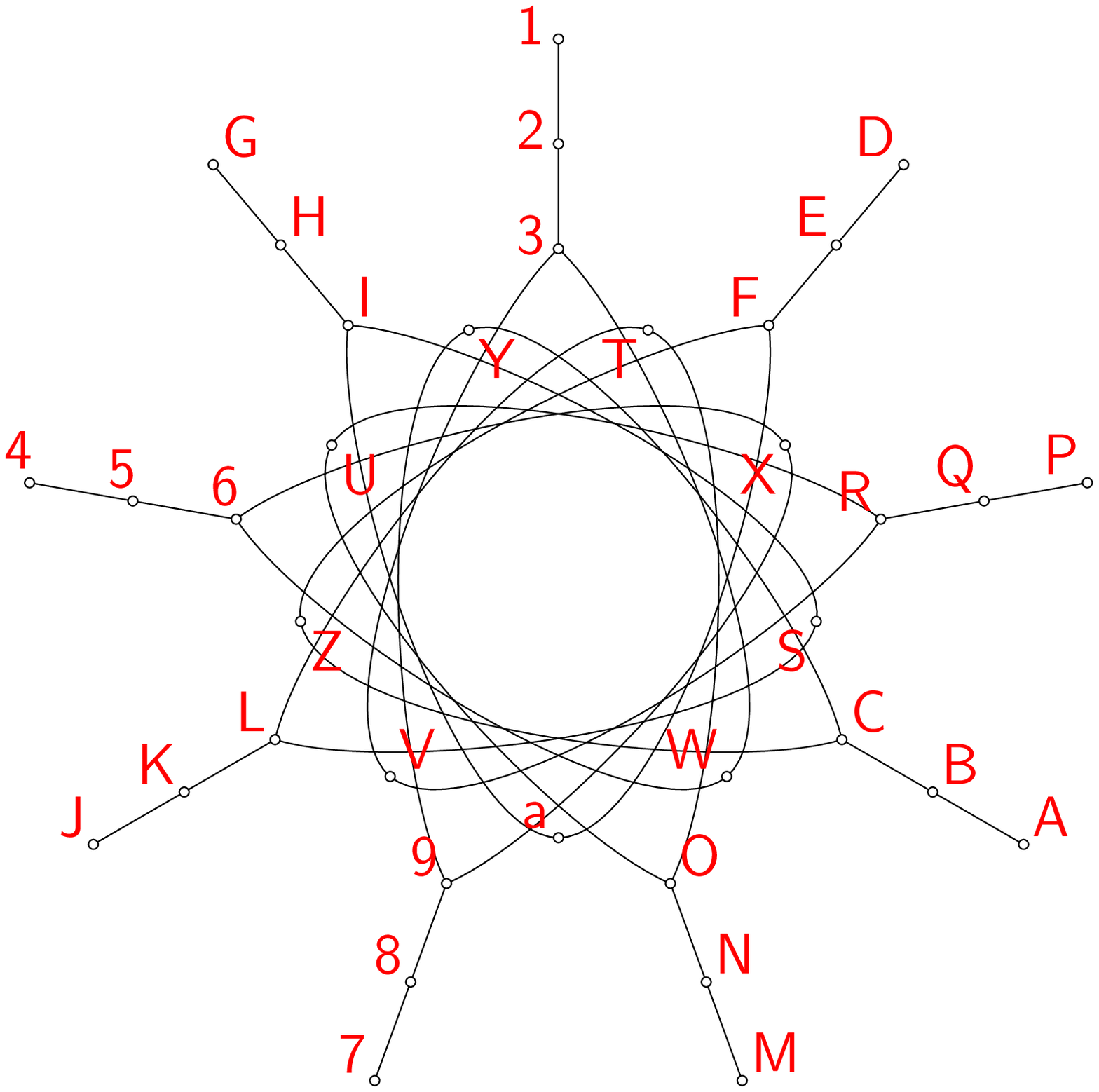}
\end{center}
\caption{36-36 OML that admits exactly one state and
is dual to itself. It is given in the standard compact
representation in the 1st figure and in our separate
cycle representation in the other 3 figures. The
figures are explained in detail in Sec.~\ref{sec:figs}.}
\label{fig:36-36}
\end{figure}

The new presentation is based on a feature of such big
lattices that one can recognize separate cycles of blocks
through a maximal set of vertices that belong to isolated
blocks that mostly do not take part in the cycles.
The terminology ``isolated blocks'' and
``cycles'' will be explained in Sec.~\ref{sec:figs}.
The approach stems from the way the lattice 36-36 is
presented in Fig.~2 from \cite{bdm-ndm-mp-1} which is
here shown as the first figure of Fig.~\ref{fig:36-36}.
We separately present the three cycles in the remaining
three figures and see that we have three separated
closed cycles. In all the other cases below we also
recognize three independent cycles most of which are
closed.

The cycles themselves will allow us to generate new
lattice equations following the procedure developed in
\cite{mp-gen-godowski06-arXiv,pm-ql-l-hql2,mp-alg-hilb-eq-09},
but they do not automatically follow possible geometrical
symmetries of the hypergraphs. In the 36-36 case they
do, but, e.g., they do not exhibit the left right symmetry
of the 35-35 lattice shown in Fig.~\ref{fig:35-35}.
Closed cycle representation does not exhibit any
symmetry.

\begin{figure}[htp]
\begin{center}
\includegraphics[width=0.23\textwidth]{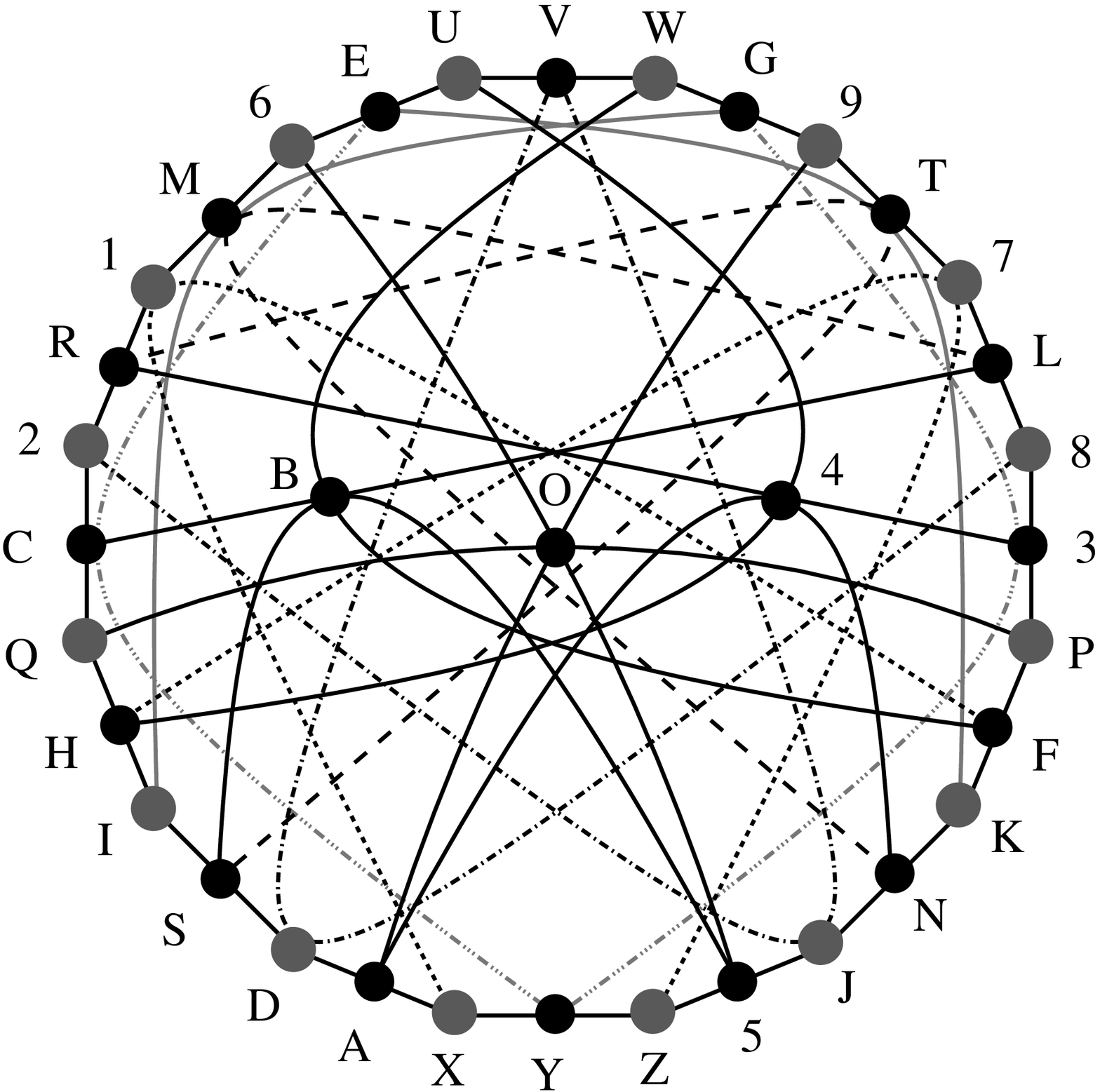}
\includegraphics[width=0.23\textwidth]{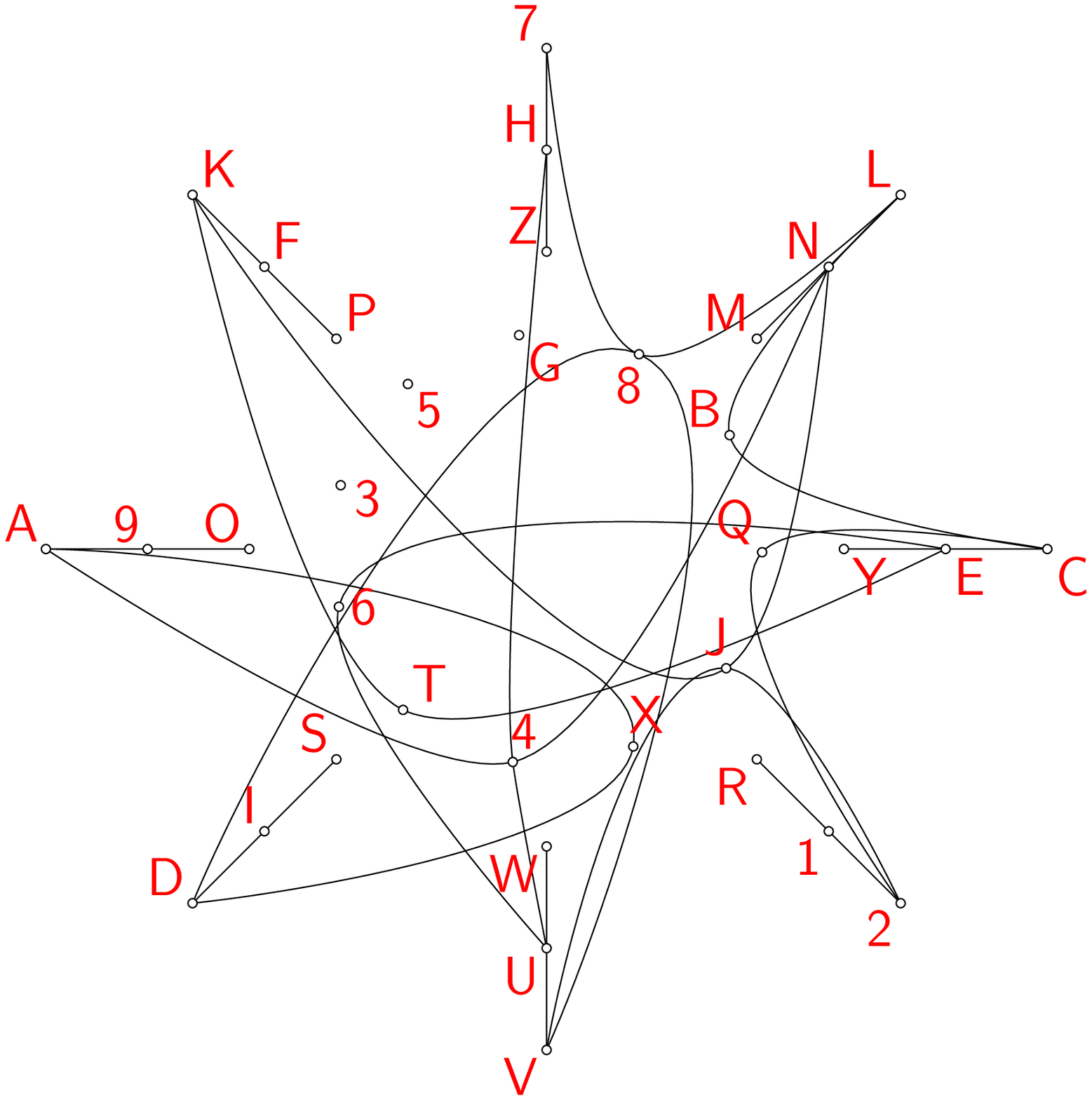}
\includegraphics[width=0.23\textwidth]{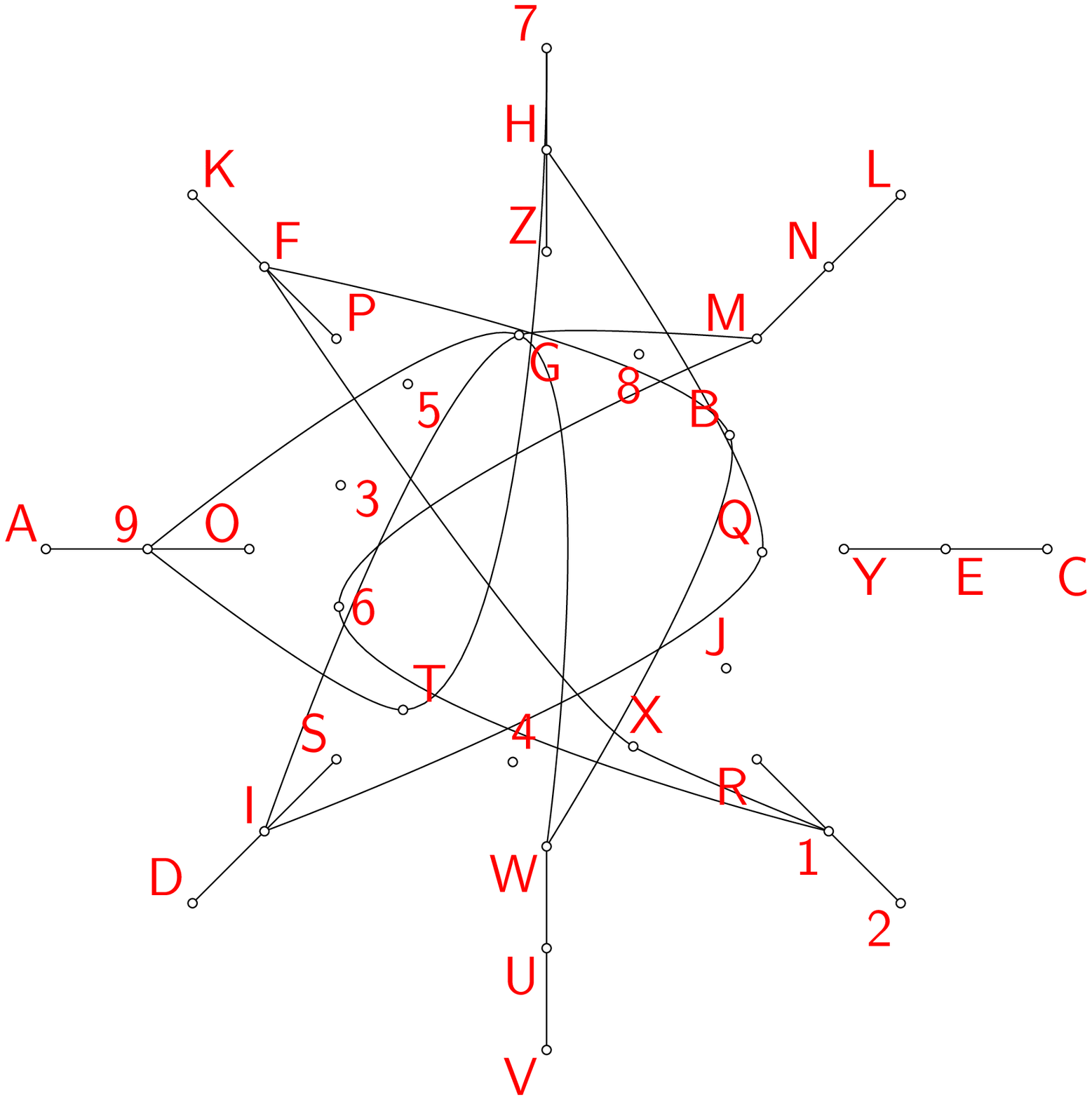}
\includegraphics[width=0.23\textwidth]{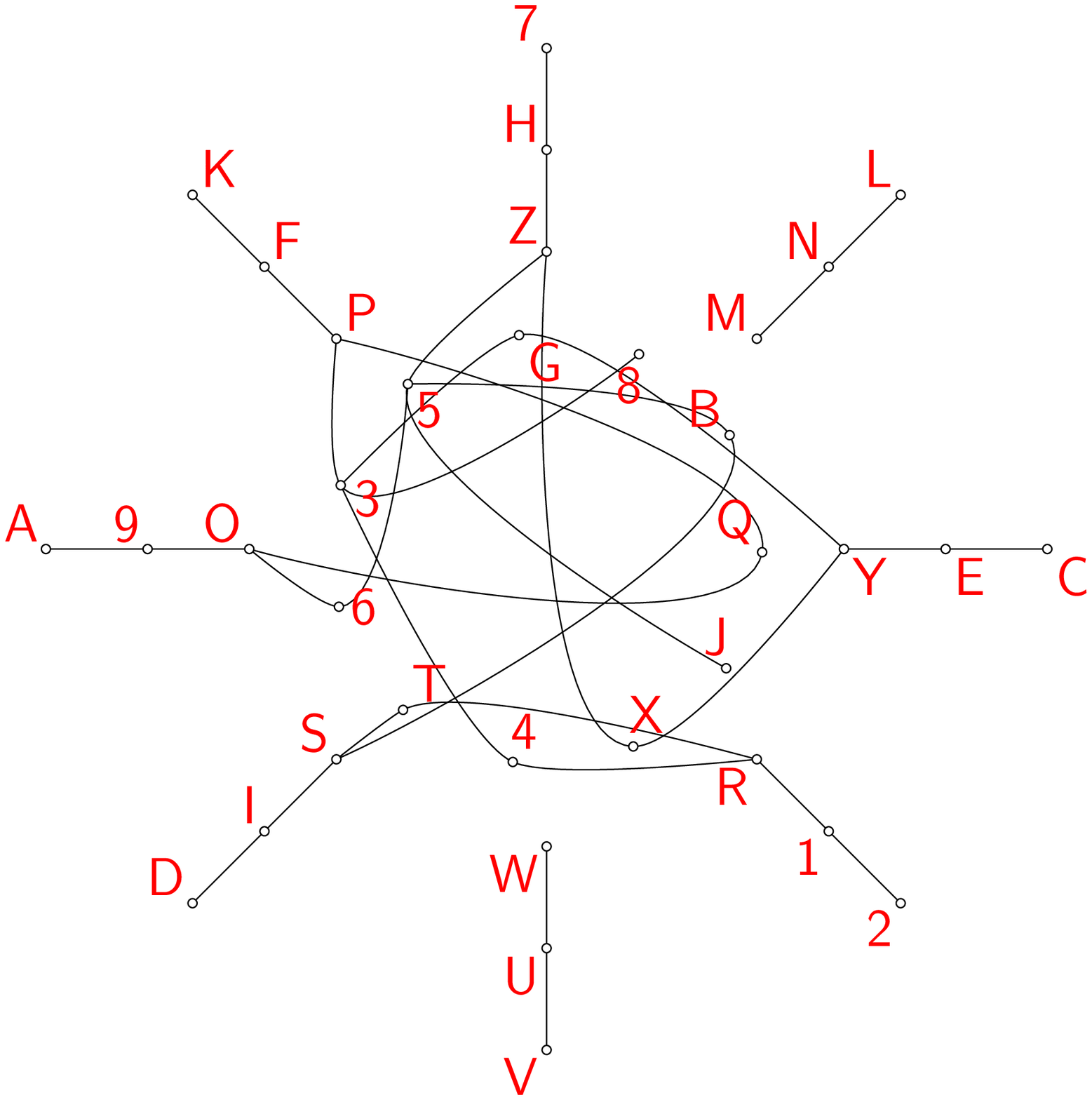}
\end{center}
\caption{1st figure shows a 35-35 lattice presented
by means of its biggest loop, hexadecagon; it exhibits a
left-right symmetry with respect to an axis through
vertices V and Y. Three other figures show the same OML
in the separate cycle representation. They are explained
in detail in Sec.~\ref{sec:figs}.}
\label{fig:35-35}
\end{figure}

There are 11 eleven bipartite graphs with 78 vertices that give 39-39 OMLs.
Nine of them correspond to the MMP hypergraphs that are dual to
themselves---when we exchange their atoms for blocks and vice
versa we obtain OMLs that are isomorphic to the original ones.
\begin{description}
{\baselineskip=8pt
\item[\hbox to 40pt{\3(39-39-00)}]{\1123,\hfil 145,\hfil 167,\hfil 289,\hfil 2AB,\hfil 3CD,\hfil 3EF,\hfil 4GH,\hfil 4IJ,\hfil 5KL,\hfil 5MN,\hfil 6OP,\hfil 6QR,\hfil 7ST,\hfil 7UV,\hfil 8GO,\hfil 8MV,\hfil 9Ia,\hfil 9LT,\hfil AKU,\hfil AQc,\hfil BPb,\hfil BXd,\hfil CGS,\hfil CRY,\hfil DVW,\hfil DPa,\hfil EKO,\hfil EIX,\hfil FTb,\hfil FHQ,\hfil NSc,\hfil UYZ,\hfil MRX,\hfil LWd,\hfil JWc,\hfil NZa,\hfil JYb,\hfil HZd.\hfill}

\item[\hbox to 40pt{\3(39-39-02)}] {\1123,\hfil 145,\hfil 167,\hfil 289,\hfil 2AB,\hfil 3CD,\hfil 3EF,\hfil 4GH,\hfil 4IJ,\hfil 5KL,\hfil 5MN,\hfil 6OP,\hfil 6QR,\hfil 7ST,\hfil 7UV,\hfil 8GO,\hfil 8XY,\hfil 9KQ,\hfil 9IT,\hfil AMP,\hfil AVd,\hfil BRa,\hfil BWb,\hfil CGS,\hfil CQZ,\hfil DIP,\hfil DYb,\hfil EOW,\hfil ELd,\hfil FMX,\hfil FHU,\hfil LSa,\hfil KUb,\hfil NWZ,\hfil JVZ,\hfil NTc,\hfil HRc,\hfil JXa,\hfil Ycd.\hfill}

\item[\hbox to 40pt{\3(39-39-03)}] {\1123,\hfil 145,\hfil 167,\hfil 289,\hfil 2AB,\hfil 3CD,\hfil 3EF,\hfil 4GH,\hfil 4IJ,\hfil 5KL,\hfil 5MN,\hfil 6OP,\hfil 6QR,\hfil 7ST,\hfil 7UV,\hfil 8GO,\hfil 8MU,\hfil 9IT,\hfil 9QY,\hfil AKW,\hfil AHV,\hfil BNP,\hfil BXZ,\hfil CGS,\hfil CQZ,\hfil DJP,\hfil Dac,\hfil EKO,\hfil EIX,\hfil FNY,\hfil FVb,\hfil SWd,\hfil MRd,\hfil UXa,\hfil Jbd,\hfil LTc,\hfil LZb,\hfil WYa,\hfil HRc.\hfill}

\item[\hbox to 40pt{\3(39-39-04)}] {\1123,\hfil 145,\hfil 167,\hfil 289,\hfil 2AB,\hfil 3CD,\hfil 3EF,\hfil 4GH,\hfil 4IJ,\hfil 5KL,\hfil 5MN,\hfil 6OP,\hfil 6QR,\hfil 7ST,\hfil 7UV,\hfil 8GO,\hfil 8UX,\hfil 9TZ,\hfil 9Ia,\hfil ASW,\hfil ANP,\hfil BHV,\hfil Bbc,\hfil CGS,\hfil CQa,\hfil DNX,\hfil DJb,\hfil EUY,\hfil EKZ,\hfil FIP,\hfil Fcd,\hfil KOb,\hfil MTc,\hfil LVa,\hfil MQY,\hfil HRZ,\hfil RXd,\hfil JWY,\hfil LWd.\hfill}

\item[\hbox to 40pt{\3(39-39-05)}] {\1123,\hfil 145,\hfil 167,\hfil 289,\hfil 2AB,\hfil 3CD,\hfil 3EF,\hfil 4GH,\hfil 4IJ,\hfil 5KL,\hfil 5MN,\hfil 6OP,\hfil 6QR,\hfil 7ST,\hfil 7UV,\hfil 8GO,\hfil 8XY,\hfil 9NU,\hfil 9Ra,\hfil AKW,\hfil AHQ,\hfil BMP,\hfil BVc,\hfil CGS,\hfil CNb,\hfil DVY,\hfil DIa,\hfil EKO,\hfil EJd,\hfil FHU,\hfil FXZ,\hfil LSZ,\hfil IPZ,\hfil QYd,\hfil WXb,\hfil MTd,\hfil LRc,\hfil Jbc,\hfil TWa.\hfill}

\item[\hbox to 40pt{\3(39-39-06)}] {\1123,\hfil 145,\hfil 167,\hfil 289,\hfil 2AB,\hfil 3CD,\hfil 3EF,\hfil 4GH,\hfil 4IJ,\hfil 5KL,\hfil 5MN,\hfil 6OP,\hfil 6QR,\hfil 7ST,\hfil 7UV,\hfil 8GO,\hfil 8XY,\hfil 9MQ,\hfil 9Td,\hfil AKZ,\hfil AJV,\hfil BRb,\hfil BHc,\hfil COW,\hfil CKS,\hfil DNa,\hfil DVX,\hfil EQZ,\hfil EIY,\hfil FHU,\hfil FLd,\hfil GZa,\hfil NPc,\hfil JPd,\hfil MUW,\hfil IWb,\hfil SYc,\hfil Tab,\hfil LRX.\hfill}

\item[\hbox to 40pt{\3(39-39-07)}] {\1123,\hfil 145,\hfil 167,\hfil 289,\hfil 2AB,\hfil 3CD,\hfil 3EF,\hfil 4GH,\hfil 4IJ,\hfil 5KL,\hfil 5MN,\hfil 6OP,\hfil 6QR,\hfil 7ST,\hfil 7UV,\hfil 8GO,\hfil 8MU,\hfil 9IT,\hfil 9QY,\hfil AKW,\hfil AHV,\hfil BNP,\hfil BXa,\hfil CGS,\hfil CQZ,\hfil DJP,\hfil DLd,\hfil EKO,\hfil EIX,\hfil FVY,\hfil FNc,\hfil Sab,\hfil MRb,\hfil UXZ,\hfil JWb,\hfil Tcd,\hfil WZc,\hfil LYa,\hfil HRd.\hfill}

\item[\hbox to 40pt{\3(39-39-09)}] {\1123,\hfil 145,\hfil 167,\hfil 289,\hfil 2AB,\hfil 3CD,\hfil 3EF,\hfil 4GH,\hfil 4IJ,\hfil 5KL,\hfil 5MN,\hfil 6OP,\hfil 6QR,\hfil 7ST,\hfil 7UV,\hfil 8GO,\hfil 8SX,\hfil 9NV,\hfil 9bd,\hfil AUY,\hfil AHZ,\hfil BJP,\hfil BLc,\hfil CGW,\hfil CVc,\hfil DZb,\hfil DMP,\hfil EKO,\hfil EIU,\hfil FHT,\hfil FNQ,\hfil KZa,\hfil JSa,\hfil IRb,\hfil LTd,\hfil QWa,\hfil WYd,\hfil MXY,\hfil RXc.\hfill}

\item[\hbox to 40pt{\3(39-39-10)}] {\1123,\hfil 145,\hfil 167,\hfil 289,\hfil 2AB,\hfil 3CD,\hfil 3EF,\hfil 4GH,\hfil 4IJ,\hfil 5KL,\hfil 5MN,\hfil 6OP,\hfil 6QR,\hfil 7ST,\hfil 7UV,\hfil 8GW,\hfil 8OY,\hfil 9QZ,\hfil 9IU,\hfil AKS,\hfil APb,\hfil BRa,\hfil BXc,\hfil CGX,\hfil CKQ,\hfil DMb,\hfil DJT,\hfil ESW,\hfil EZc,\hfil FLa,\hfil FIP,\hfil HVb,\hfil HZd,\hfil NOd,\hfil MRW,\hfil NUX,\hfil LVY,\hfil JYc,\hfil Tad.\hfill}

\baselineskip=8pt
}
\end{description}
Two bipartite graphs give 4 MMP hypergraphs that are
not dual to themselves.
\begin{description}
{\baselineskip=8pt
\item[\hbox to 40pt{\3(39-39-01a)}] {\1123,\hfil 145,\hfil 167,\hfil 289,\hfil 2AB,\hfil 3CD,\hfil 3EF,\hfil 4GH,\hfil 4IJ,\hfil 5KL,\hfil 5MN,\hfil 6OP,\hfil 6QR,\hfil 7ST,\hfil 7UV,\hfil 8GO,\hfil 8SW,\hfil 9Rb,\hfil 9NX,\hfil AMP,\hfil AVZ,\hfil BHa,\hfil BQd,\hfil CKO,\hfil CUX,\hfil DQY,\hfil DJW,\hfil EIP,\hfil ETa,\hfil FRc,\hfil FLZ,\hfil GYZ,\hfil Kab,\hfil LSd,\hfil HUc,\hfil IXd,\hfil MWc,\hfil NTY,\hfil JVb.\hfill}

\item[\hbox to 40pt{\3(39-39-01b)}] {\1123,\hfil 145,\hfil 167,\hfil 289,\hfil 2AB,\hfil 3CD,\hfil 3EF,\hfil 4GH,\hfil 4IJ,\hfil 5KL,\hfil 5MN,\hfil 6OP,\hfil 6QR,\hfil 7ST,\hfil 7UV,\hfil 8GW,\hfil 8MZ,\hfil 9Sa,\hfil 9Rd,\hfil AOX,\hfil AVY,\hfil BKb,\hfil BJc,\hfil CGO,\hfil CKS,\hfil DNQ,\hfil DIU,\hfil EHY,\hfil ETc,\hfil FPZ,\hfil FLd,\hfil HRb,\hfil JPa,\hfil QWc,\hfil LVW,\hfil MTX,\hfil IXd,\hfil UZb,\hfil NYa.\hfill}

\item[\hbox to 40pt{\3(39-39-08a)}] {\1123,\hfil 145,\hfil 167,\hfil 289,\hfil 2AB,\hfil 3CD,\hfil 3EF,\hfil 4GH,\hfil 4IJ,\hfil 5KL,\hfil 5MN,\hfil 6OP,\hfil 6QR,\hfil 7ST,\hfil 7UV,\hfil 8GO,\hfil 8WX,\hfil 9JU,\hfil 9MR,\hfil ALY,\hfil AId,\hfil BVZ,\hfil BNP,\hfil CGS,\hfil CLQ,\hfil DVW,\hfil DMd,\hfil EKO,\hfil EIT,\hfil FHY,\hfil FUc,\hfil KZa,\hfil Sab,\hfil QXc,\hfil WYb,\hfil HRZ,\hfil NTX,\hfil JPb,\hfil acd.\hfill}

\item[\hbox to 40pt{\3(39-39-08b)}] {\1123,\hfil 145,\hfil 167,\hfil 289,\hfil 2AB,\hfil 3CD,\hfil 3EF,\hfil 4GH,\hfil 4IJ,\hfil 5KL,\hfil 5MN,\hfil 6OP,\hfil 6QR,\hfil 7ST,\hfil 7UV,\hfil 8GO,\hfil 8Ua,\hfil 9LT,\hfil 9Ic,\hfil ASW,\hfil AKP,\hfil BJR,\hfil BNb,\hfil CGS,\hfil CNc,\hfil DPY,\hfil DJa,\hfil EOX,\hfil ETb,\hfil FIV,\hfil FMQ,\hfil HQZ,\hfil HYb,\hfil KUZ,\hfil MWa,\hfil WXd,\hfil XZc,\hfil VYd,\hfil LRd.\hfill}

\baselineskip=8pt
}
\end{description}

\parindent=20pt

The above OMLs admit neither a strong set of states nor
any known property stronger than the orthomodularity itself
apart from the Mayet vector field Eq.~(\ref{eq:E-4}).
One pair (39-39-01a,b) of the duals that are not dual to
each other admit at least two
states while the other  (39-39-08a,b) admit one single
state. All OMLs that are dual to themselves
(39-39-00,-02--07,-09--10) admit exactly one state (1/3
for each atom).

Bipartite graphs with 80 vertices that give 40-40 OMLs are
much more numerous than those with 78 vertices above.
There are 174 such graphs and they give 80 OMLs that are
dual to themselves. Among them there is only one (40-40-038)
that admits more than one state. Among the others
(94 graphs) there are eight OMLs that admit more than one state
(40-40-043a,b, -097a,b, -111a,b, -130a,b).

There are 2515 bipartite graphs with 82 vertices that give
4612 41-41 OMLs. 418 of the MMP hypergraphs are dual to
themselves and the other 4194 are not. The latter graphs form
2097 pairs of duals that are not dual to themselves. Of the
former ones, 10 admit two or more states (all the others
admit only one single state) and of the latter ones, 78
dual pairs admit two or more states and the remaining
2019 pairs admit only one state. We can recognize that the
more vertices we have the smaller is the portion of lattices
dual to themselves.

The biggest loops of 39-39 are enneadecagons (19-gons)
and of 40-40 and 41-41 icosagons (20-gons)\footnote{The
biggest loop of the OMLs that admit only one state are in
average neither significantly smaller nor bigger than
those that admit two or more states. Therefore neither
of these two properties impose significant restrictive
conditions on the OMLs.}  which makes them
inappropriate for the standard graphical presentation---there
are too many lines over each other in their
figures to discern patterns. Therefore and because of the
new feature of the existence of three separate cycles for
3D OMLs with equal number of vertices (atoms) and edges
(blocks) we present details of our separate cycle
representation and give several figures in the next section.

\section{\label{sec:figs}Separate Level Representation of the
MMP Hypergraphs}

As already mentioned in section~\ref{sec:one-state-p}, our new
layout of MMP hypergraphs is inspired by the presentation of
the 36-36 one given in Ref.~\onlinecite{bdm-ndm-mp-1}
and repeated here as the first figure in Fig.~\ref{fig:36-36}.
Our goal is to simplify graphical
representation of big MMP hypergraphs and big
arbitrary hypergraphs with the same number of atoms
and blocks, i.e., vertices and edges, respectively.

In the latter figure one can notice 9 radially placed blocks which
do not have common atoms (and which therefore include 27 atoms), while
9 remaining atoms form an inner ring.
We call radial blocks {\em independent blocks} and remaining atoms
{\em free atoms}. The outermost atom of each independent block is
connected to the outermost atoms of two other independent blocks by two
blocks, middle atoms of which are free atoms. These connecting
blocks form a cycle shown separately in the second figure of
Fig.~\ref{fig:36-36} (as opposed to the original layout, where
oppositely placed independent blocks are connected, we connect
adjacent blocks). Similarly, middle atoms of independent blocks are
connected by blocks with free atoms as their middle atoms and there
is again a cycle of connecting blocks (shown separately in the third
figure of Fig.~\ref{fig:36-36}).
Finally, innermost atoms of independent blocks are also connected with
blocks that contain one free atom. In the original layout free atoms
are ``last'' atoms of connecting blocks, but as the atoms in a
block can be freely permuted, we can again form a cycle, shown here as
the fourth figure of Fig.~\ref{fig:36-36}.

Based on described analysis of the layout of the 36-36 MMP,
we break the representation of MMPs with equal
number of atoms and blocks into three separate levels.

The first step is to identify sets of {\em independent blocks}, i.e.,
those that meet two criteria: they do not share atoms and no three such
blocks are connected by a single block. In the archetype case of
the 36-36 MMP all {\em connecting blocks} (blocks that
connect independent ones) contain one free atom.
When all sets of independent blocks are found, we extract the largest
ones.

In the second step, for each such set we try to identify all
cycles that visit all independent blocks in the set.
Here we do not use the term ``cycle'' in the sense of graph theory---our
cycle is a sequence of blocks that connect atoms of two independent
blocks and pass through a free atom (if required, atoms of connecting
blocks are permuted so that the free atom becomes the middle one).
The shortest cycle forms the first level of our presentation.
Independent blocks and free atoms are arranged in the sequence in
which they are visited. But, as compared to the archetypal 36-36
MMP, there are some differences: (1) a cycle is usually not closed,
that is, it does not finish in the same atom in which it starts (as
can be seen on the uppermost blocks in the second figure of
Fig.~\ref{fig:35-35} and first ones of Figs.~\ref{fig:39-39} and
\ref{fig:40-40-34}), although sometimes it does (first figure in
Fig.~\ref{fig:40-40-38}); (2) in most cases some independent blocks
are visited two or even three times (Figs.~\ref{fig:35-35},
 \ref{fig:39-39}, and \ref{fig:40-40-38}); (3) in most
(maybe even all) cases some free atoms are visited more than once
and, of course, there are free
atoms that are not visited at all (all examples).
(Figs.~\ref{fig:35-35}, \ref{fig:39-39}, \ref{fig:40-40-34}
and \ref{fig:40-40-38}).

\begin{figure*}[htp]
\begin{center}
\includegraphics[width=0.3\textwidth]{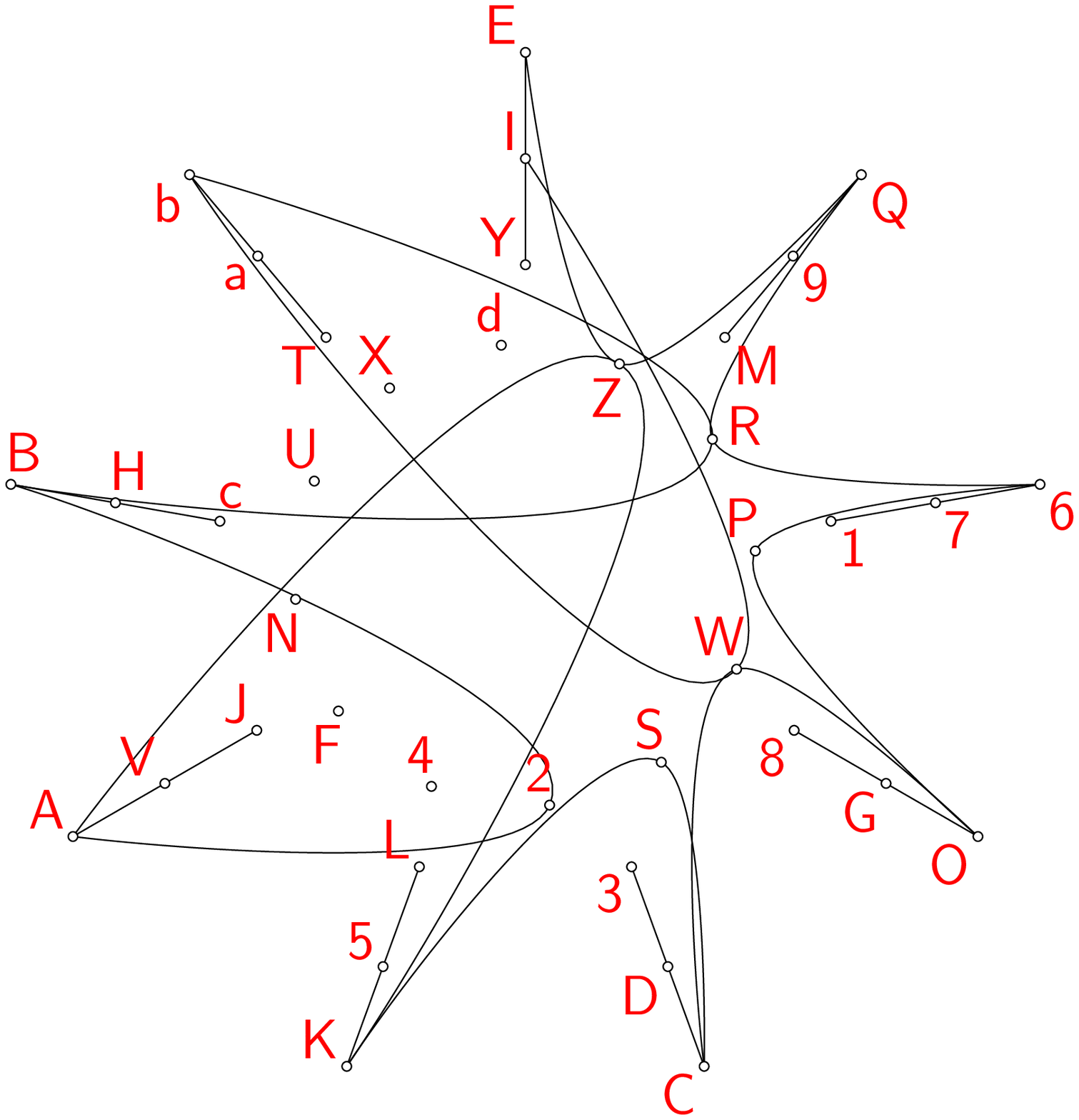}
\includegraphics[width=0.3\textwidth]{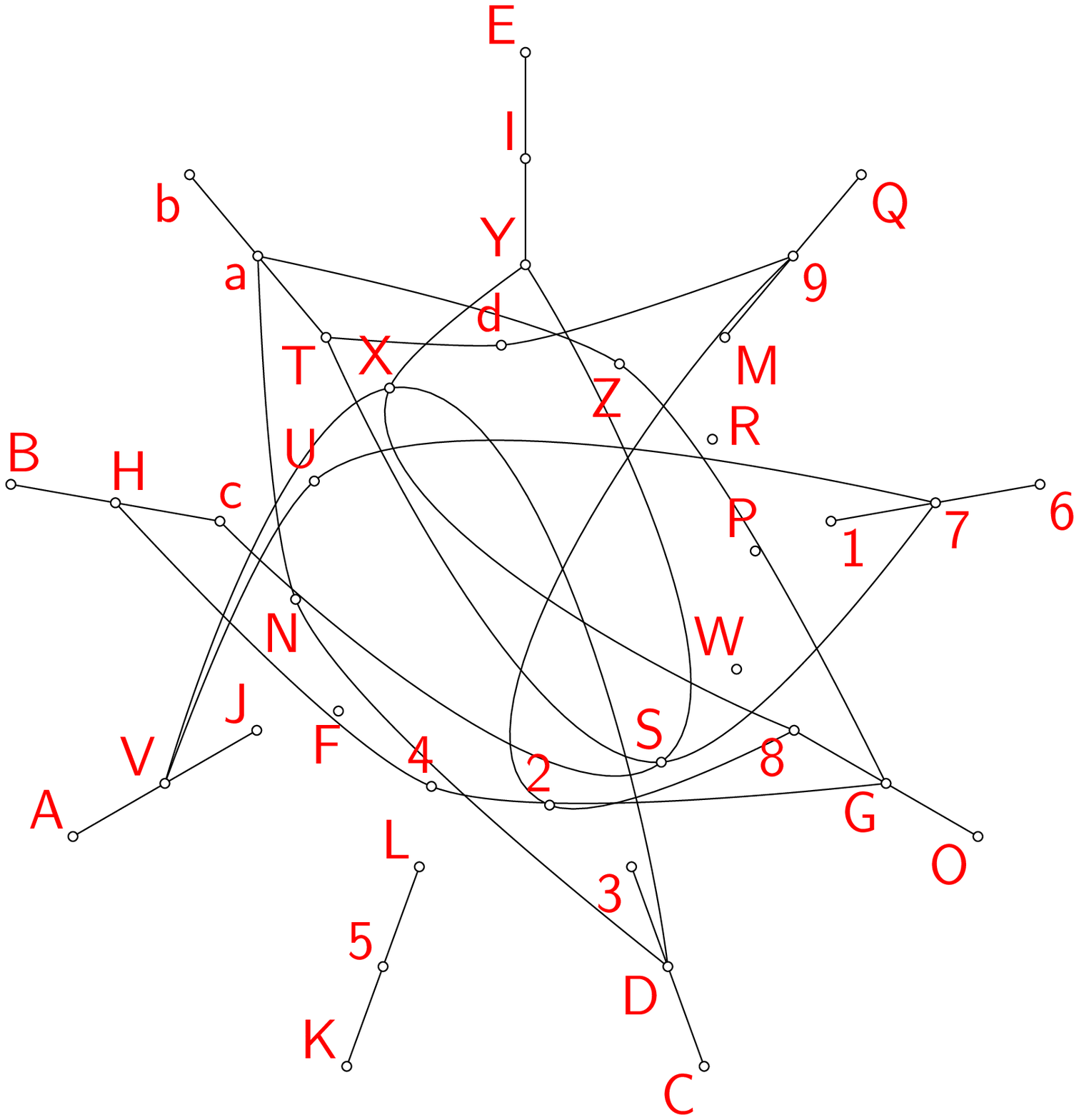}
\includegraphics[width=0.3\textwidth]{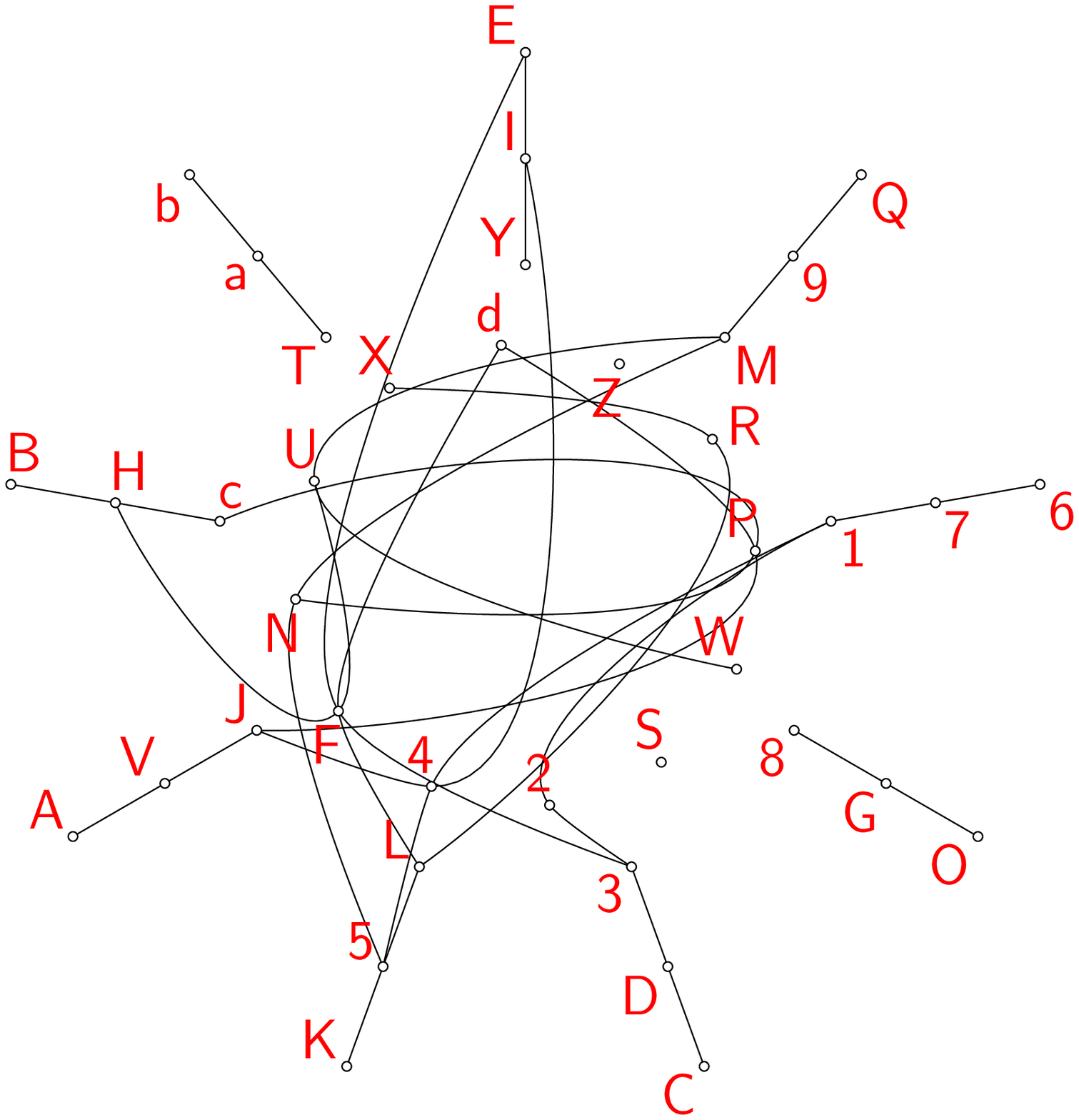}
\end{center}
\caption{39-39-06 OML dual to itself}
\label{fig:39-39}
\end{figure*}

\begin{figure*}[htp]
\begin{center}
\includegraphics[width=0.3\textwidth]{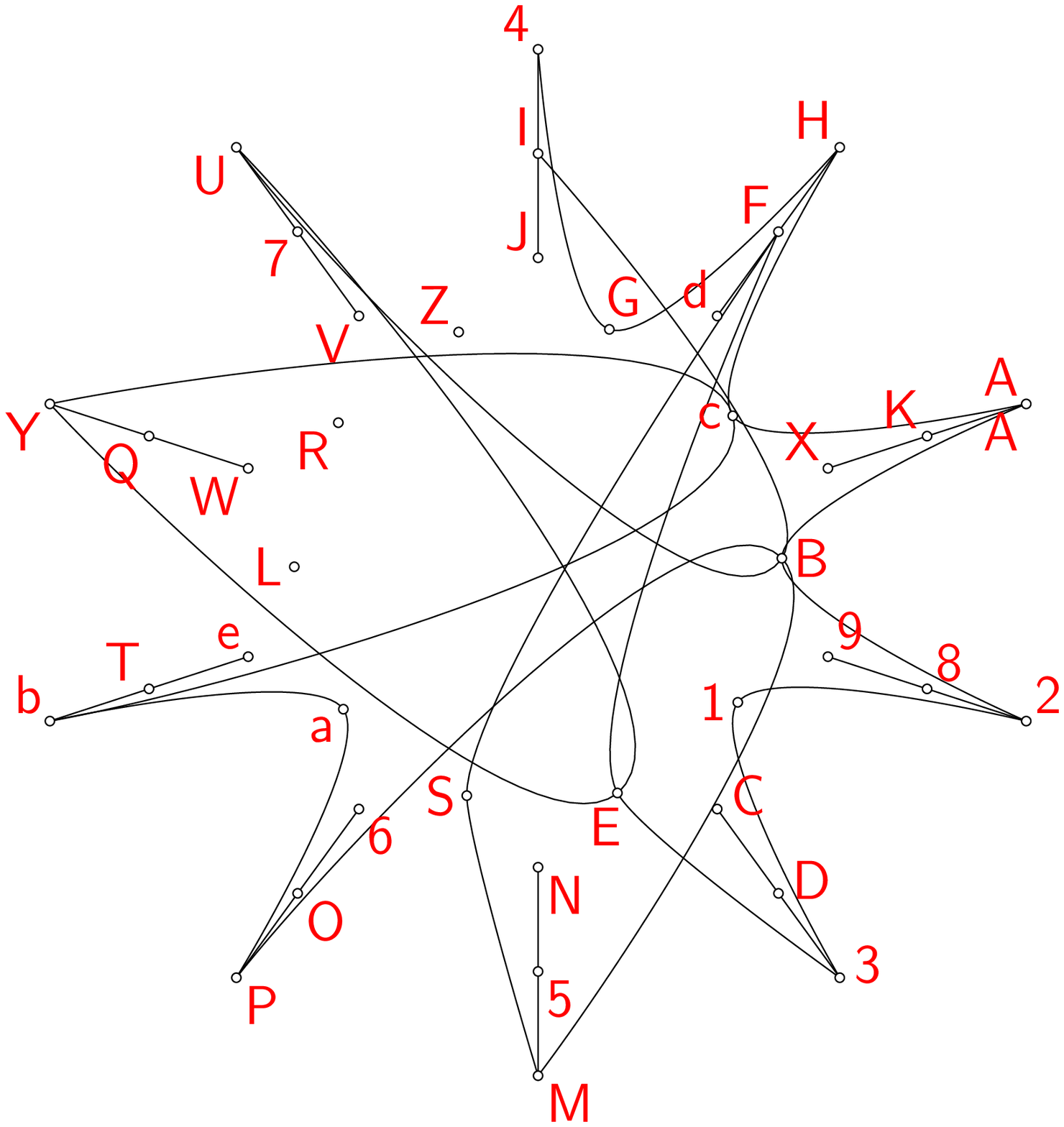}
\includegraphics[width=0.3\textwidth]{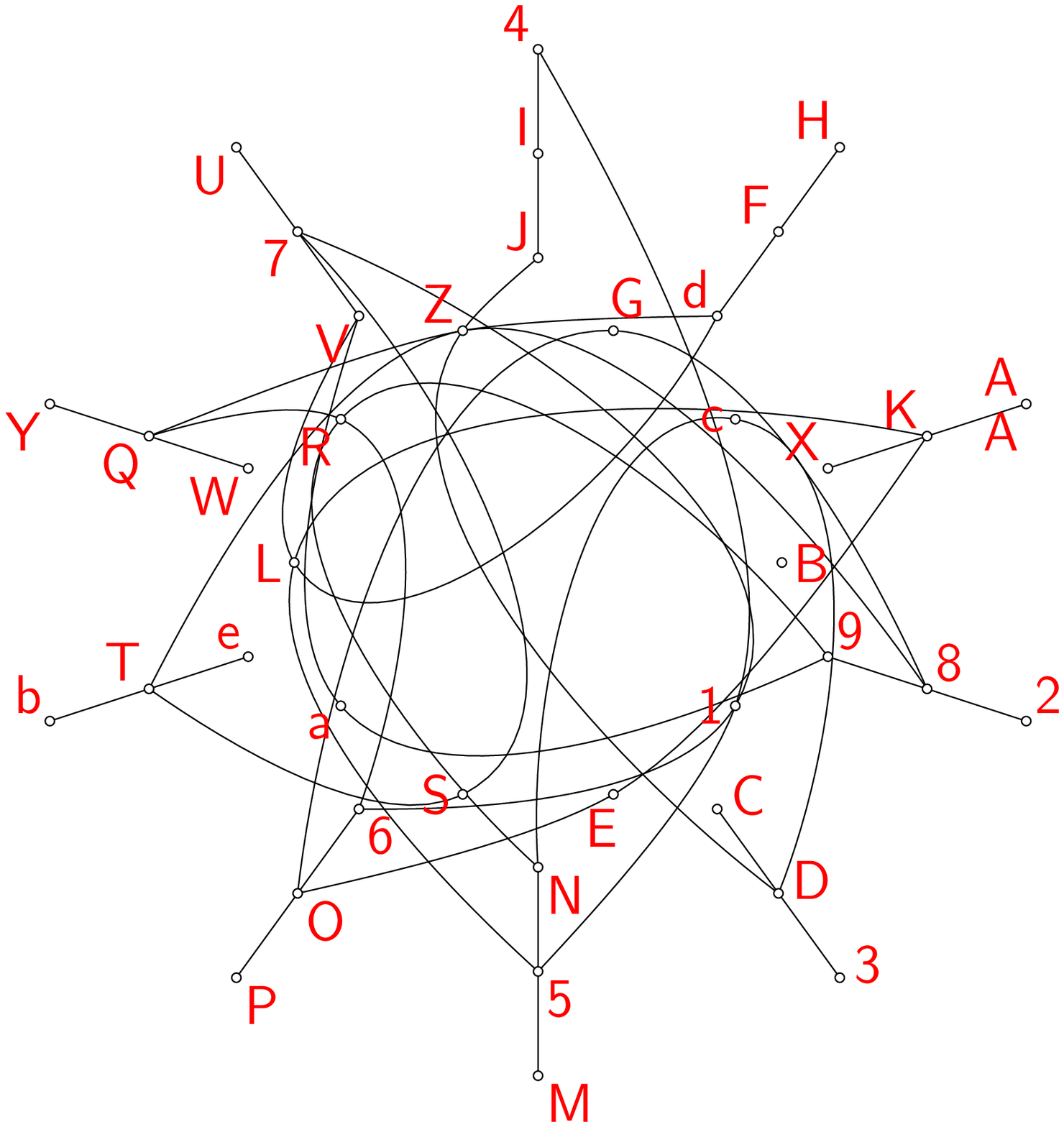}
\includegraphics[width=0.3\textwidth]{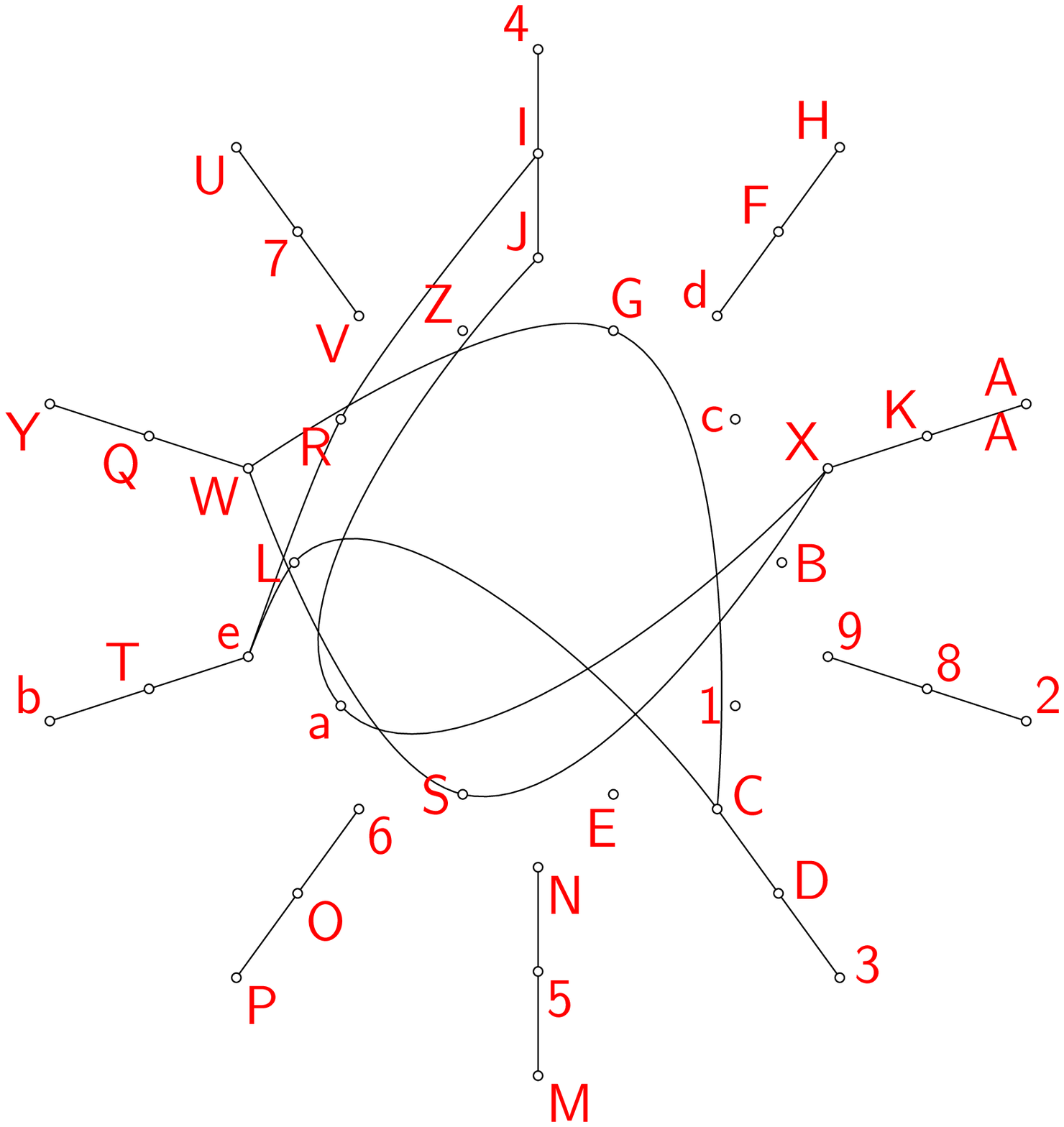}
\end{center}
\caption{40-40-34 OML dual to itself}
\label{fig:40-40-34}
\end{figure*}

\begin{figure*}[htp]
\begin{center}
\includegraphics[width=0.3\textwidth]{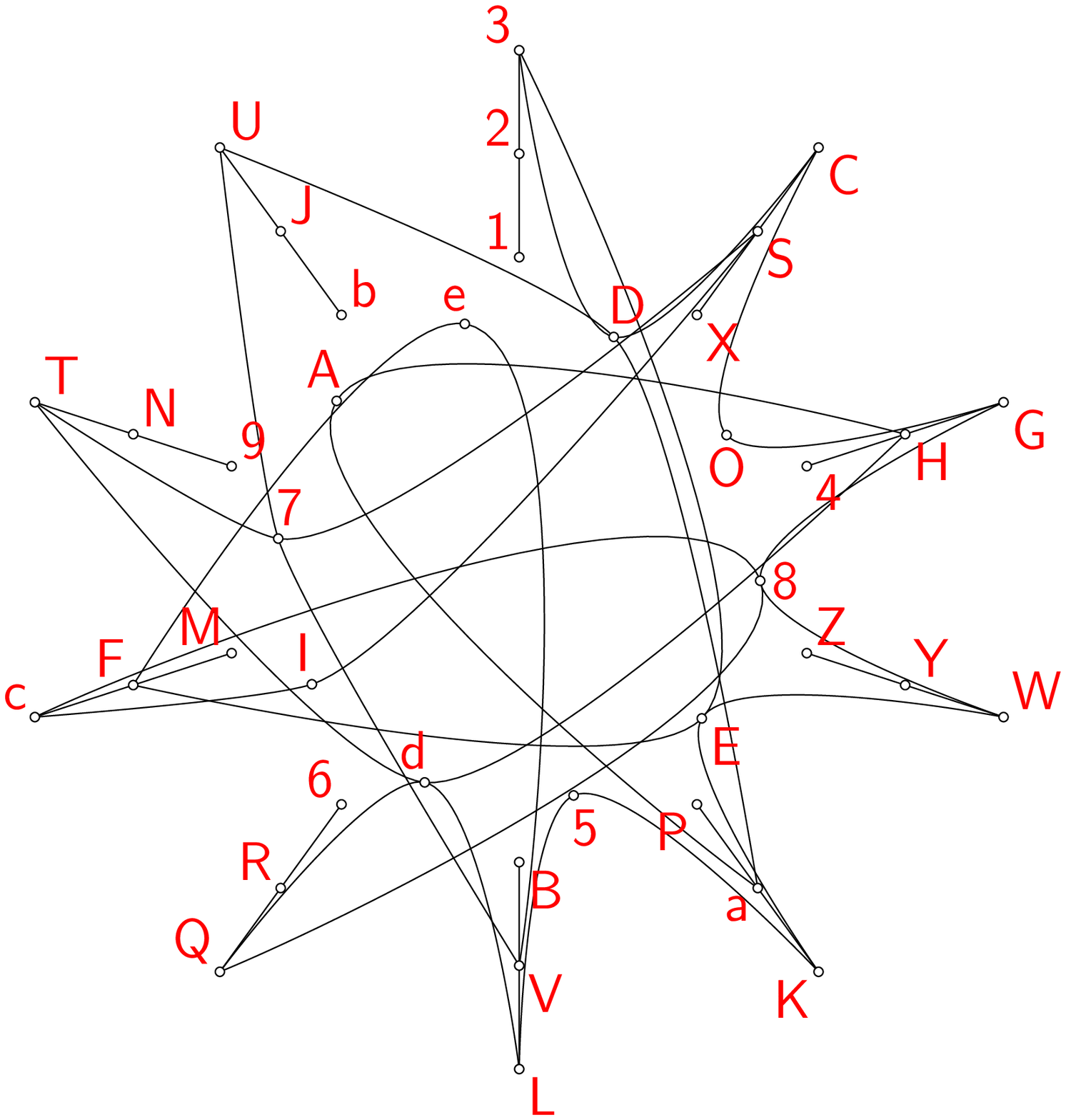}
\includegraphics[width=0.3\textwidth]{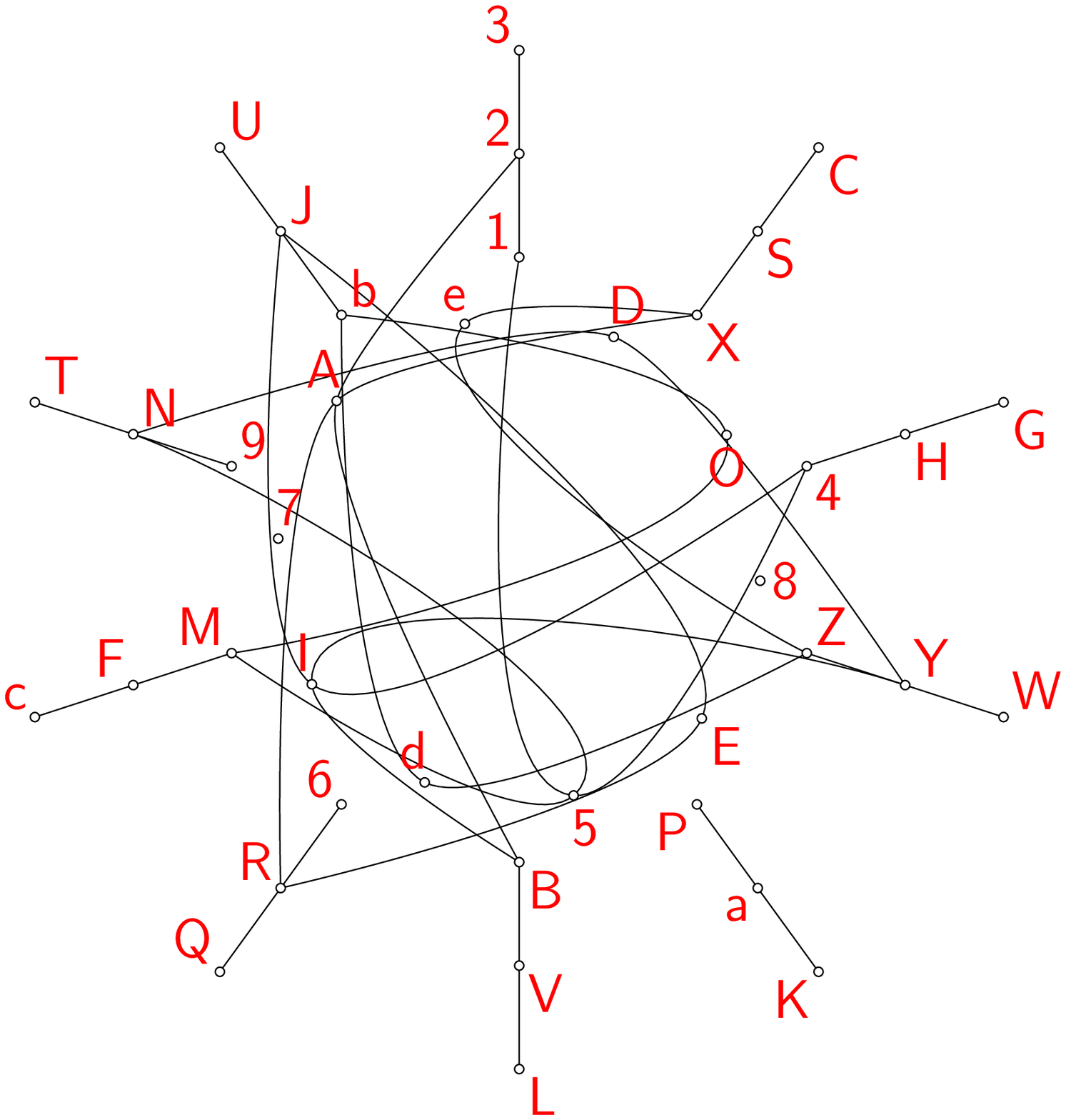}
\includegraphics[width=0.3\textwidth]{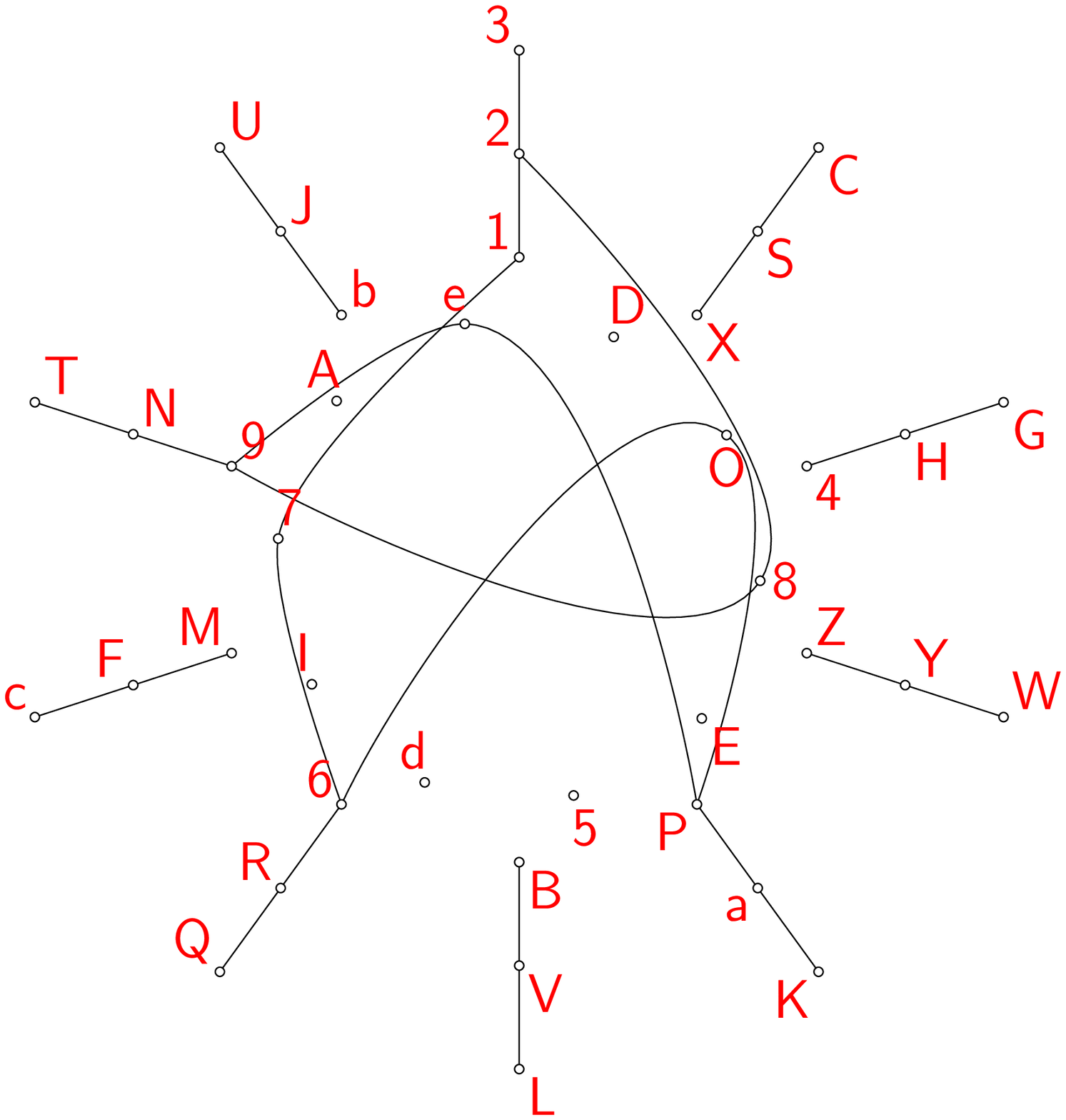}
\end{center}
\caption{40-40-38 OML dual to itself}
\label{fig:40-40-38}
\end{figure*}

If required---and if possible---atoms of independent blocks are
permuted so that the visited atom becomes the first/outermost atom
(if the block is visited twice, then visited atom is placed in the
middle).

In most cases we can find a second cycle that begins and ends on the
same independent block, but not in the same atom; besides, these
cycles usually do not visit all independent blocks.
This can be seen in all our examples: the third figure of
Fig.~\ref{fig:35-35} and the second figures
of Figs.~\ref{fig:39-39}, \ref{fig:40-40-34} and \ref{fig:40-40-38}.
Such cycles form the second level of our layout (again, if required and
if possible, atoms of independent blocks are permuted so that
connecting blocks visit their second/middle atoms).

The remaining blocks are contained in the third level. In some cases they
again form a cycle: the third figures of Figs.~\ref{fig:40-40-34} and
\ref{fig:40-40-38} (in fact, in these two examples the second and third
cycles can be regarded as a single cycle, but we broke that cycle when
the independent block in which it began was visited for the second
time). But usually the third level contains two or even more unconnected
sequences of blocks. Namely, some blocks connect one atom of some
independent block and two free atoms, that is, there are some blocks
that do not connect two independent blocks: the fourth figure of
Fig.~\ref{fig:35-35} and the third figure of Fig.~\ref{fig:39-39}.

The previously described parts of our algorithm are implemented in
the C{\small$++$} programming language using the Boost Graph
library.~\cite{bgl}
The program for the final graph layout (including the calculation of the
atoms' coordinates and drawing of the graph) is written in
the {\tt Asymptote}~\cite{asympt}
vector graphics language based on Donald Knuth's
{\footnotesize\sf METAFONT}.

\section{\label{sec:conclusions}Conclusions}

In this paper, we found a correct way to establish a correlation
between a lattice description and a Hilbert space description
of quantum systems as well as their preparation, handling, and
measurement. Our description also allows for a straightforward
reconstruction of the quantum formalism from empirically
justified axioms. In Sec.~\ref{sec:represent}  we explain how
this can be done and why the previous descriptions from
the literature were wrong. Essentially they were wrong
because they were based on Greechie diagrams that
handle only orthogonalities between Hilbert space
subspaces and have no way to describe conditions and
equations that have to be satisfied in any Hilbert space
or any Hilbert lattice quantum formalism and that involve
detailed relations between non-orthogonal subspaces.

We describe several families of equations and other
conditions that must hold in every Hilbert lattice in
Sec.~\ref{sec:ortho}.  We made use of correspondences
between graphs and lattices, which in turn correspond to
Hilbert space subspaces, in order to visualize and study
3-dim quantum setups in Sections
\ref{sec:represent}--\ref{sec:figs}.  In particular, we found and
investigated MMP hypergraphs (see Def.~\ref{def:mmp}) with 
equal numbers of vertices and edges, which correspond to bipartite
graphs (Sec.~\ref{sec:one-state}).  Separately, we studied
Greechie diagrams used in the literature to represent
Kochen-Specker and other quantum setups
(Secs.~\ref{sec:one-state} and \ref{sec:no-god}) to
see which Hilbert lattice properties hold and which do
not hold in them.

In Sec.~\ref{sec:represent} we developed a new graphical
representation of the known KS setups by means
of MMP hypergraphs (see Figs.~\ref{fig:ks-117}, \ref{fig:peres},
\ref{fig:bub-proof}, and \ref{fig:conw-k}) to visualize their
properties. Then, using our algorithms and programs,
we showed, in particular in Eq.~(\ref{eq:bubs-proof}) and
Fig.~\ref{fig:bub-proof}, that Greechie diagrams cannot
represent KS setups because they are not subalgebras of a
Hilbert lattice. This is obvious from the fact that in
a Greechie diagram, the join of nonorthogonal atoms
(lines) $a$ and $q$ (in Fig.~\ref{fig:bub-proof}) is the
whole space (1), while in a Hilbert space, it is a plane
$a+q$. Therefore, if we wanted to have a lattice
representation of KS setups, we should add lattice
elements missing in Greechie diagrams as shown in
Fig.~\ref{fig:bub-proof}. This provides us with a new type 
of lattices (MMPL) that include all relations---between 
both orthogonal and non-orthogonal elements---needed for a 
lattice description of a considered quantum system. We 
define MMPL by Def.~\ref{def:mmpl}.  However, a detailed
elaboration of such a representation is outside of the
scope of the present paper.

Application of such an approach is in any case computationally
unfeasible for the time being, and therefore we consider
non-quantum setups to narrow down classes of lattices that
we can use to obtain complex setups in the future and
in particular KS setups.

The Kochen-Specker theorem claims that there are quantum
experimental setups that cannot be given a classical
rendering. Its proof was based on setups (KS setups)
that were considered quantum and to which it was
impossible to ascribe classical 0-1 values.
A number of authors have represented KS setups or indeed
any spin-1 experimental setup by means of Greechie
lattices.\cite{shimony,hultgren,svozil-tkadlec,svozil-book-ql,tka98,tkadlec,tkadlec-2,smith03,foulis99}
However, in Sec.~\ref{sec:represent} we proved that no
known 3-dimensional KS setup represented by Greechie/Hasse
diagrams, in particular, Kochen-Specker's, Bub's, 
Conway-Kochen's, and Peres' pass the equations that
hold in every Hilbert space. These KS setups themselves do, of course,
pass these equations in the Hilbert space itself.

A Hilbert space des\-cription of such systems
is ortho-isomorphic to a Hilbert lattice (Th.~\ref{th:repr}).
\cite{beltr-cass-book}  An OML
equipped with additional properties described in
Sec.~\ref{sec:ortho} such as admitting
strong sets of states and Mayet vector states, atomicity,
the superposition principle, the orthoarguesian property, etc.,
is easier to handle in the lattice theory than in the original
Hilbert space. This is because, e.g., Peres' KS design,
shown in Fig.~\ref{fig:peres}, has 40 triples of mutually
orthogonal vectors. The majority of the vectors are orthogonal
to vectors from other triples and rotated at various
angles in space with respect to every other.
We would have to extract this vector edifice from the
Schr\"odinger equations describing the deflections
of a spin-1 system in electric and magnetic fields.
Lattices, as opposed to such a standard Hilbert space
approach, are easier to handle, but even they
are too demanding at present.

Therefore it is viable to approach the problem from the other
end, to see whether we can generate lattices that would admit
neither quantum nor classical interpretation from the very
start (see Subsec.~ref{subsec:semi}). Such finding of 
properties and lattices that are not sufficient for a full i
description of a quantum system (e.g., the aforementioned 
description by means of Greechie diagrams) is likely to enable 
us to achieve, eventually, a complete lattice description (with
superposition included) of quantum experiments.

Here we stress that the superposition we refer to above
and in Theorem \ref{th:rksa} and Corollary \ref{th:rks}
is a  superposition of vectors contained in 1-dim Hilbert
space subspaces. As opposed to this, when we look at all
possible superpositions of two vectors they span a plane
in a 3-dim Hilbert space. That is trivial in
the sense that for some definite constants we can always
find a value that a superposition of two vectors has in
particular direction, but is nontrivial in the sense
that for bigger lattices we can find a superposition
for vectors for which only mutual orthogonalities are
known.

Another reason for a ``semi-quantum approach''
is that there exist several methods of finding and
generating new properties and equations in the theory of
OMLs and Hilbert lattices based on the lattices that do
not admit some states or other properties. The most relevant
here is a method of generating the Mayet-Godowski
equations (Def.~\ref{def:gge}) using lattices
that do not admit strong sets of
states.\cite{pm-ql-l-hql2,mp-alg-hilb-eq-09}
Based on all that together with several previous results
based on lattices admitting only one state,
\cite{shultz74,ptak87,weber94,navara08} in Sec.\
\ref{sec:ortho} we formulated the following theorem:

{\parindent=0pt
{\bf Theorem \ref{th:rksa}} [Semi-quantum lattice algorithms]
{\em There exist algorithms that generate finite
sequences of OMLs that admit superposition,
real-valued states, and a vector state given by Eq.~(\ref{eq:E-4})
but do not admit other conditions that have to be satisfied
by every Hilbert lattice, in particular equations like orthoarguesian
and Godowski ones (Subsection \ref{subsec:eqns}).
As a consequence of violating Godowski
equations, these OMLs do not admit strong sets of states. }}

Such a choice is determined by our recent finding that
OMLs with equal number of atoms and blocks
possess and lack properties stated in the theorem. They
all satisfy the superposition principle and therefore
do not admit classical interpretation, they all admit
real-valued states,  and they all admit a vector state which,
when applied to Hilbert lattices, select those over
which a field (real, for the time being) can be defined.
They admit neither strong sets of states nor orthoarguesian
properties, and this makes them non-quantum but suitable for
generation of quantum properties such as Mayet-Godowski
equations.\cite{mp-alg-hilb-eq-09}
We generate them by means of novel algorithms which first
generate bipartite graphs (Sec.~\ref{sec:one-state}) and
then convert them into hypergraphs that correspond to
OMLs with equal numbers of atoms and blocks as
described in Sec.~\ref{sec:one-state}.
These results substantiate the following corollary of
Theorem \ref{th:rksa}:

{\parindent=0pt
{\bf Corollary \ref{th:rks}} [Semi-quantum lattices]
{\em There exists a class of {\rm OML\/}s that admit superposition,
real-valued states, and a vector state but do not admit other
conditions that have to be satisfied by every Hilbert lattice.}}

To verify these and find new properties of lattices
with equal number of atoms and blocks
we had to generate a significant number of them. Towards that
goal we developed several algorithms for generating and
verifying properties on them as well as for their graphical
representations, in Secs.~\ref{sec:one-state},
\ref{sec:one-state-p}, and \ref{sec:figs}, respectively.

The generation was performed by representing lattices as
graphs then applying an extended algorithm that exhaustively
determines all the associated graphs.

As a final note, we point out that in Sec.~\ref{sec:represent}
(antepenultimate paragraph) we obtained an important ``by-product'' in the
field of Hilbert lattice equations while we were checking whether
$n$OA equations (\ref{eq:noa}) pass Peres' OML
that corresponds to Peres' MMP hypergraph shown in Fig.\
\ref{fig:peres}. In Ref.~\onlinecite{mpoa99}, we found the
new infinite class of generalized orthoarguesian equations of
Theorem~\ref{th:noa}, but at the time the
computing power of available clusters were only sufficient to
find lattices in which the equations up to 4OA would pass and
a 5OA fail. In Ref.~\onlinecite{pm-ql-l-hql2} we generated lattices in
which 6OA failed and OAs up to 5OA passed. Such examples
are important because they prove that the equations form a
successively stronger sequence at least up to those orders.
In Ref.~\onlinecite{mpoa99}, we proved
that all individual orthoarguesian equations previously found
(by other authors) were equivalent to either 3OA or 4OA.
When we found our $n$OA, it was unknown whether the same
might occur with $n$OA at the 6OA level.\cite{mayet06-hql2}
Our result (the aforementioned passing of 3OA through 6OA and failure
of 7OA in Peres' lattice) dispels any doubt. It was serendipitous
that we obtained this result in this way, because no present-day
supercomputer is capable of generating 7OA examples by brute
force---at least not with our present algorithms.

\begin{acknowledgments}
One of us (M. P.) would like to thank his host Hossein Sadeghpour
for a support during his stay at ITAMP.

Supported by the US National Science Foundation through a
grant for the Institute for Theoretical Atomic, Molecular,
and Optical Physics (ITAMP) at Harvard University and Smithsonian
Astrophysical Observatory and Ministry of
Science, Education, and Sport of Croatia through the project
No.~082-0982562-3160.

Computational support was provided by the cluster Isabella of
the Zagreb University Computing Centre and
by the Croatian National Grid Infrastructure.
\end{acknowledgments}

\end{document}